\documentclass[journal,peerreview,nodraftcls]{IEEEtran}
\usepackage{graphicx,subfig,amsthm,amsmath,latexsym,amssymb,times}
\usepackage{epsfig,multirow,rotating,times,verbatim,wrapfig,cite,booktabs,float}
\usepackage{color,xr,array,hyperref, algpseudocode}

\usepackage{array}

\newtheorem{theorem}{Theorem}

\newtheorem{corollary}[theorem]{Corollary}

\newcommand{\R}{\mathbb{R}}

\newcommand{\mP}{\mathbb{P}}

\newcommand{\abs}[1]{|#1|}

\newcommand{\bLambda}{\boldsymbol{\Lambda}}

\newcommand{\bbeta}{\boldsymbol{\beta}}

\newcommand{\bphi}{\boldsymbol{\phi}}
\newcommand{\bPhi}{\boldsymbol{\Phi}}

\newcommand{\bGamma}{\boldsymbol{\Gamma}}

\newcommand{\bepsilon}{\boldsymbol{\epsilon}}

\newcommand{\bmu}{\boldsymbol{\mu}}
\newcommand{\bzeta}{\boldsymbol{\zeta}}

\newcommand{\bC}{\boldsymbol{C}}

\newcommand{\bH}{\boldsymbol{H}}
\newcommand{\bh}{\boldsymbol{h}}

\newcommand{\bI}{\boldsymbol{I}}
\newcommand{\bJ}{\boldsymbol{J}}

\newcommand{\bS}{\boldsymbol{S}}

\newcommand{\bR}{\boldsymbol{R}}
\newcommand{\bW}{\boldsymbol{W}}

\newcommand{\bU}{\boldsymbol{U}}

\newcommand{\bX}{\boldsymbol{X}}

\newcommand{\bY}{\boldsymbol{Y}}
\newcommand{\bZ}{\boldsymbol{Z}}

\newcommand{\bzero}{\boldsymbol{0}}

\newcommand{\bone}{\boldsymbol{1}}

\newcommand{\mN}{\mathcal N}

\newcommand{\mG}{\mathcal G}
\newcommand{\mW}{\mathcal W}
\newcommand{\mJ}{\mathcal J}

\DeclareMathOperator*{\argmin}{argmin}
\newcommand{\argmax}{\operatornamewithlimits{argmax}}

\newcommand{\bL}{\boldsymbol{L}}

\newcommand{\bit}{\begin{itemize}}
\newcommand{\eit}{\end{itemize}}
\newcommand{\ben}{\begin{enumerate}}
\newcommand{\een}{\end{enumerate}}
\newcommand{\beqn}{\begin{equation}}
\newcommand{\eeqn}{\end{equation}}
\newcommand{\bea}{\begin{eqnarray*}}
\newcommand{\eea}{\end{eqnarray*}}
\newcommand{\bpf}{\begin{proof}}
\newcommand{\epf}{\end{proof}\ms}
\newcommand{\ms}{\medskip}

\newcommand{\citep}{\cite}
\newcommand{\citet}{\cite}
\newcommand{\citeyear}{\cite}
\newcommand{\citeauthor}{\cite}

\begin{document}
\title{FAST Adaptive Smoothing and Thresholding for Improved
  Activation Detection in Low-Signal fMRI}
\author{Israel~Almod\'ovar-Rivera~and~Ranjan~Maitra

  \thanks{I. Almod\'ovar-Rivera is with the Department of
    Biostatistics and Epidemiology at the University of Puerto Rico,
    Medical Science Campus, San Juan, Puerto Rico, USA.}
  \thanks{R.Maitra is with the Department of Statistics, Iowa State
    University, Ames, Iowa, USA.}
   \thanks{An earlier version of this article won
     I. Almod\'ovar-Rivera a  Student Paper Competition
     award sponsored by the  American Statistical Association Section
     on Medical Devices and Diagnostics at the 2018 Joint Statistical Meetings.}
   \thanks{This research was supported in part by the
    National Institute of Biomedical Imaging and Bioengineering (NIBIB) of the National
Institutes of Health (NIH) under its Award No. R21EB016212,
I. Almod\'ovar-Rivera also acknowledges receipt of a fellowship from 
Iowa State University's Alliance for Graduate Education and the
Professoriate (AGEP) program for underrepresented graduate students
in STEM fields. The content of this paper however is solely the responsibility of the 
authors and does not represent the official views of either the
NIBIB or the NIH.} 

\thanks{\copyright 201? IEEE. Personal use of this material is
  permitted. However, permission to use this material for any other
  purposes must be obtained from the IEEE by sending a request to
  pubs-permissions@ieee.org.}
\thanks{Manuscript received xxxx xx,201x; revised xxxxxxxx xx,
  201x. Accepted xxxxxxxx xx, 201x.
  First published xxxxxxxx x, xxxx, current version published
  yyyyyyyy y, yyyy}
\thanks{Color versions of one or more of the figures in this paper are
  available online at http://ieeexplore.org.} 
\thanks{Digital Object Identifier}
}

\markboth{IEEE Transactions on Medical Imaging,~Vol.~?, No.~?, February~201?}%
{Almod\'ovar and Maitra \MakeLowercase{\textit{et al.}}: A new approach to determine activation detection based on smoothing}

\maketitle

\begin{abstract}
Functional Magnetic Resonance Imaging is a noninvasive tool for
studying cerebral function. Many factors challenge activation
detection, especially in low-signal scenarios
that arise in the performance of high-level cognitive tasks. We provide a fully
automated fast adaptive smoothing and thresholding (FAST) 
algorithm that uses smoothing and extreme value theory on correlated
statistical parametric maps for thresholding. Performance on experiments spanning a range of low-signal settings is very encouraging. The methodology also performs well in a study to identify the cerebral regions that perceive only-auditory-reliable or only-visual-reliable speech stimuli.
\end{abstract}



\section{Introduction}\label{sintro}

\IEEEPARstart{F}{unctional} Magnetic Resonance Imaging (fMRI)
\citep{belliveauetal91,kwongetal92,bandettinietal93,fristonetal95,howsemanandbowtell98,pennyetal06,lindquist08,lazar08,ashby11}
studies the spatial characteristics and extent of brain
function while at rest or, more commonly, while performing tasks or 
responding to external stimuli. The latter scenario is the
setting for this paper. Here, the imaging modality acquires voxel-wise Blood-Oxygen-Level-Dependent~(BOLD)  
measurements~\citep{ogawaetal90a,ogawaetal90b} at rest and during
stimulation or performance of a task. After pre-processing, a general
linear or other statistical model~\citep{fristonetal95,worsleyetal02} is fit to
the time course sequence against the expected BOLD 
response~\citep{fristonetal98,glover99,buxtonetal04}. Statistical
Parametric Mapping~(SPM)~\citep{fristonetal90} techniques provide
voxel-wise test statistics summarizing the association between the time
series response at each voxel and the expected BOLD  response
\citep{bandettinietal93}.  
The map of test statistics is then thresholded to identify  significantly 
activated voxels~\citep{friston1994statistical,worsley1995analysis,genoveseetal02}.   The analysis of fMRI datasets is
challenged~\citep{hajnaletal94,biswaletal96,woodetal98,gullapallietal05}
by factors such as scanner, inter- and intra-subject variability,
voluntary/involuntary or stimulus-correlated motion and also the
several-seconds delay in the BOLD response as the  neural
stimulus passes through the hemodynamic
filter~\citep{maitraetal02,gullapallietal05,maitra09b}. Pre-processing 
~\citep{woodetal98,saadetal09} mitigates some of these effects,
but additional challenges are presented by the fact that an fMRI study
is expected to have no more than 1-3\% activated voxels~\citep{chenandsmall07,lazar08}. Also, many activation 
studies involving  high-level cognitive processes have low
contrast-to-noise ratios (CNR), throwing up additional challenges as illustrated next.

\subsection{Activation Detection during Noisy Audiovisual Speech}
\label{intro:av} 
The most important 
form of human communication is
speech~\citep{kryter94,hauser96,dupontandluettin00}, 
which the brain is  adept at understanding even in noisy
surroundings. This ability may be due~\citep{nathandbeauchamp11}
to the brain's capacity  for  
multisensory integration of independently-acquired visual and auditory
input information which reduces noise and allows for more accurate
perception~\citep{sumbyandpollack54,steinandmeredith93}. Recently,
\citet{nathandbeauchamp11} studied the role of the superior 
temporal sulcus (STS) in  perceiving noisy speech, through fMRI
and behavioral experiments, and established increased connectivity
between the STS and the auditory or the visual cortex depending on
whichever modality was more reliable, that is,
less noisy.

\citet{nathandbeauchamp11} provided results on regions of interest
(ROIs) drawn on the STS and the 
auditory and visual cortices. However, the full benefit of  fMRI can
be realized only if we move beyond assessing
cerebral function at the ROI level to understanding it at the voxel
level.  Reliable voxel-wise activation detection in individual
subjects may increase the adoption 
of fMRI in a  clinical setting. All these are potential scenarios with
low CNRs where accurate activation detection methods are needed.

\subsection{Background and Current Practices}
\label{background}
Many thresholding
methods~\citep{helleretal06,benjaminiandheller07,smithandfahrmeir07,smithandnichols09,wooetal14}
in fMRI address multiple testing issues in determining significance of 
test statistics but  ignore spatial resolution.
Acquired images are instead often spatially smoothed prior to
analysis, but such  non-adaptive smoothing reduces both the adaptive spatial 
resolution and the number of 
available independent tests for activation
detection~\citep{tabelowetal06}. There are also iterative adaptive 
smoothing and segmentation methods such as 
propagation-separation (PS)~\citep{tabelowetal06} and
adaptive-segmentation (AS)~\citep{polzehletal10} that essentially 
segment the SPM into  activated and inactivated voxels. PS
approximately yields a random $t$-field and uses Random Field Theory
for segmentation while  AS uses multi-scale
testing. \citep{polzehletal10} argued for AS because of its more general 
development and fewer model assumptions. AS also requires no heuristic
corrections for spatial correlation, provides decisions at 
prescribed significance levels and showed~\citep{polzehletal10} better
performance over PS in an auditory experiment. However, 
AS requires pre-specified bandwidth sequences and ignores correlation
within the SPM. So Section!\ref{sec:methodology} of this paper
develops theory and methodology for fully automated Fast Adaptive 
Smoothing and Thresholding (FAST) algorithms that account for correlation and
obviate the need for setting all but one parameters. Performance 
evaluations on real datasets and large-scale 
simulation experiments are in Section~\ref{sec:simulation}. 
Section~\ref{sec:Apps} revisits  the dataset of 
Section~\ref{intro:av}, while Section~\ref{discussion} provides
discussion. 
A supplement with sections, figures and tables referenced
using  the prefix ``S'' is available. 


\section{Theory and Methods} \label{sec:methodology}
\subsection{Preliminaries}\label{method:prelim}
Let $\bY_i$ be the time series vector of the observed BOLD response at
the $i$th voxel obtained after preprocessing for registration and
other corrections. It is  common to relate 
$\bY_i$ to the expected BOLD response via the general linear model
\begin{equation} 
  {\bY}_{i} = \bX\bbeta_i + \bepsilon_i,
  \label{eq:lm}
\end{equation}
where $\bepsilon_i$ is a $p$th-order auto-regressive (AR) Gaussian error
vector with AR coefficients $\bphi_i 
\!\!=\!\!(\phi_{i1},\phi_{i2},\ldots,\phi_{ip})$ and marginal variance
$\sigma_i^2$. Without loss of generality (w.l.o.g.), assume that the 
design matrix $\bX$ has the intercept in the first column, the 
expected BOLD response for the $k$ stimulus levels in the next $k$
columns, and polynomial terms for the drift parameter in the remaining
$m$ columns. Therefore, $\bbeta$ is a coefficient vector of length
$d\!=\!k\!+\!m\!+\!1$. We assume that the image volume has 
$n$ voxels, so $i\!=\!1,2,\ldots,n$ in \eqref{eq:lm}.
The parameters $(\hat\bbeta_i,\hat\sigma_i^2,\hat\bphi_i)$s are
usually estimated via generalized least squares 
or restricted maximum likelihood. A typical analysis
approach then applies (voxel-wise) hypothesis tests with 
the null hypothesis specifying no activation owing to the stimulus or
task. SPMs of the form $\bGamma =
\{c'\hat{\boldsymbol{\beta}}_i\}_{i \in V}$ with
appropriate contrasts 
$c'{\boldsymbol{\beta}}_i$ are then formulated at each voxel.

Many researchers use models that assume independent or AR(1) errors,
while others pre-whiten the time series before fitting
\eqref{eq:lm} under independence. Misspecified models can yield less  
precise SPMs~\citep{monti11,luoandnichols03,lohetal08,lindquistetal09}
so here we assume AR($p$) errors, with $p$ assessed by the
Bayes Information 
Criterion (BIC)~\citep{schwarz78,shumwayandstoffer06} that trades  a
fitted model's  complexity against its fidelity to the data.  
Tests on the SPM $\bGamma$ identify voxels that are activated with the
application of the stimulus. 
Our objective is to develop an approach
that adaptively and automatically smooths and thresholds the SPM while
incorporating spatial correlation and the fact that the
sequential thresholding results in SPMs from truncated
distributions. Before detailing  our methods, we provide  some
theoretical development.   
\subsection{Theoretical Development}
\label{method:theory}
We assume $t$-distributed SPMs with degrees of
freedom large enough for them to be approximately standard normally
distributed under the hypothesis of no activation. The SPM 
has a homogeneous correlation structure, a reasonable
assumption with our use of radially 
symmetric smoothing kernels. 
We have  
\begin{theorem}
\label{theo:evt}
  Let ${\bX} \sim \mN_n(\bzero,{\bR})$ where $\bX =
(X_1,\ldots,X_n)'$ and ${\bR}$ is a circulant correlation
matrix with only nonnegative elements such that $\bR^{1/2}$ also has
no negative entries. Writing $\bone = (1,1,\ldots,1)'$, we
let $\varrho$ be the sum 
of the elements in any row of ${\bR}^{\frac12}$. Further, let
$X_{(n)}$ be the maximum value of ${\bX}$, that is, $X_{(n)} =
\max{\{X_1,X_2,\ldots,X_n\}}\equiv\max\bX$. Then  the cumulative distribution
function (CDF) $F_{(n)}(x)$ of $X_{(n)}$ is given by $F_{(n)}(x) = P(X_{(n)} \leq
x)\geq 
[\Phi(x/\varrho)]^n$, where
$\Phi(\cdot)$ is the CDF of the
standard normal random variable. The equality holds when
$\bR^{-1/2}$ also has no negative entries.
\end{theorem}

\begin{proof}
  Let $\bZ\!\sim\! \mN_n(\bzero,\bI_n)$ and $Z_{(n)} =
  \max{\{Z_1,Z_2,\ldots,Z_n\}}$ have CDF $\Phi_{(n)}(z)\equiv [\Phi(z)]^n$. 
  Then $\Phi_{(n)}(x/\varrho)  =
  \mP[Z_{(n)} \leq x/\varrho] = \mP[\bZ \leq   x \bone/\varrho]$ so that
  \begin{equation}
    \begin{split}
\Phi_{(n)}(x/\varrho) & \leq\mP[{\bR}^{1/2}\bZ \leq x {\bR}^{1/2}\bone/\varrho]\\ 
      & =   \mP[\bX \leq x \bone], \mbox{ where } \bX\sim N_n(\bzero,\bR)\\
      &    = \mP[X_{(n)}\leq x] = F_{(n)}(x).\\
      \label{th1}
  \end{split}
\end{equation}
Now $\bR^{-1/2}$ is also circulant and 
$\bR^{1/2}\bR^{-1/2}\bone = \tilde\varrho\bR^{1/2}\bone =
\tilde\varrho\varrho\bone$ where $\tilde\varrho$ is the
sum of the elements of any row of  
$\bR^{-1/2}$ and 
$\tilde\varrho=1/\varrho$. If $\bR^{-1/2}$ has no negative elements, $F_{(n)}(x)=\mP[\bX \leq x \bone]\leq \mP[{\bR}^{-1/2}\bX \leq x
{\bR}^{-1/2}\bone] =\mP[\bZ \leq x\tilde\varrho\bone] =
\Phi_{(n)}(x/\varrho) $ and then equality holds in \eqref{th1}.
\end{proof}
\begin{corollary}
  \label{theo:gumbel}
  For $\bX$ and $X_{(n)}$ as in Theorem~\ref{theo:evt}, the limiting
  distribution of $X_{(n)}$ is bounded below by one that lies in the
  domain of  attraction of   the Gumbel distribution, 
  and satisfies 
  \begin{equation}
    \lim_{n\rightarrow \infty}[F_n(a_n x + b_n)] \geq \exp\{-\exp(-x)\},
\end{equation}
where $a_n = \varrho/[ n \phi(b_n/\varrho)]$ and $b_n =\varrho
\Phi^{-1}(1-1/n)$, with $\phi(\cdot)$ the standard normal
probability density function (PDF).
\end{corollary}
\begin{proof}
Each element of $\bW\sim \mN_n(\bzero,\varrho^2\bI_n)$ has CDF
$\Phi(\frac w\varrho)$ and PDF
  $\phi(\frac w\varrho)/\varrho$. Then
  $W_{(n)}\!\doteq\!\max\bW$  has CDF $G_{(n)}(\cdot)$ that satisfies
  $\lim_{n\rightarrow     \infty}G_{(n)}(a_n x + b_n)\!=\! 
  \exp\{-\exp(-x)\}$ with $a_n\!=\!\varrho/[n \phi(b_n/\varrho)]$ and
  $b_n\!=\!\varrho\Phi^{-1}(1-1/n)$~\citet{resnick13}.
  The result follows from   Theorem~\ref{theo:evt}.
\end{proof}
This paper uses $\bR$ for which $\bR^{1/2}$ can be
shown, using Theorem 2 of \citet{maitra19}, to have no negative
entries. Then Corollary~\ref{theo:gumbel}  provides a
conservative bound for the quantiles of the limiting distribution of
$\max\bX\sim\mN(\bzero,\bR)$ with the conservatism 
determined by the negative elements of $\bR^{-1/2}$.

The thresholding steps yield {\em truncated} (and correlated) random
variables for potential thresholding  in subsequent
steps. We account for this added complication by deriving the limiting 
distribution of the maximum of a correlated sample from a
right-truncated normal distribution. 

Suppose that 
$Y_1,Y_2,\ldots Y_n$ are independent identically distributed (IID) random
variables from $\mN(0,\varrho^2)$ but truncated at $\eta$, then
each $Y_i$ has PDF $\phi_\eta^{\bullet}(y;\varrho)$ and
CDF $\Phi_\eta^{\bullet}(y;\varrho)$, where
\begin{equation}
  \label{truncnormal}
  \phi_\eta^{\bullet}(y;\varrho) = \frac{\phi(\frac y\varrho)}{\varrho\Phi(\frac\eta\varrho)} I(y < \eta);\:\:\: \Phi_\eta^{\bullet}(y;\varrho) = \frac{\Phi[\frac{\min(y,\eta)}\varrho]}{\Phi(\frac\eta\varrho)},
\end{equation}
where $I(\cdot)$ is the indicator function. Then
$Y_{(n)}=\max{\{Y_1,Y_2,\ldots,Y_n\}}$ has CDF 
$\Phi_{\eta,(n)}^\bullet(y;\varrho)\!= \!
\left[{\Phi(\frac{\min(y,\eta)}\varrho)}/{\Phi(\frac\eta\varrho)}\right]^n $ with limiting
distribution as follows.
\begin{theorem}
\label{theo:indeptruncnormal}
Let $Y_1,Y_2,\ldots,Y_n$ be a sample from \eqref{truncnormal}. Then
the limiting distribution of $Y_{(n)} $ 
  belongs to the domain of attraction of the reverse Weibull
distribution and satisfies
\begin{equation}
  \label{rev.weibull}
  \lim_{n\rightarrow\infty}[\Phi^{\bullet}_{\eta,(n)}(a^\bullet_nx+b_n^\bullet;\varrho)] = \exp{\{-(-x)^{-\nu}\}} I(x \leq 0).
\end{equation}
for some $\nu >0$. Here $a_n^\bullet = \eta - {\Phi_\eta^\bullet}^{-1}(1-1/n;\varrho)$ and $b_n^\bullet = \eta$.
\end{theorem}
\begin{proof}
Note that $\eta = \sup\{x\mid \Phi_\eta^\bullet(x) < 
1\}$.
From Theorem 10.5.2 in \citet{davidandnagaraja03}, for $Y_{(n)}$ to be
in the domain of attraction of the reverse Weibull, it sufficient to
show that 
\begin{equation*}
\lim_{y \rightarrow \eta}
\frac{(\eta-y)\phi^\bullet(y;\varrho)}{1-\Phi^\bullet(y;\varrho)} = \nu
\end{equation*}
for some $\nu>0$~\citep{vonMises36}. In our case, the limit holds
because  $\eta <\infty$. Then upon using  L'H\^opital's 
rule, we have 
\begin{equation*}
\begin{split}
  \lim_{y \rightarrow \eta}
  \frac{(\eta-y)\phi_\eta^\bullet(y;\varrho)}{1-\Phi_\eta^\bullet(y;\varrho)}  & =
  \lim_{y \rightarrow \eta} \frac{(\eta-y)\frac{d}{dx}\phi_\eta^{\bullet}(y;\varrho)-\phi_\eta^\bullet(y;\varrho)}{-\phi_\eta^\bullet(y;\varrho)}  \\
& = \lim_{y\rightarrow \eta} \frac{(\eta-y)\phi'(\frac
  y\varrho)/\varrho-\phi(\frac y\varrho)/\varrho}{-\phi(\frac y\varrho)/\varrho }=1.\\
\end{split}
\end{equation*}
Thus the right-truncated normal distribution satisfies the reverse Weibull condition and
converges to the reverse Weibull distribution with 
$\nu\! \equiv\! 1$ in \eqref{rev.weibull}. The constants in the
theorem are as per extreme value
theory~\citep{davidandnagaraja03,resnick13}. 
\end{proof}
\begin{theorem}
  \label{theo:trunc.corr.normal}
  Let $\bX$ 
  be a random vector from the
  $\mN_n(\bzero, \bR)$ density but that is right-truncated in each coordinate
  at $\eta$, with ${\bR}$ and $\varrho$ as in Theorem ~\ref{theo:evt}. 
Then  the CDF $F^\bullet_{\eta,(n)}(x)$ of $X_{(n)}$ is $F^\bullet_\eta(x) = \mP[X_{(n)} \leq
x] \geq  
\Phi_{\eta,(n)}^\bullet(x;\varrho)$.
\end{theorem}
\begin{proof}
  For IID random variables  $Y_1,Y_2,\ldots,Y_n$ from \eqref{truncnormal},
  $Y_{(n)}$ has
CDF $\Phi^\bullet_{\eta,(n)}(x;\varrho) = \mP[Y_{(n)}\leq x]
=\mP[U_{(n)}\leq \frac x\varrho] =
\Phi^\bullet_{\frac\eta\varrho,(n)}(\frac x\varrho;1)$, where $U_{(n)}
=\max\bU$ with $\bU$ a vector of $n$ IID random variables from $\bPhi^\bullet_{\frac\eta\varrho}$. Then 
  \begin{equation*}
    \begin{split}
      \Phi^\bullet_{\frac\eta\varrho,(n)}(x/\varrho;1)  = \mP[\bU \leq x \bone/\varrho]
      &\leq  \mP[{\bf R}^{1/2}\bY \leq x {\bf R}^{1/2}\bone/\varrho]\\
      &       = \mP[\bX \leq x \bone] \equiv   F^\bullet_\eta(x), 
    \end{split}
  \end{equation*}
  proving the statement of the theorem.
\end{proof}

\begin{corollary}
  \label{theo:evt.weibull}
  Let $\bX$ and $X_{(n)}$ be as in Theorem~\ref{theo:trunc.corr.normal}. Then the limiting
  distribution of $X_{(n)}$ belongs to the domain of attraction of
  the reverse Weibull distribution, and satisfies:
  \begin{equation}
    \lim_{n\rightarrow \infty}[F_{\eta}^\bullet (a^\bullet_nx
      + b^\bullet_n)] \geq \exp{\{-(-x)^{-\nu}\}} I(x \leq 0).
\end{equation}
where $a_n^\bullet = \eta - {\Phi_{\eta}^\bullet}^{-1}(1-1/n;\varrho)$ and $b_n^\bullet = \eta$.
\begin{proof}
The result is immediate from Theorems~\ref{theo:indeptruncnormal} 
and 
 \ref{theo:trunc.corr.normal}. 
\end{proof}
\end{corollary}

\subsection{Fast Adaptive Smoothing and Thresholding}
\label{fast}
We propose our FAST algorithm that adaptively and, in sequence, smooths and
identifies activated regions by thresholding. We estimate the
amount of smoothing robustly or from the correlation structure that we
assume is well-approximated by an ellipsoidally-symmetric 3D Gaussian
kernel oriented along the three axes and with 
parameters $\bh=(h_1,h_2,h_3)$. That is, under the 
null hypothesis (of no activation anywhere), we assume that the 
SPM $\bGamma\sim \mN_n(\bzero, \sigma^2\bR)$ where $\bR\! =\! \bS_{\bh}$, with
$\bS_{\bh}$ a circulant smoothing matrix~\citep{maitraandosullivan98}.
Let $\bGamma_{(-{\bh})} \sim \bS_{\bh}^{-\frac12}\bGamma$. We estimate
$\bh$ and $\sigma$
by maximizing 
the loglikelihood function
\begin{equation}
  \label{llhd}
  \ell(\bh\mid\sigma,\bGamma_{(-\bh)}) = 
  constant -
  \frac12\log|\sigma^2\bS_{\bh}| - \frac1{2\sigma^2}\bGamma_{(-\bh)}'\bGamma_{(-\bh)}.
\end{equation}
Both $\bGamma_{(-\bh)}$ and $|\bS_{\bh}|$ are
speedily computed using Fast Fourier Transforms (FFTs). 
Starting with the
SPM $\bGamma$, obtained
as discussed in Section~\ref{method:prelim}, we  propose the algorithm:
\begin{enumerate}
\item {\em Initial Setup}. At this point, assume that $\zeta_i
  \equiv 0$ $\forall$ $i$, where $\zeta_i$ is the activation status of the $i$th
  voxel. That is, assume that all voxels are inactive. Set
  $\zeta_i^{(0)} \equiv \zeta_i$. Also denote $\bGamma^{(0)} =
  \bGamma$, and $n_0=n$, where $n_k$ denotes the number of voxels for which
  $\zeta_i^{(k)} = 0$.  
\item {\em Iterative Steps}, \label{step2}
  For $k\!=\!1,2,\ldots,$ iterate as follows:
  \begin{enumerate}
  \item 
    {\em Smoothing}. Smooth  $\bGamma^{(k-1)}$ in one of three
    ways:
    \begin{enumerate}
    \item\label{ALL-FAST} {\em Adaptive Likelihood maximization
        (ALL-FAST, pronounced {\em\^ol-fast})}:
    Maximize~\eqref{llhd} given $\bGamma^{(k-1)}$     to obtain 
    ${\bh}^{(k)}$. Smooth $\bGamma^{(k-1)}$ with
    $\bS_{{\bh}^{(k)}}$ to get $\bGamma^{(k)}$. 
  \item\label{AM-FAST} {\em Adaptive Model-based smoothing
      (AM-FAST, pronounced {\em\u{a}m-fast})}:
    Use a Markov Random Field (MRF) prior model with parameters
    estimated using empirical Bayes methods as described in
    Section~\ref{mb.smooth} to get $\bGamma^{(k)}$ from
    $\bGamma^{(k-1)}$. 
    \item\label{AR-FAST} {\em Adaptive Robust smoothing (AR-FAST,
        {\em ahr-fast})}: Use \citep{garcia10} to smooth
      $\bGamma^{(k-1)}$ and get    $\bGamma^{(k)}$.
    \end{enumerate}
  \item\label{std} {\em Standardization.} Maximize~\eqref{llhd} given
    $\bGamma^{(k)}$ to obtain $\sigma_{(k)}$, ${\bh}^{(k)}$ and $\varrho_k=
    {\bS}_{{\bh}^{(k)}}^{\frac12}{\bone}$. Standardize
      $\bGamma^{(k)}$ by scaling with a robust version of
      $\sigma_{(k)}$, say $\tilde\sigma_{(k)}$.
\item {\em Adaptive Thresholding}. This consists of two steps:
  \ben
\item
  For $k=1$, use Corollary \ref{theo:gumbel} to obtain $(a_n,b_n)$,
  otherwise ({\em i.e.} for $k > 1$)
  use Corollary \ref{theo:evt.weibull} to get
  $(a^\bullet_{n_{k-1}},b_{n_{k-1}}^\bullet)$. In both cases,
use $\bR=\bS_{{\bh}^{(k)}}$.
\item From the Gumbel (for $k=1$) or reverse Weibull distributions (for
  $k>1$), get
  \begin{equation}
    \label{gumbel.weibull.cutoff}
    \eta_k = \begin{cases} {a_{n_0}}\iota_{\alpha}^\mG+ b_{n_0} &
      \mbox{ for }k=1\\
    {a^\bullet_{n_{k-1}}}  \iota_{\alpha}^\mW + b^\bullet_{n_{k-1}}&
      \mbox{ otherwise. }\end{cases}
  \end{equation}
  where $\iota_{\alpha}^\mG$ and $\iota_{\alpha}^\mW$ are the
  upper-tail $\alpha$-values for the Gumbel and the reverse Weibull
  (with $\nu = 1$) distributions, respectively. 
\item Set $\zeta^{(k)}_i = 1$ if $\zeta^{(k-1)}_i = 0$ and if the
  $i$th coordinate of $\bGamma^{(k)}$ exceeds $\eta_k$. Let $n_k=\sum_{i=1}^n\zeta_i^{(k)}$.
    \een
    \een
  \item {\em Termination}.
    \label{stop}
    Declare no activation and terminate if $\bzeta^{(1)}\equiv
    0$. Otherwise, let
    $J(\boldsymbol{\zeta}^{(k)},\boldsymbol{\zeta}^{(k-1)})$ be 
    the Jaccard Index~\citep{jaccard1901,maitra10} of the activation
    maps in the  $k$th and $(k-1)$th iterations.
    If $J(\boldsymbol{\zeta}^{(k)},\boldsymbol{\zeta}^{(k-1)}) \geq
    J(\boldsymbol{\zeta}^{(k+1)},\boldsymbol{\zeta}^{(k)})$,  the
    algorithm terminates -- the final activation map is $\boldsymbol{\zeta}^{(k)}$.
    \een
    
    \subsubsection*{Comments} A few comments are in order:
   \paragraph{Correlation structure} A circulant
    correlation structure allows for spatial context in the association
    between the values of the voxel-wise test statistics, while having
    the added benefit of speedy computations via the use of
    FFTs.
    \paragraph{Robust Estimation of $\sigma_{(k)}$} The estimate  
    $\hat\sigma_{(k)}$ from \eqref{llhd} assumes no  activation in
    $\bGamma_{(k)}$. Ignoring the activation can inflate the estimate,
    so we obtain $\tilde\sigma_{(k)}=\hat\sigma_{(k)}\tilde s_{(k;w)}/s_{(k)}$
    where $s_{(k)}$ is the estimated SD of $\bGamma_{(k)}$, and
    $\tilde s_{(k;w)}$ is its biweight-estimated
    SD~\citep{hoaglinetal00}, both assuming a zero mean. Specifically,
    if $\bGamma_{(k)}$ has $i$th component $\Gamma_{i(k)}$, we calculate 
 \begin{equation*}
\tilde s_{(k,w)}=\frac{\left(n\right)^{\frac{1}{2}}\left[\sum_{|e_{i(k),j}|<1}\Gamma_{i(k)}^2\left(1-e_{i(k)}^2\right)^4\right]^\frac{1}{2}}{\abs{\sum_{\abs{e_{i(k)}<1}}\left(1-e_{i(k)}^2\right)\left(1-5e_{i(k)}^2\right)}}
\end{equation*}
where $e_{i(k),j}={\Gamma_{i(k)}}/(w\tilde s_{(k)})$,  $\tilde
s_{(k)}$ the median absolute deviation of $\bGamma_{(k)}$ from 0, and
$w=\argmin_{w\in(0,6)} \tilde s_{(k,w)}$.   
  
   \paragraph{Comparison with AS} Our FAST algorithms are 
   similar to AS \citet{polzehletal10} in 
   that they also smooth and threshold 
    iteratively. But there are a few 
    fundamental differences. The AS approach has a set
    user-specified sequence of bandwidths that smooths
    $\bGamma^{(k)}$ at
    each step. In contrast, ALL-, AM- and AR-FAST use likelihood, empirical Bayes and robust methods to optimally determine
    $\bh$ at each step. AS also 
    thresholds but uses a general
    Fr\'echet extreme value distribution that ignores  
    spatial context and the correlated truncated nature of the random
    variables that arise from the smoothing and thresholding at each
    iteration. Our
    development represents the procedure more accurately because we
    account for both the correlation structure (with the initial
    cut-off decided as per the Gumbel distribution) and the
    truncation (with subsequent cut-offs determined by the
    reverse Weibull distribution). Our more general $\bh$ 
    allows for different amounts of
    smoothing in each     axis. Finally, our method is     entirely
    data-driven, with termination declared only if there is no initial
    activation or when $\mJ$ between subsequent activation maps decreases.
   \paragraph{Two-sided alternatives} Our development here builds from 
    one-sided tests where large values are the extreme
    values of the SPM. For two-sided alternatives, we  use the
    algorithm individually on the SPM and its negative, but replacing
    $\alpha$ in     \eqref{gumbel.weibull.cutoff} with
    $\alpha/2$. This provides two (disjoint) activation  
    maps, the union of which is the two-sided activation map.

\section{Performance Evaluations} \label{sec:simulation}
We studied performance of FAST relative to the most popular and
relevant alternatives. Our 
evaluations were on real and simulated datasets and compared FAST with
cluster thresholding (CT) applied with $\alpha = 0.001$ per
\citet{wooetal14}, a second-order
neighborhood and number of voxels in cluster determined by ~\citet{cox96}'s {\tt
  3dClustSim} function, 
threshold-free cluster enhancement (TFCE) \citep{smithandnichols09},
permutation-based 
testing~(PBT)~\citep{winkleretal14}, AS and PS (applied as AWS or
adaptive-weighted smoothing \citep{polzehl2006propagation}). We used
$\alpha=0.01$ and $\alpha=0.05$ in FAST to obtain insight 
into the role of $\alpha$. We
used R packages {\sc RFASTfMRI} for FAST, {\sc fMRI} for AS and AWS,
{\sc AnalyzeFMRI} for CT and  {\sc  permuco} for TFCE and PBT.


\subsection{Finger-Tapping Experiments}
Our first set of evaluations used the 12 replicated SPMs 
\citet{maitraetal02,maitra09b} 
from the right-hand (RH) and left-hand (LH) finger tapping
study  of a RH-dominant male.
For each method, Figure~\ref{fig:jaccard.index.finger} displays
\begin{figure}[h]
  \centering
\includegraphics[width=\columnwidth]{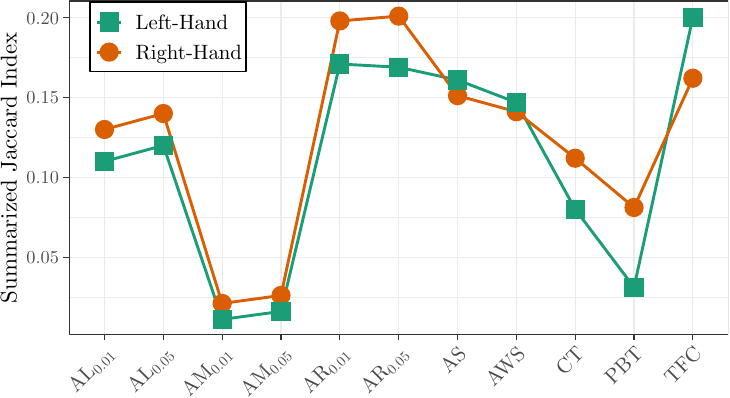}
\caption{The summarized Jaccard index ($\ddot\omega$) of activations obtained by
  each method over the replications for the RH and LH
    experiments. In this and other figures, AL$_a$, AR$_a$ and
    AM$_a$ denote ALL-, AM- and AR-FAST methods with $\alpha = a$, while
    TFC is used to abbreviate TFCE.}
  \label{fig:jaccard.index.finger}
\end{figure}
the summarized Jaccard similarity ($\ddot\omega$) between the
activation maps from the 12 replications. 
For the RH experiments, AR-FAST showed the highest reliability of
detected activation with $\alpha=0.05$. AR-FAST at $\alpha=0.01$ 
and TFCE were marginally behind and AS and AWS also doing reasonably. TFCE
was a bit better than AR-FAST for the LH experiments. 
The generally low $\ddot\omega$ for all methods points to potential issues in data  quality and processing
\citep{maitra10}.

\subsection{Experiments on Simulated Phantom Data}
Our next set of examples evaluated performance on phantom data
simulated using \eqref{eq:lm} under different conditions. 
\subsubsection{Motif and Stripes}
\label{tabelow}
\begin{figure}[h]
  \mbox{
    \subfloat[Motif]{
      \includegraphics[width=0.245\columnwidth]{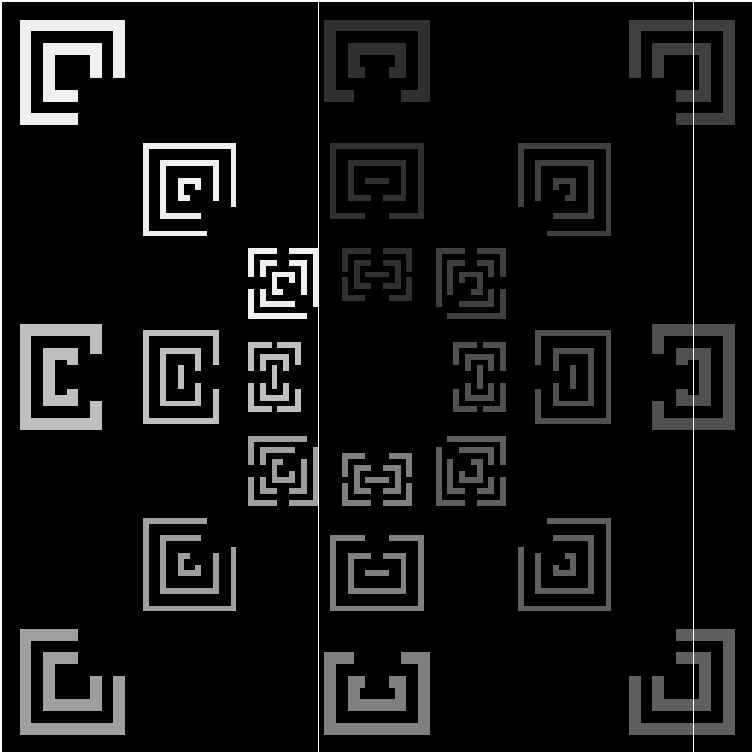}\label{motif}}
    \subfloat[$16\!\!\times\!\!16$ stripes]{
    \includegraphics[width=0.245\columnwidth]{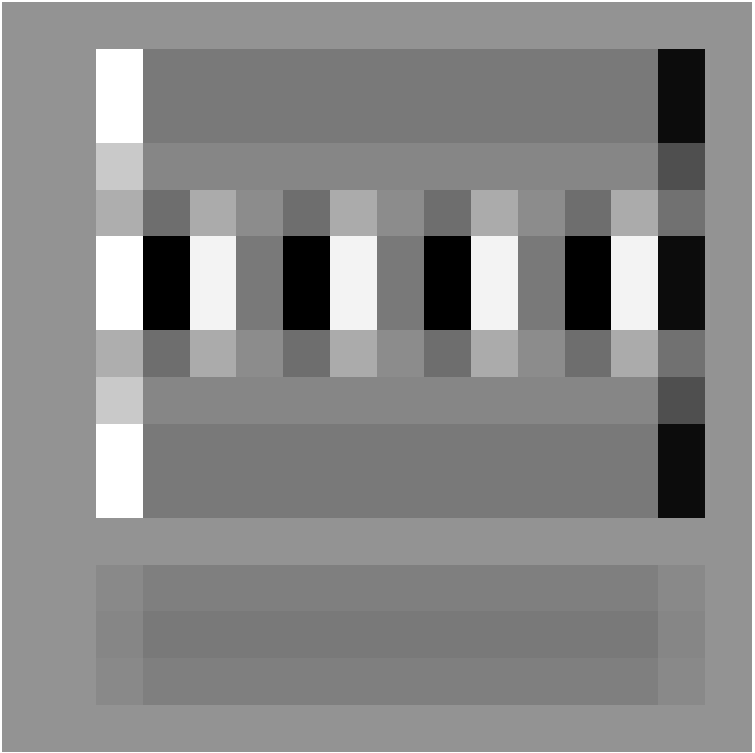}\label{str16}}
    \subfloat[$32\!\!\times\!\!32$ stripes]{
    \includegraphics[width=0.245\columnwidth]{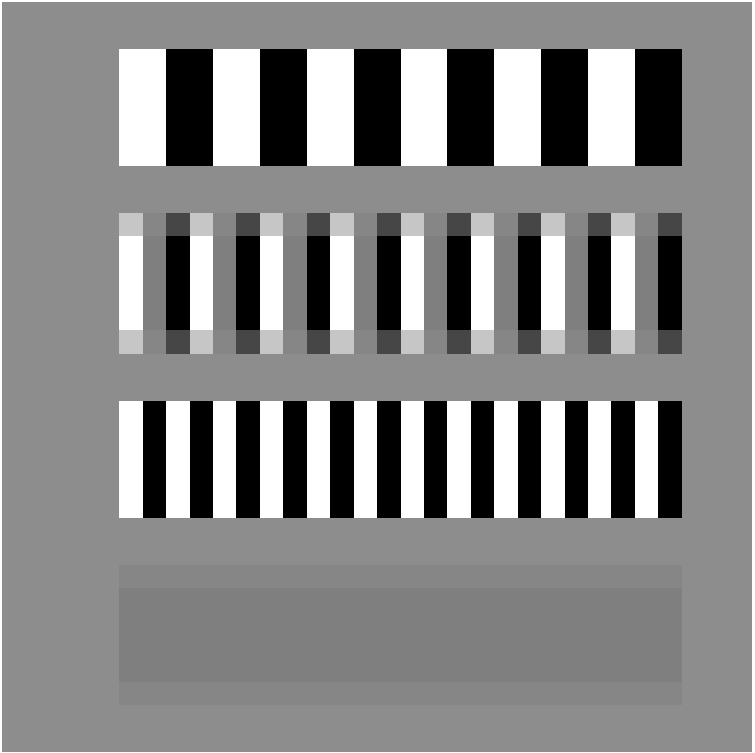}\label{str32}}
    \subfloat[$64\!\!\times\!\! 64$ stripes]{
    \includegraphics[width=0.245\columnwidth]{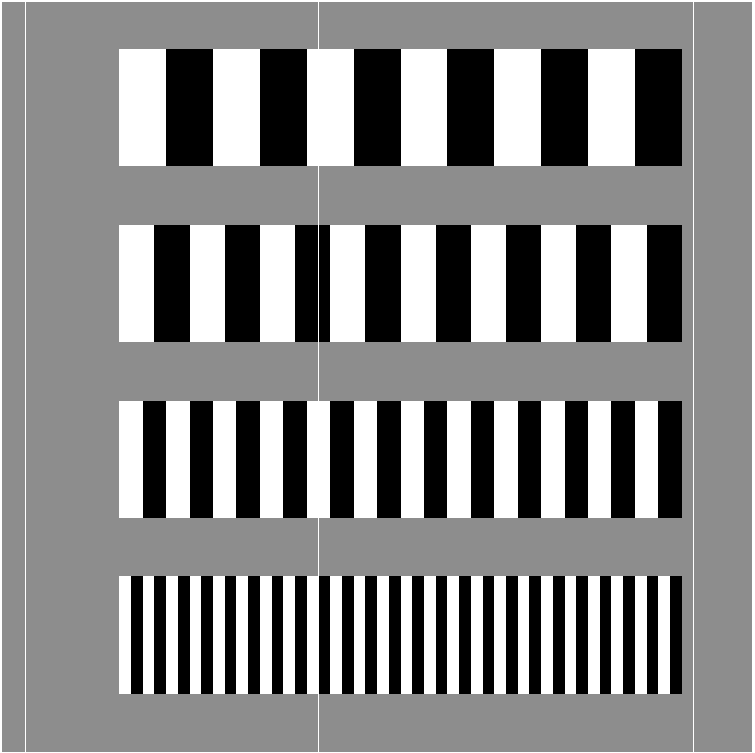}\label{str64}}
}
\mbox{
\subfloat[Performance using $\mJ$ for the two-sided phantom experiments]{\includegraphics[width=\columnwidth]{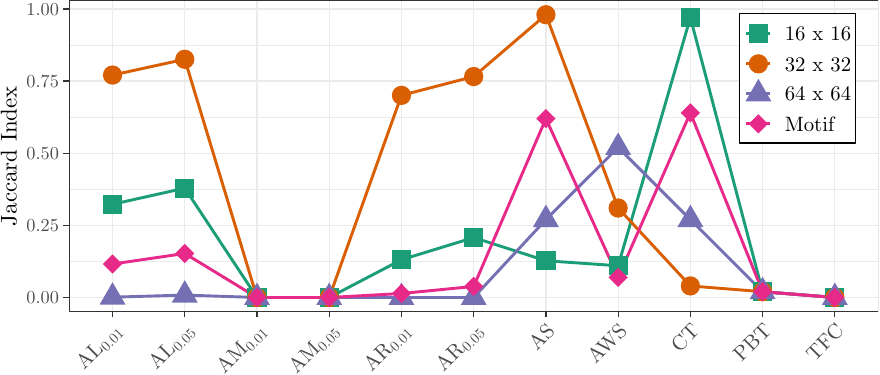}\label{JI-Tabelow}}}
\caption{(a)-(d) The phantoms from \citep{polzehletal10} and (e)
  performance of the  methods  under  different
  settings.}
\label{fig:tabelow}
\end{figure}
We first study performance using the simulation setup of 
\citet{polzehletal10}. We thank K. Tabelow for readily sharing code
that created  the motif and three striped ($16\times 16$, $32\times 32$,
$64\times 64$) phantoms of
Figures~\ref{fig:tabelow}\subref{motif}-\subref{str64}.
The phantoms have 14, 50, 28 and 47\% truly activated voxels, or
more than the 1-3\% expected in typical fMRI experiments. 
We used $\beta$s as per \citet{polzehletal10} and  CNRs
of between  0.75 to 2.68 for the motif and  1 to 2 for
the stripes. These examples are of two-sided alternatives. All simulations  had AR(1) errors with
$\rho=0.3$. For AS and AWS, we adopted the maximum bandwidth sequence values 
($h_k^*=3.06$, 1, 2 and 3)  in \citet{polzehletal10} for the four
respective phantoms as 
the best-case specific choices.
 Figure~\ref{fig:tabelow}\subref{JI-Tabelow} summarizes
 performance. There is no clear overall winner but AM-FAST,
 PBT and TFCE find no activation at
 all~(Figure~\ref{fig:tabelow-all}). ALL-FAST and AR-FAST, in that order,
 perform creditably  in some situations but not in others.
\subsubsection{Large-scale study with modified Hoffman phantom}
\label{hoff}
The phantoms in \citet{polzehletal10}, with uniform underlying
structure ($\beta_0$) and no drift, but varying CNR, are not particularly
representative of cerebral activation and do not provide much
insight into performance of different activation methods. So we
performed a large 
simulation study using a more realistic phantom and experimental setup
that matches \eqref{eq:lm}. 
We used a modified version of the digitized
$128\!\times\!128$ 2D Hoffman phantom~\citep{hoffmanetal90}  of
Positron Emission Tomography, with 
 3465 in-brain pixels, representing two types (say, A and B) of
anatomic structures -- the latter has 138 deemed truly activated  pixels in  two
distinct regions (Figure~\ref{fig:Hoff}).  
\begin{figure}[h]
    \centering
    \subfloat{
      \raisebox{-.5\height}{\includegraphics[width=0.12\textwidth]{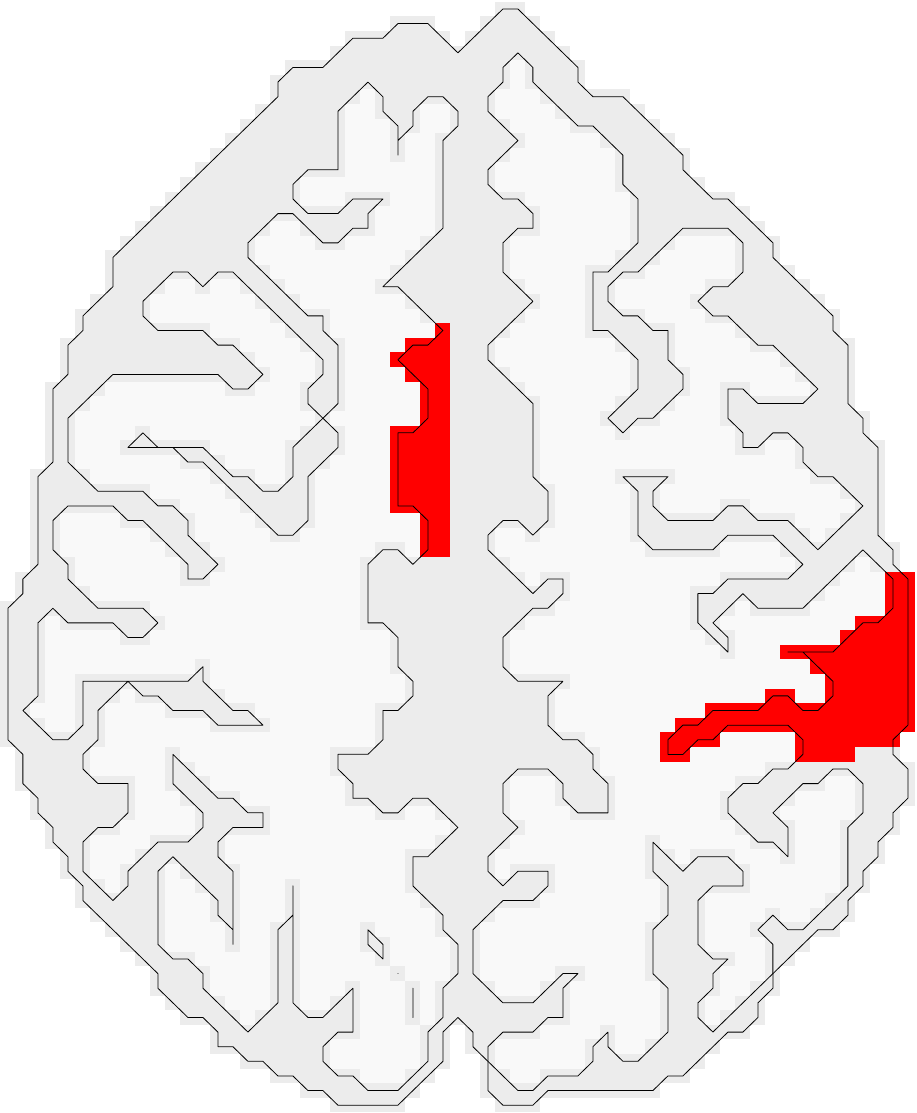}}
    }
    \subfloat{
 \begin{tabular}{|c|ccc|}
   \hline
   Region & $\beta_{i0}$ & $\beta_{i1}$ & $\beta_{i2}$\\
   \hline
   Background & 0  & 0 & 0 \\
   Inactivated, A & 4500 & 0 & -155.32\\
   Inactivated, B & 6000 & 0 & -155.32 \\
   Activated & 6000 & 600 & -155.32 \\
    \hline
  \end{tabular}
    }
    \caption{The modified Hoffman phantom. Putative anatomic  regions (A and B) are in
      shades of grey, with truly activated pixels  in red. The
      table lists the $\bbeta_i$s used in our simulations.}
\label{fig:Hoff}
\end{figure}
The $i$th pixel in the phantom had values 
$\bbeta_i\!=\!(\beta_{i0},\beta_{i1},\beta_{i2})$ in \eqref{eq:lm} as
per its location (see Figure~\ref{fig:Hoff}).

As in~\eqref{eq:lm}, the design matrix $\bX$ had the intercept in the
first column. The second column had the hemodynamic response function
(HRF)~\citep{lindquist08} convolved with the input stimulus time
series that alternated as 16 on-off blocks of 6 time-points
each. 
The third column of $\bX$ represented linear drift and was set to $t$
($t\!=\!1,2,\ldots,96$). As per~\eqref{eq:lm}, AR($p$) Gaussian errors were
simulated for different $p$ and at each pixel. Specifically, for each $p$, we considered AR
coefficients for a range of $\bphi\!\equiv\!\bphi_i$s with  coefficients
$(\phi_1,\phi_2,\ldots,\phi_p)$ that were, with lag,  (a) all equal, (b)
decreasing, (c) increasing, (d) first decreasing, 
then increasing and (e) first increasing and then decreasing. We
restricted $\sum^p_{j=1}\phi_{j} = 0.9$ to ensure stationary
solutions. So $\phi_1\equiv0.9$ for all AR(1) cases. For $p\!=\!2,3,4$,
$\phi_i\equiv 0.9/p$ for the equal AR coefficients scenario and as per 
Table~\ref{tab:Phivalues} for the other cases.
\begin{table}[htbp]
\caption[PhiValuesSimulation]{$\bphi$s for the AR($p$)
  scenarios used in our simulations.}\label{tab:Phivalues}
\centering
\begin{tabular}{|c|c|c|}
\hline
$p$ & Decreasing  & Increasing \\  
\hline
2  & (0.6, 0.3) & (0.3, 0.6) \\
3  & (0.4, 0.3, 0.2) & (0.2, 0.3, 0.4) \\
4  & (0.3, 0.25, 0.20, 0.15)& (0.15, 0.20, 0.25, 0.3) \\
5  & (0.3, 0.25, 0.20, 0.10, 0.05) & (0.05, 0.10, 0.20, 0.25, 0.30) \\
\hline
\hline
$p$ & Decreasing-Increasing & Increasing-Decreasing  \\  
\hline
2  & (0.6,0.3) & (0.3,0.6)\\
3  & (0.4,0.1,0.4) & (0.1,0.7,0.1) \\
4  & (0.4,0.05,0.05,0.4)&(0.05,0.4,0.4,0.05) \\
5  & (0.25,0.15,0.1,0.15,0.25) & (0.1,0.15,0.4,0.15,0.1)\\
\hline
\end{tabular}
\end{table}
Finally,  $\sigma_0 = 1200, 800, 600$ to correspond to very low to
moderate CNR = 0.50, 0.75, and 1.0. By design, our SNRs were 10 times
our CNRs. Time series images were
simulated using~\eqref{eq:lm} and the setup of Fig.~\ref{fig:Hoff} and
Table \ref{tab:Phivalues}.

For each pixel, we fit \eqref{eq:lm}  with $\hat p$
chosen from \{0,1,2,3,4,5\} using BIC. SPMs were generated as per
Section~\ref{method:prelim}. Figure~\ref{fig:sim.phant} provides
sample SPMs and activation maps with the three top-performing methods: AR-FAST, AM-FAST and AS, for the three CNR
settings with AR(2) errors and coefficients decreasing with order.
\begin{figure}[h]
  \begin{center}
\includegraphics[width=0.98\columnwidth]{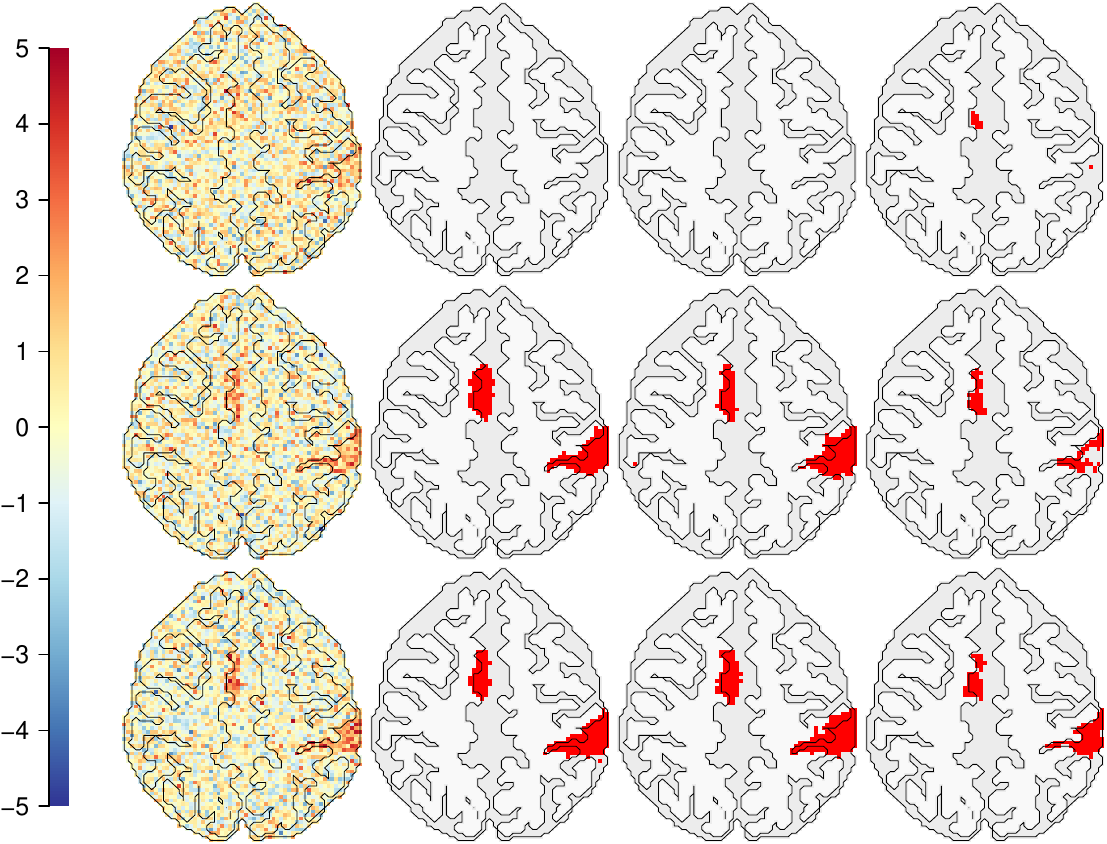}
\caption{Left-to-right: sample SPMs and best performers (AM-F$_{0.05}$,
  AR-F$_{0.05}$ and AS) for experiments with CNR = 0.5 (top row), CNR
  = 0.75 (middle) and CNR = 1.0 (bottom). AR-F$_{0.01}$ and
  AM-F$_{0.01}$ out-performed AS but are not displayed.}
\label{fig:sim.phant}
  \end{center}
\end{figure}
\begin{figure*}
  \begin{center}
\includegraphics[width=\textwidth]{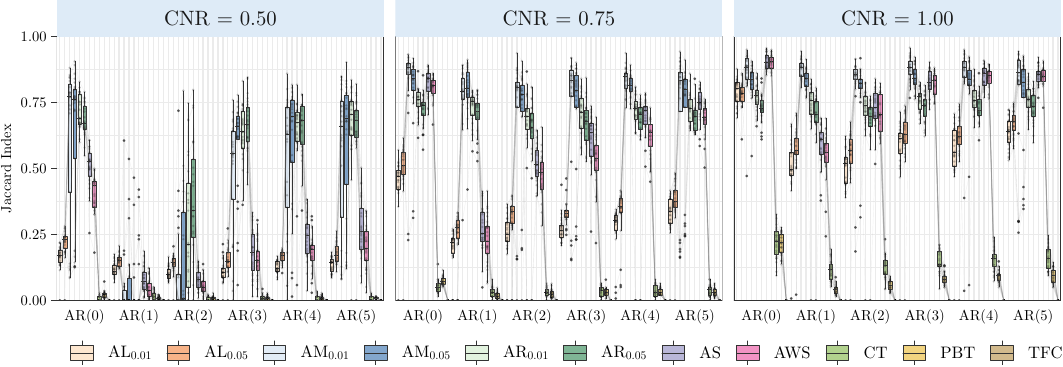}
\caption{Performance of the activation detection algorithms for different
  settings and when the AR coefficient decreases with order.
For clarity, each setting displays performance of the methods in the
same order as in the legend. 
}
\label{fig:Jaccard.index.decreasing}
  \end{center}
\end{figure*}
(See Figures~\ref{fig:ActMapDec1} for activation maps using all
methods and Figures~\ref{fig:ActMapDec1}--\ref{fig:ActMapInc2} for cases with
other $p$ and/or decreasing order.) All methods do poorly for CNR=0.5
in this example, but AS correctly finds a few activated voxels. AM-
and AR-FAST perform 
very well for CNR=0.75 and 1.0. Other methods -- in
particular ALL-FAST, CT, TFCE and PBT -- barely find  activation on
this example. 

To more fully understand performance, we replicated our experiment
25 times for each simulation setting. Figure~\ref{fig:ARp.estimated.order} shows performance in
estimating $p$, with over-estimation and mild under-estimation  
for large and small values of the true AR order. The pattern
broadly holds at all CNRs and $\phi$s. We now discuss performance of 
activation detection methods on SPMs obtained upon fitting AR($\hat p$). 
Figure~\ref{fig:Jaccard.index.decreasing} displays  overall
performance, in terms of $\omega$, of all methods in the
decreasing-$\phi$ cases. Performances with other types of $\phi$s
are similar
(Figure~\ref{fig:Jaccard.complete}). Figures~\ref{fig:ActivatedVoxels.complete}--\ref{fig:TrueNegative.complete} display the number of activated 
voxels and the false and true positive rates (FPRs and TPRs). 
FAST methods are among the top performers at all CNRs: AM- and
AR-FAST do very well for experiments with orders other than
AR(1). (All methods do poorly at CNR=0.50 in the AR(1) case: we 
surmise that is because of highly-correlated noise obtained with  $\phi=0.9$.)
AS and AWS are the next best  performers but CT, TFCE and PBT perform  
very poorly with  TPRs of below 25\%. FAST
methods have very high TPRs but FPRs of up to 
0.2, 0.3 and 0.5\% for ALL-, AM- and AR-FAST  for $\alpha = 0.05$, 
low CNR and $p$.
Overall, the slightly higher FPRs of the  FAST methods are overwhelmed
by their vastly higher TPRs, leading to their having the highest $\mJ$s.

Our threshold $\alpha$ has a role, with smaller values performing
better at higher CNRs and conversely. We suggest $\alpha\approx0.05$
for low-CNR tasks 
and $\alpha\approx0.01$ for high-CNR tasks. We suggest determining low or high CNR scenarios
accordingly as whether the upper percentile of the estimated
voxel-wise CNRs is less than the standard normal upper percentile
(2.33) or not: the upper percentile of the estimated CNRs 
is chosen to include an activated voxel (if such exists) in the
CNR determination. In our studies, ALL-FAST required more Step~\ref{step2}
iterations, but was computationally faster than AM-FAST AR-FAST and
had lower TPR, FPR and $\omega$, especially at low 
CNRs. Regardless, our methods were the fastest among all  
methods. 
\subsection{Performance in Null Activation Scenarios}
\subsubsection{Resting-State Dataset}
\label{resting}
A reviewer's suggestion led us to apply our FAST algorithms on SPMs
obtained upon 
fitting \eqref{eq:lm} to a  resting state 
dataset~\citet{barberetal2011,nebeletal2014}, with no  activation
identified even at $\alpha=0.05$. This zero FPR (when CNR=0) as
opposed to the small FPR in low-CNR experiments may be due
to  Step~\ref{stop} of our algorithm correctly allowing more
Step~\ref{step2} iterations to
attenuate stray high-valued SPM voxels -- termination is earlier
in the low CNR cases 
given the spatially located weaker-signal peaks. For higher CNRs,
Step~\ref{stop} again adaptively admits more smoothing iterations
that  dampen stray high values in the SPM without substantially
degrading the true high-signal peaks. 

\subsubsection{Null-simulated SPMs} Another reviewer was concerned
about multiple significance. Our use of $\alpha$ is to set a
threshold and not as a significance level: still, experiments on
simulated null SPMs~(Section~\ref{null}) indicate no practical
concerns.

\section{Activation During Perception of Noisy Speech}\label{sec:Apps}
The dataset, provided as  {\tt data6} in the AFNI tutorial~\citet{cox96},
is originally from an fMRI study~\citet{nathandbeauchamp11} where
\begin{figure}[h]
\subfloat[]{\includegraphics[width=0.5\columnwidth]{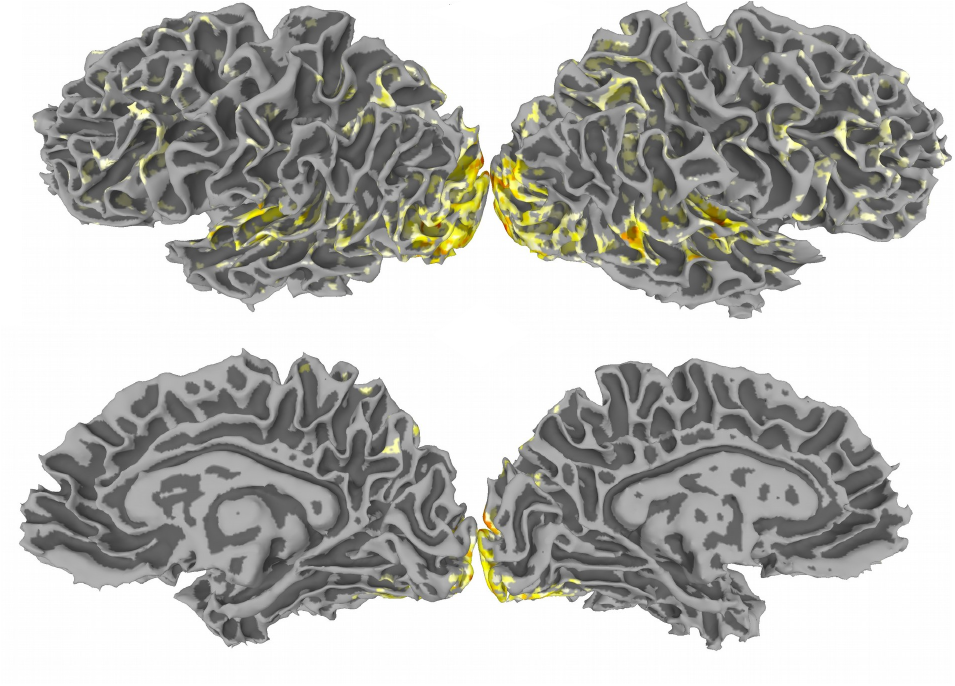}}
\subfloat[]{\includegraphics[width=0.5\columnwidth]{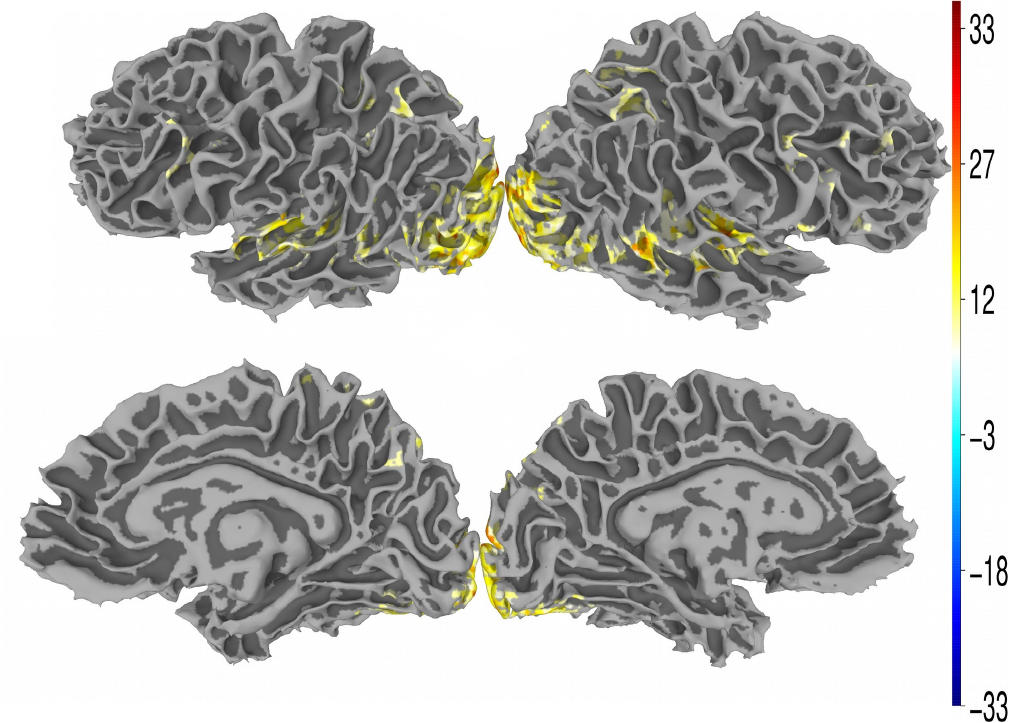}}
\caption{ AR-FAST-identified activation regions on SPMs obtained by
  fitting ~\ref{eq:lm} with   AR($\hat{p}$) to AFNI's {\tt data6} for
  (a) visual-reliable stimulus and (b) audio-reliable
  stimulus.}
\label{fig:AMSmoothingAFNI}
\end{figure}
a subject heard and saw a female volunteer speak words, separately, in
two different formats. The audio-reliable setting had the subject
clearly hear the spoken word but see a degraded image of the speaker
while the visual-reliable case had the subject clearly see the speaker
vocalize the word but the audio was of reduced quality.  There were
three experimental runs, each 
consisting of a randomized design of 10 blocks, equally divided into blocks of
audio-reliable and visual-reliable stimuli. 
$\mbox{T}_2^*$-weighted images with volumes of $80 \times 80 \times
33$ (with voxels of dimension $2.75 \times  2.75 \times 3.0\  mm^3$)
from  echo-planar sequences (TR=2s) 
were obtained  over $152$ time-points. Our interest was in determining 
activation corresponding to the audio
($H_0:\beta_{a}=0$) and visual
($H_0:\beta_{v}=0$) tasks.
At each voxel, we fitted AR models for 
$p=0,1,2,3,4,5$ and chose $p$ with the highest BIC. 
Figure \ref{fig:AMSmoothingAFNI} uses AFNI and Surface Mapping (SUMA)
to display activated regions obtained using AR-FAST on the SPM:
see  Figure~\ref{fig:Visual-Audio} for  maps drawn from ALL-FAST, AS,
AWS and CT. We used $\alpha = 0.01$ because of the high (greater than
4) upper percentile of the voxel-wise estimated CNRs. Most of the activation 
occurs in Brodmann areas 18 and 19 (BA18 and BA19)
which comprise the
occipital cortex 
and the extrastriate (or peristriate) cortex. In humans with normal
sight,  this area is for visual association where 
feature-extraction, shape recognition, attentional and multimodal
integrating functions occur. We also see increased activation in
the STS, which recent
studies~\citep{grossman2001brain} have related to  distinguishing
voices from environmental sounds, 
stories versus nonsensical speech, moving faces versus moving objects,
biological motion and so on. ALL-FAST performs similarly as AR-FAST,
while the other methods also identify the same regions but they identify
a lot more activated 
voxels, some of which appear to be false positives. Although a 
detailed analysis of the results of this study is beyond the purview
of this paper, we note that AR-FAST 
finds interpretable results even when applied to a single
subject high-level cognition experiment. 

\section{Discussion}\label{discussion}
We propose a new  fully automated fast adaptive 
smoothing and thresholding algorithm suite called FAST with the
ability to detect activation in  low-signal settings. Three
variants -- ALL-FAST, AM-FAST and AR-FAST -- are proposed with AR-FAST
generally recommended because of its consistent good performance
across a range of low-CNR experiments and real datasets.
 AM-FAST's performance, while good, is more
variable, while   ALL-FAST appears to undersmooth but performs better
in two-sided activation detection scenarios. Our methodology realistically
accounts for both spatial correlation structure and is also developed
under more accurate extreme value theory. Our algorithm suite is
implemented in a R package {\sc RFASTfMRI} available at
\href{https://github.com/ialmodovar/RFASTfMRI}{https://github.com/ialmodovar/RFASTfMRI} and is fully automated with one threshold choice for which we provide 
easily-implemented guidance. This contrasts with AS and AWS that
require setting  maximum smoothing bandwidths related to the expected
diameter of activated regions~\cite{polzehletal10} -- a determination
that may require considerable dexterity and is ambivalent when different-sized
activation regions are expected.

A reviewer has pointed to the joint detection-estimation
literature~\citep{maknietal05,maknietal06} where estimation of the HRF
and activation detection occur jointly. The FAST, AS and AWS algorithms
can be placed  in a related framework, with the distinction that the
estimation step is of a more spatially consistent (smoothed)
SPM. We also agree with another reviewer on other ways of ensuring
spatial contiguity such as through Markov Random Field
priors~\citep{ngetal12} and on the need to incorporate
approaches also allowing for nonhomogeneous smoothing.
Our algorithms converge by construction and are guaranteed to
terminate. They are also seen to have good overall performance  but,
as observed by a  reviewer, establishing the 
optimality   properties  and conditions and
assumptions governing such properties may 
provide more solid theoretical grounding for FAST and improve its
understanding and widen its applicability. 
Developing FAST for more  sophisticated time series  and spatial
models, including in the 
context of complex-valued 
fMRI~\citep{adrianetal18} as well as increased use of diagnostics in
understanding activation and cognition are other important research areas
and directions that would benefit from further attention.

\section*{Acknowledgment}
The authors sincerely thank four anonymous reviewers and an Associate
Editor whose helpful and insightful comments on an earlier
version of this article greatly improved its content.
\bibliographystyle{IEEEtran}
\bibliography{rm,references}

\begin{thebibliography}{10}
\providecommand{\url}[1]{#1}
\csname url@samestyle\endcsname
\providecommand{\newblock}{\relax}
\providecommand{\bibinfo}[2]{#2}
\providecommand{\BIBentrySTDinterwordspacing}{\spaceskip=0pt\relax}
\providecommand{\BIBentryALTinterwordstretchfactor}{4}
\providecommand{\BIBentryALTinterwordspacing}{\spaceskip=\fontdimen2\font plus
\BIBentryALTinterwordstretchfactor\fontdimen3\font minus
  \fontdimen4\font\relax}
\providecommand{\BIBforeignlanguage}[2]{{%
\expandafter\ifx\csname l@#1\endcsname\relax
\typeout{** WARNING: IEEEtran.bst: No hyphenation pattern has been}%
\typeout{** loaded for the language `#1'. Using the pattern for}%
\typeout{** the default language instead.}%
\else
\language=\csname l@#1\endcsname
\fi
#2}}
\providecommand{\BIBdecl}{\relax}
\BIBdecl

\bibitem{belliveauetal91}
J.~W. Belliveau, D.~N. Kennedy, R.~C. McKinstry, B.~R. Buchbinder, R.~M.
  Weisskoff, M.~S. Cohen, J.~M. Vevea, T.~J. Brady, and B.~R. Rosen,
  ``Functional mapping of the human visual cortex by magnetic resonance
  imaging,'' \emph{Science}, vol. 254, pp. 716--719, 1991.

\bibitem{kwongetal92}
K.~K. Kwong, J.~W. Belliveau, D.~A. Chesler, I.~E. Goldberg, R.~M. Weisskoff,
  B.~P. Poncelet, D.~N. Kennedy, B.~E. Hoppel, M.~S. Cohen, R.~Turner, H.-M.
  Cheng, T.~J. Brady, and B.~R. Rosen, ``Dynamic magnetic resonance imaging of
  human brain activity during primary sensory stimulation,'' \emph{Proceedings
  of the National Academy of Sciences of the United States of America},
  vol.~89, pp. 5675--5679, 1992.

\bibitem{bandettinietal93}
P.~A. Bandettini, A.~Jesmanowicz, E.~C. Wong, and J.~S. Hyde, ``Processing
  strategies for time-course data sets in functional {MRI} of the human
  brain,'' \emph{Magnetic Resonance in Medicine}, vol.~30, pp. 161--173, 1993.

\bibitem{fristonetal95}
K.~J. Friston, A.~P. Holmes, K.~J. Worsley, J.-B. Poline, C.~D. Frith, and
  R.~S.~J. Frackowiak, ``Statistical parametric maps in functional imaging: A
  general linear approach,'' \emph{Human Brain Mapping}, vol.~2, pp. 189--210,
  1995.

\bibitem{howsemanandbowtell98}
A.~M. Howseman and R.~W. Bowtell, ``Functional magnetic resonance imaging:
  imaging techniques and contrast mechanisms,'' \emph{Philosophical
  Transactional of the Royal Society, London}, vol. 354, pp. 1179--94, 1999.

\bibitem{pennyetal06}
W.~D. Penny, K.~J. Friston, J.~T. Ashburner, S.~J. Kiebel, and T.~E. Nichols,
  Eds., \emph{Statistical Parametric Mapping: The Analysis of Functional Brain
  Images}, 1st~ed.\hskip 1em plus 0.5em minus 0.4em\relax Academic Press, 2006.

\bibitem{lindquist08}
M.~A. Lindquist, ``The statistical analysis of {fMRI} data,'' \emph{Statistical
  Science}, vol.~23, no.~4, pp. 439--464, 2008.

\bibitem{lazar08}
N.~A. Lazar, \emph{The Statistical Analysis of Functional MRI Data}.\hskip 1em
  plus 0.5em minus 0.4em\relax Springer, 2008.

\bibitem{ashby11}
F.~G. Ashby, \emph{Statistical Analysis of {fMRI} Data}.\hskip 1em plus 0.5em
  minus 0.4em\relax MIT Press, 2011.

\bibitem{ogawaetal90a}
S.~Ogawa, T.~M. Lee, A.~S. Nayak, and P.~Glynn, ``Oxygenation-sensitive
  contrast in magnetic resonance image of rodent brain at high magnetic
  fields,'' \emph{Magnetic Resonance in Medicine}, vol.~14, pp. 68--78, 1990.

\bibitem{ogawaetal90b}
S.~Ogawa, T.~M. Lee, A.~R. Kay, and D.~W. Tank, ``Brain magnetic resonance
  imaging with contrast dependent on blood oxygenation,'' \emph{Proceedings of
  the National Academy of Sciences, {USA}}, vol.~87, no.~24, pp. 9868--9872,
  1990.

\bibitem{worsleyetal02}
K.~J. Worsley, C.~H. Liao, J.~Aston, V.~Petre, G.~H. Duncan, F.~Morales, and
  A.~C. Evans, ``A general statistical analysis for fmri data,''
  \emph{NeuroImage}, vol.~15, pp. 1--15, 2002.

\bibitem{fristonetal98}
K.~J. Friston, P.~Fletcher, O.~Josephs, A.~Holmes, M.~Rugg, and R.~Turner,
  ``Event-related {fMRI}: characterizing differential responses,''
  \emph{Neuroimage}, vol.~7, no.~1, pp. 30--40, 1998.

\bibitem{glover99}
G.~H. Glover, ``Deconvolution of impulse response in event-related bold fmri,''
  \emph{NeuroImage}, vol.~9, pp. 416--429, 1999.

\bibitem{buxtonetal04}
R.~B. Buxton, K.~Uluda{\u{g}}, D.~J. Dubowitz, and T.~T. Liu, ``Modeling the
  hemodynamic response to brain activation,'' \emph{Neuroimage}, vol.~23, pp.
  S220--S233, 2004.

\bibitem{fristonetal90}
K.~J. Friston, C.~D. Frith, P.~F. Liddle, R.~J. Dolan, A.~A. Lammertsma, and
  R.~S.~J. Frackowiak, ``The relationship between global and local changes in
  {PET} scans,'' \emph{Journal of Cerebral Blood Flow and Metabolism}, vol.~10,
  pp. 458--466, 1990.

\bibitem{friston1994statistical}
K.~J. Friston, A.~P. Holmes, K.~J. Worsley, J.-P. Poline, C.~D. Frith, and
  R.~S. Frackowiak, ``Statistical parametric maps in functional imaging: a
  general linear approach,'' \emph{Human brain mapping}, vol.~2, no.~4, pp.
  189--210, 1994.

\bibitem{worsley1995analysis}
K.~J. Worsley and K.~J. Friston, ``Analysis of {fMRI} time-series
  revisited—again,'' \emph{Neuroimage}, vol.~2, no.~3, pp. 173--181, 1995.

\bibitem{genoveseetal02}
C.~R. Genovese, N.~A. Lazar, and T.~Nichols, ``Thresholding of statistical maps
  in functional neuroimaging using the false discovery rate,''
  \emph{Neuroimage}, vol.~15, pp. 870--878, 2002.

\bibitem{hajnaletal94}
J.~V. Hajnal, R.~Myers, A.~Oatridge, J.~E. Schweiso, J.~R. Young, and G.~M.
  Bydder, ``Artifacts due to stimulus-correlated motion in functional imaging
  of the brain,'' \emph{Magnetic Resonance in Medicine}, vol.~31, pp. 283--291,
  1994.

\bibitem{biswaletal96}
B.~Biswal, A.~E. DeYoe, and J.~S. Hyde, ``Reduction of physiological
  fluctuations in {fMRI} using digital filters.'' \emph{Magnetic Resonance in
  Medicine}, vol.~35, no.~1, pp. 107--113, January 1996.

\bibitem{woodetal98}
R.~P. Wood, S.~T. Grafton, J.~D.~G. Watson, N.~L. Sicotte, and J.~C. Mazziotta,
  ``Automated image registration. ii. intersubject validation of linear and
  non-linear models,'' \emph{Journal of Computed Assisted Tomography}, vol.~22,
  pp. 253--265, 1998.

\bibitem{gullapallietal05}
R.~P. Gullapalli, R.~Maitra, S.~Roys, G.~Smith, G.~Alon, and J.~Greenspan,
  ``Reliability estimation of grouped functional imaging data using penalized
  maximum likelihood,'' \emph{Magnetic Resonance in Medicine}, vol.~53, pp.
  1126--1134, 2005.

\bibitem{maitraetal02}
R.~Maitra, S.~R. Roys, and R.~P. Gullapalli, ``Test-retest reliability
  estimation of functional mri data,'' \emph{Magnetic Resonance in Medicine},
  vol.~48, pp. 62--70, 2002.

\bibitem{maitra09b}
R.~Maitra, ``Assessing certainty of activation or inactivation in test-retest
  {fMRI} studies,'' \emph{Neuroimage}, vol.~47, no.~1, pp. 88--97, 2009.

\bibitem{saadetal09}
Z.~S. Saad, D.~R. Glen, G.~Chen, M.~S. Beauchamp, R.~Desai, and R.~W. Cox, ``A
  new method for improving functional-to-structural mri alignment using local
  pearson correlation,'' \emph{NeuroImage}, vol.~44, pp. 839--848, 2009.

\bibitem{chenandsmall07}
E.~E. Chen and S.~L. Small, ``Test-retest reliability in {fMRI} of language:
  Group and task effects,'' \emph{Brain and Language}, vol. 102, no.~2, pp.
  176--85, 2007.

\bibitem{kryter94}
K.~D. Kryter, \emph{The handbook of hearing and the effects of noise:
  Physiology, psychology, and public health}.\hskip 1em plus 0.5em minus
  0.4em\relax Academic Press, 1994.

\bibitem{hauser96}
M.~D. Hauser, \emph{The evolution of communication}.\hskip 1em plus 0.5em minus
  0.4em\relax MIT press, 1996.

\bibitem{dupontandluettin00}
S.~Dupont and J.~Luettin, ``Audio-visual speech modeling for continuous speech
  recognition,'' \emph{IEEE {T}ransactions on {M}ultimedia}, vol.~2, no.~3, pp.
  141--151, 2000.

\bibitem{nathandbeauchamp11}
A.~R. Nath and M.~S. Beauchamp, ``Dynamic changes in superior temporal sulcus
  connectivity during perception of noisy audiovisual speech,'' \emph{The
  Journal of Neuroscience}, vol.~31, no.~5, p. 1704 –1714, 2011.

\bibitem{sumbyandpollack54}
W.~H. Sumby and I.~Pollack, ``Visual contribution to speech intelligibility in
  noise,'' \emph{The {J}ournal of the {A}coustical {S}ociety of {A}merica},
  vol.~26, no.~2, pp. 212--215, 1954.

\bibitem{steinandmeredith93}
B.~E. Stein and M.~A. Meredith, \emph{The merging of the senses.}\hskip 1em
  plus 0.5em minus 0.4em\relax The MIT Press, 1993.

\bibitem{helleretal06}
R.~Heller, D.~Stanley, D.~Yekutieli, N.~Rubin, and Y.~Benjamini,
  ``Cluster-based analysis of {fMRI} data,'' \emph{NeuroImage}, vol.~33, no.~2,
  pp. 599--608, nov 2006.

\bibitem{benjaminiandheller07}
Y.~Benjamini and R.~Heller, ``False discovery rates for spatial signals,''
  \emph{Journal of the American Statistical Association}, vol. 102, no. 480,
  pp. 1272--1281, 2007.

\bibitem{smithandfahrmeir07}
M.~Smith and L.~Fahrmeir, ``Spatial {B}ayesian variable selection with
  application to functional {M}agnetic {R}esonance {I}maging,'' \emph{Journal
  of the American Statistical Association}, vol. 102, no. 478, pp. 417--431,
  2007.

\bibitem{smithandnichols09}
S.~M. Smith and T.~E. Nichols, ``Threshold-free cluster enhancement: Addressing
  problems of smoothing, threshold dependence and localisation in cluster
  inference,'' \emph{Neuroimage}, vol.~44, pp. 83--98, 2009.

\bibitem{wooetal14}
C.-W. Woo, A.~Krishnan, and T.~D. Wager, ``Cluster-extent based thresholding in
  {fMRI} analyses: Pitfalls and recommendations,'' \emph{Neuroimage}, vol.~91,
  p. 412–419, 2014.

\bibitem{tabelowetal06}
K.~Tabelow, J.~Polzehl, H.~U. Voss, and V.~Spokoiny, ``Analyzing
  {fMRI}experiments with structural adaptive smoothing procedures,''
  \emph{NeuroImage}, vol.~33, no.~1, pp. 55--62, 2006.

\bibitem{polzehletal10}
J.~Polzehl, H.~U. Voss, and K.~Tabelow, ``Structural adaptive segmentation for
  statistical parametric mapping,'' \emph{NeuroImage}, vol.~52, no.~2, pp.
  515--523, 2010.

\bibitem{monti11}
M.~M. Monti, ``Statistical analysis of {fMRI} time-series: A critical review of
  the glm approach,'' \emph{Frontiers in Human Neuroscience}, vol.~5, no.
  00028, pp. 1--13, 2011.

\bibitem{luoandnichols03}
W.-L. Luo and T.~E. Nichols, ``Diagnosis and exploration of massively
  univariate neuroimaging models,'' \emph{Neuroimage}, vol.~19, no.~3, pp.
  1014--1032, 2003.

\bibitem{lohetal08}
J.~M. Loh, M.~A. Lindquist, and T.~D. Wager, ``Residual analysis for detecting
  mis-modeling in {fMRI},'' \emph{Statistica Sinica}, pp. 1421--1448, 2008.

\bibitem{lindquistetal09}
M.~Lindquist, J.~Loh, L.~Atlas, and T.~Wager, ``Modeling the hemodynamic
  response function in {fMRI}: {E}fficiency, bias and mis-modeling,''
  \emph{Neuroimage}, vol.~45, no.~1, p. S187–S196, 2009.

\bibitem{schwarz78}
G.~Schwarz, ``Estimating the dimensions of a model,'' \emph{Annals of
  Statistics}, vol.~6, pp. 461--464, 1978.

\bibitem{shumwayandstoffer06}
R.~H. Shumway and D.~S. Stoffer, \emph{Time Series Analysis and Its
  Applications}, 2nd~ed.\hskip 1em plus 0.5em minus 0.4em\relax Springer, 2006.

\bibitem{resnick13}
S.~I. Resnick, \emph{Extreme values, regular variation and point
  processes}.\hskip 1em plus 0.5em minus 0.4em\relax Springer, 2013.

\bibitem{maitra19}
R.~{Maitra}, ``{Efficient Bandwidth Estimation in Two-dimensional Filtered
  Backprojection Reconstruction},'' \emph{arXiv e-prints}, p. arXiv:1803.08027,
  Mar 2018.

\bibitem{davidandnagaraja03}
H.~A. David and H.~N. Nagaraja, \emph{Order Statistics}.\hskip 1em plus 0.5em
  minus 0.4em\relax Hoboken, New Jersey: John Wiley and Sons, Inc., 2003.

\bibitem{vonMises36}
R.~Von~Mises, ``La distribution de la plus grande de n valeurs,'' \emph{Rev.
  math. Union interbalcanique}, vol.~1, no.~1, 1936.

\bibitem{maitraandosullivan98}
R.~Maitra and F.~O'Sullivan, ``Variability assessment in {P}ositron {E}mission
  {T}omography and related generalized deconvolution models,'' \emph{Journal of
  the American Statistical Association}, vol.~93, pp. 1340--1355, 1998.

\bibitem{garcia10}
D.~Garcia, ``Robust smoothing of gridded data in one and higher dimensions with
  missing values,'' \emph{Computational statistics \& data analysis}, vol.~54,
  no.~4, pp. 1167--1178, 2010.

\bibitem{jaccard1901}
P.~Jaccard, ``\`{E}tude comparative de la distribution florale dans une portion
  des alpes et des jura,'' \emph{Bulletin del la Soci\`{e}t\`{e} Vaudoise des
  Sciences Naturelles}, vol.~37, p. 547–579, 1901.

\bibitem{maitra10}
R.~Maitra, ``A re-defined and generalized percent-overlap-of-activation measure
  for studies of {fMRI} reproducibility and its use in identifying outlier
  activation maps,'' \emph{Neuroimage}, vol.~50, no.~1, pp. 124--135, 2010.

\bibitem{hoaglinetal00}
D.~Hoaglin, F.~Mosteller, and J.~Tukey, \emph{Understanding Robust and
  Exploratory Data Analysis}.\hskip 1em plus 0.5em minus 0.4em\relax New York,
  NY: John Wiley and Sons, Inc., 2000.

\bibitem{cox96}
R.~W. Cox, ``{AFNI}: software for analysis and visualization of functional
  magnetic resonance neuroimages,'' \emph{Computers and Biomedical research},
  vol.~29, no.~3, pp. 162--173, 1996.

\bibitem{winkleretal14}
A.~M. Winkler, G.~R. Ridgway, M.~A. Webster, S.~M. Smith, and T.~E. Nichols,
  ``Permutation inference for the general linear model,'' \emph{Neuroimage},
  vol.~92, pp. 381--397, 2014.

\bibitem{polzehl2006propagation}
J.~Polzehl and V.~Spokoiny, ``Propagation-separation approach for local
  likelihood estimation,'' \emph{Probability Theory and Related Fields}, vol.
  135, no.~3, pp. 335--362, 2006.

\bibitem{hoffmanetal90}
E.~J. Hoffman, P.~D. Cutler, W.~M. Digby, and J.~C. Mazziotta, ``3-d phantom to
  simulate cerebral blood flow and metabolic images for pet,'' \emph{IEEE
  Transactions on Nuclear Science}, vol.~37, pp. 616--620, 1990.

\bibitem{barberetal2011}
A.~D. Barber, P.~Srinivasan, S.~E. Joel, B.~S. Caffo, J.~J. Pekar, and S.~H.
  Mostofsky, ``Motor “dexterity”?: evidence that left hemisphere
  lateralization of motor circuit connectivity is associated with better motor
  performance in children,'' \emph{Cerebral Cortex}, vol.~22, no.~1, pp.
  51--59, 2011.

\bibitem{nebeletal2014}
M.~B. Nebel, S.~E. Joel, J.~Muschelli, A.~D. Barber, B.~S. Caffo, J.~J. Pekar,
  and S.~H. Mostofsky, ``Disruption of functional organization within the
  primary motor cortex in children with autism,'' \emph{Human brain mapping},
  vol.~35, no.~2, pp. 567--580, 2014.

\bibitem{grossman2001brain}
E.~Grossman and R.~Blake, ``Brain activity evoked by inverted and imagined
  biological motion,'' \emph{Vision research}, vol.~41, no.~10, pp. 1475--1482,
  2001.

\bibitem{maknietal05}
S.~Makni, P.~Ciuciu, J.~Idier, and J.-B. Poline, ``Joint detection-estimation
  of brain activity in functional {MRI}: a multichannel deconvolution
  solution,'' \emph{IEEE Transactions on Signal Processing}, vol.~53, no.~9,
  pp. 3488--3502, 2005.

\bibitem{maknietal06}
------, ``Joint detection-estimation of brain activity in {fMRI} using an
  autoregressive noise model,'' in \emph{3rd IEEE International Symposium on
  Biomedical Imaging: Nano to Macro}.\hskip 1em plus 0.5em minus 0.4em\relax
  IEEE, 2006, pp. 1048--1051.

\bibitem{ngetal12}
B.~Ng, G.~Hamarneh, and R.~Abugharbieh, ``Modeling brain activation in f{MRI}
  using group {MRF},'' \emph{IEEE Transactions on Medical Imaging}, vol.~31,
  no.~5, pp. 1113--1123, 2012.

\bibitem{adrianetal18}
D.~W. Adrian, R.~Maitra, and D.~B. Rowe, ``Complex-valued time series modeling
  for improved activation detection in f{MRI} studies,'' \emph{Annals of
  Applied Statistics}, vol.~12, no.~3, pp. 1451--1478, 2018.

\end{thebibliography}
\renewcommand\thefigure{S\arabic{figure}}\setcounter{figure}{0}
\renewcommand\thetable{S\arabic{table}}\setcounter{table}{0}
\renewcommand\thesection{S\arabic{section}}\setcounter{section}{0}
\renewcommand\thesubsection{S\arabic{section}.\arabic{subsection}}
\renewcommand\theequation{E\arabic{equation}}
\newpage
\section*{Supplementary Materials}
\section{Model-based smoothing}
\label{mb.smooth}
In this section, we describe the model-based smoothing used in the
AM-FAST algorithm. The starting point of our methodology is a SPM
$\bGamma$ and our smoothed estimator is given by $\hat\bGamma = \bH_{\bbeta}\bGamma$
where $\bH_{\bbeta}$ is a smoothing matrix with parameter $\bbeta$. We
write $\bH_{\bbeta}$ = $\bLambda_{\bbeta}(\bLambda_{\bbeta}+\bR)^{-1}$ where
\begin{equation}
  \bLambda_{\bbeta}^{-1} =  \beta_1 \bJ_{n_1} \otimes \bI_{n_1} \otimes \bI_{n_1} + \beta_2 \bI_{n_1}
\otimes \bJ_{n_2} \otimes \bI_{n_2} + 
                  \beta_3 \bI_{n_1} \otimes \bI_{n_2} \otimes
                  \bJ_{n_3}, 
\end{equation}
with  $\otimes$ denoting the Kronecker product between matrices,
$\bI_q$ is a $q\times q$ identity matrix and $\bJ_q$ is a $q\times q$
tridiagonal matrix of the form  
\begin{equation}
\left( \begin{array}{rrrrrr}
1 & -1 & 0 & 0 & 0 & \cdots\\
-1 & 2 & -1 & 0 & 0 & \cdots\\
0 & -1 & 2 & -1 & 0 & \cdots\\
 & & & \ddots & & \\
\cdots & 0 & 0 & -1 & 2 & -1\\
\cdots & 0 & 0 & 0 & -1 & 1\\
\end{array}
\right)\\
\end{equation}
with $i$th eigenvalue given by $j_i = 2-2\cos((i-1)\pi/q)$ for $i=1,2,\ldots,q$.
In the above representation, $\bLambda_{\bbeta}^{-1}$ is split into three
components to represent the three axes of the volume image. Indeed
$\bLambda_{\bbeta}$ represents the inverse of the dispersion matrix of a 
3-D first-order Gaussian Markov Random Field with a neighborhood
structure for each interior voxel given by the six nearest neighbors
(one nearest in each of the two directions in each plane). Also
$\bbeta = (\beta_1,\beta_2,\beta_3)$ with $\beta_i$ measuring the
strength of the interaction between the neighbors in the $i$th plane.
For the edge voxels, the neighborhood
structure is similar and given by the nearest neighboring voxels in
each of the two directions in each plane, provided they exist in the
imaging grid.

Our smoothed SPM is  a penalized maximum likelihood estimator and is
the same as the {\em maximum a posterioiri} (MAP) or 
posterior mean estimator of $\bmu$ when we have $\bGamma\mid\bmu \sim
N(\bmu,\bR)$ and $\bmu\sim N(\bzero,\bLambda_{\bbeta})$. So, the
regularization parameter $\bbeta$ may be chosen using an empirical
Bayes estimate. This estimate is given by
\begin{equation}
\hat\bbeta = \argmin_{\bbeta} \left[\log\abs{\bLambda_{\bbeta} + \bR}
  + \bGamma'(\bLambda_{\bbeta} + \bR)^{-1}\bGamma\right].
\label{betamax}
\end{equation}
We consider $\bR = \sigma^2\bI_{n_1,n_2,n_2}
\equiv\sigma^2\bI_{n_1}\otimes\bI_{n_2}\otimes\bI_{n_3}$ in our application.
 Then the eigenvalues of
$\bLambda_{\bbeta}+\sigma^2\bI$ are all real and given by
\begin{equation}
\sigma^2 + \left\{\beta_1\left[2-2\cos\left(\frac{(i-1)\pi}{n_1}\right)\right] +
\beta_2 \left[2-2\cos\left(\frac{(j-1)\pi}{n_2}\right)\right]  +
\beta_2 \left[2-2\cos\left(\frac{(k-1)\pi}{n_3}\right)\right]\right\}^{-1}
\label{eigens}
\end{equation}
Let $\bC_q'$ be the $q\times q$ Type-2 1D Discrete Cosine Transform
(DCT) matrices and $\bC_q$ be the inverse DCT. Then $\bC'=\bC'_{n_1}\!\otimes
\bC'_{n_2}\!\otimes\!\bC'_{n_3}$ is the Type-2 DCT in the 3D volume of
$n_1\times n_2\times n_3$ voxels and $\bC=\bC_{n_1}\!\otimes
\bC_{n_2}\!\otimes\!\bC_{n_3}$ is the inverse 3D DCT. Then
$\bW = \bC'\bGamma$ is the 3D DCT of $\bGamma$ with entries
$w_{ijk}$. Then \eqref{betamax} reduces to
\begin{equation}
  \begin{split}
    \hat\bbeta  =  \argmax_{\sigma,\bbeta}
    \sum_{i,j,k} &\left[ \log\left(\sigma^2 + \left\{\beta_1\left[2-2\cos\left(\frac{(i-1)\pi}{n_1}\right)\right] +
\beta_2 \left[2-2\cos\left(\frac{(j-1)\pi}{n_2}\right)\right]  +
\beta_2 \left[2-2\cos\left(\frac{(k-1)\pi}{n_3}\right)\right]\right\}^{-1}
\right)\right.\\
&\qquad\qquad\left.
+ \frac{w^2_{ijk}}{\sigma^2 + \left\{\beta_1\left[2-2\cos\left(\frac{(i-1)\pi}{n_1}\right)\right] +
\beta_2 \left[2-2\cos\left(\frac{(j-1)\pi}{n_2}\right)\right]  +
\beta_2 \left[2-2\cos\left(\frac{(k-1)\pi}{n_3}\right)\right]\right\}^{-1}
}
\right].
\end{split}
\end{equation}
Further, the smoothed estimator $\bGamma$ can be obtained by the inverse DCT 
of $(\bL -\sigma^2\bI_n)\bL^{-1}\bW$, where $\bL$ is the diagonal matrix with
the diagonal element corresponding to the $(i,j,k)$th voxel given by
\eqref{eigens}. 
\begin{figure}[h]
  \vspace{-0.3in}
\begin{center}
  \mbox{\subfloat{\includegraphics[width=0.14\textwidth]{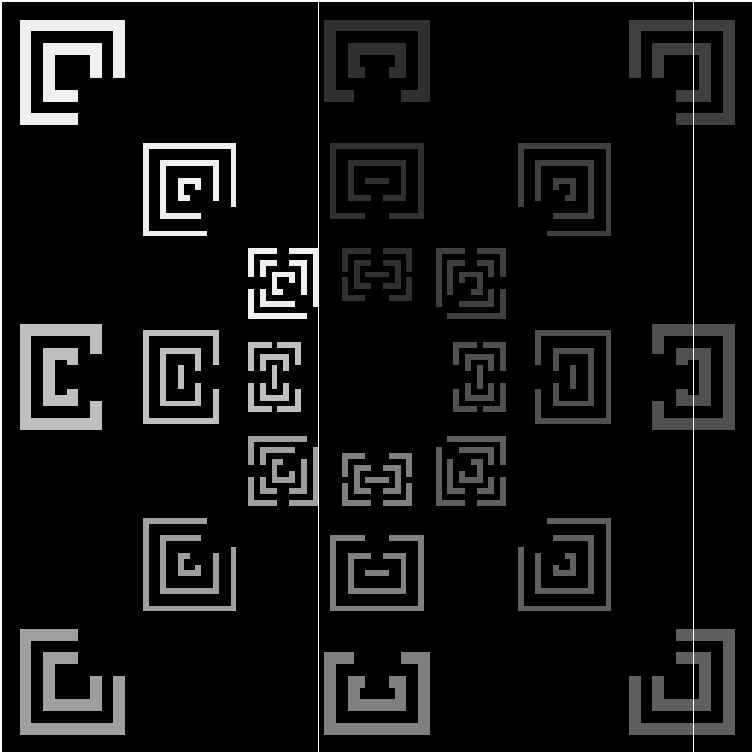}}
    \hspace{0.001in}
    \subfloat{\includegraphics[width=0.14\textwidth]{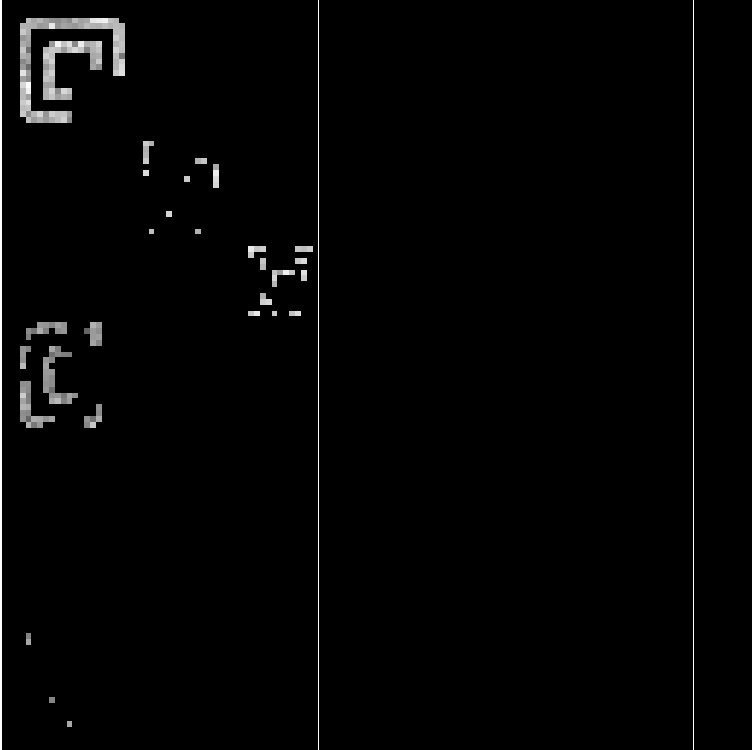}}
    \hspace{0.001in}
    \subfloat{\includegraphics[width=0.14\textwidth]{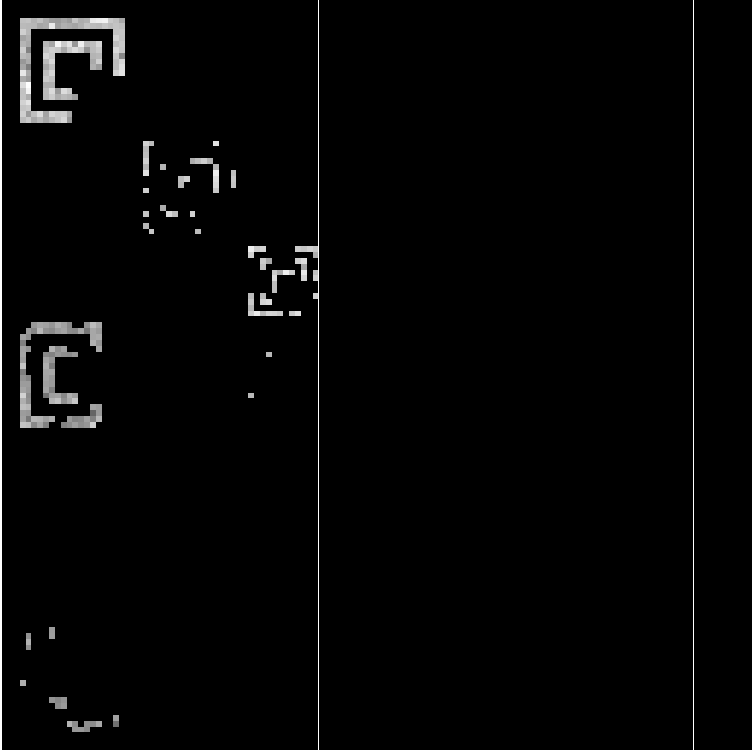}}
    \hspace{0.001in}
    \subfloat{\includegraphics[width=0.14\textwidth]{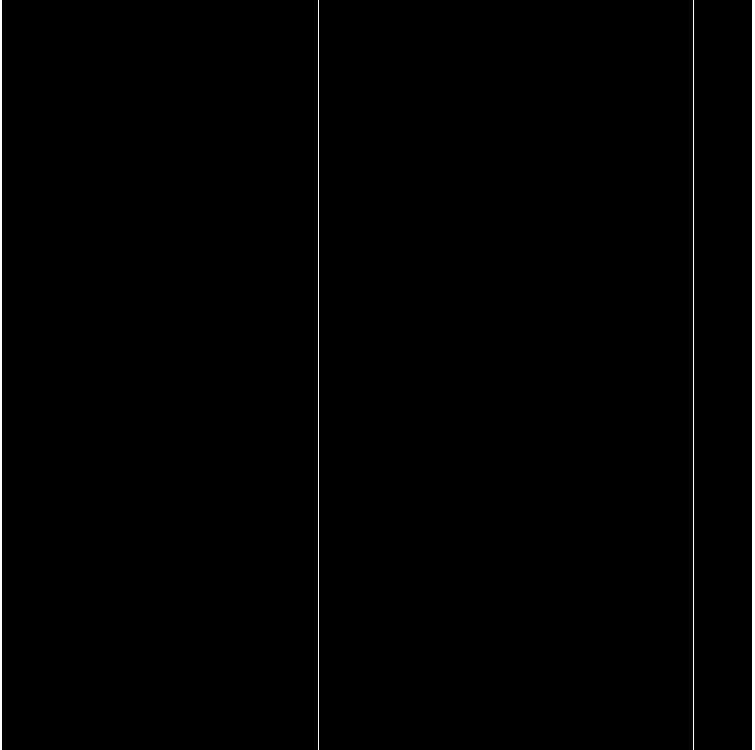}}
    \hspace{0.001in}
    \subfloat{\includegraphics[width=0.14\textwidth]{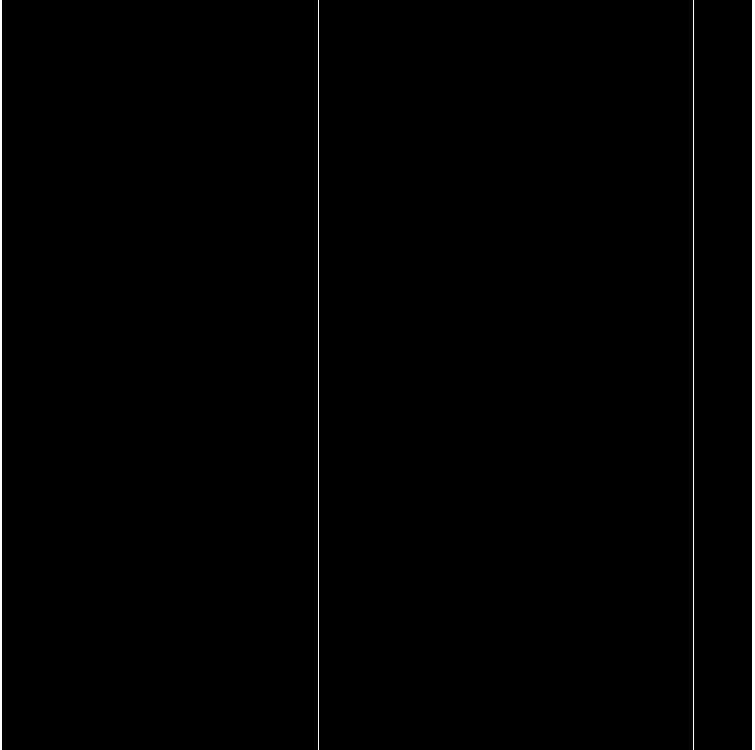}}
    \hspace{0.001in}
    \subfloat{\includegraphics[width=0.14\textwidth]{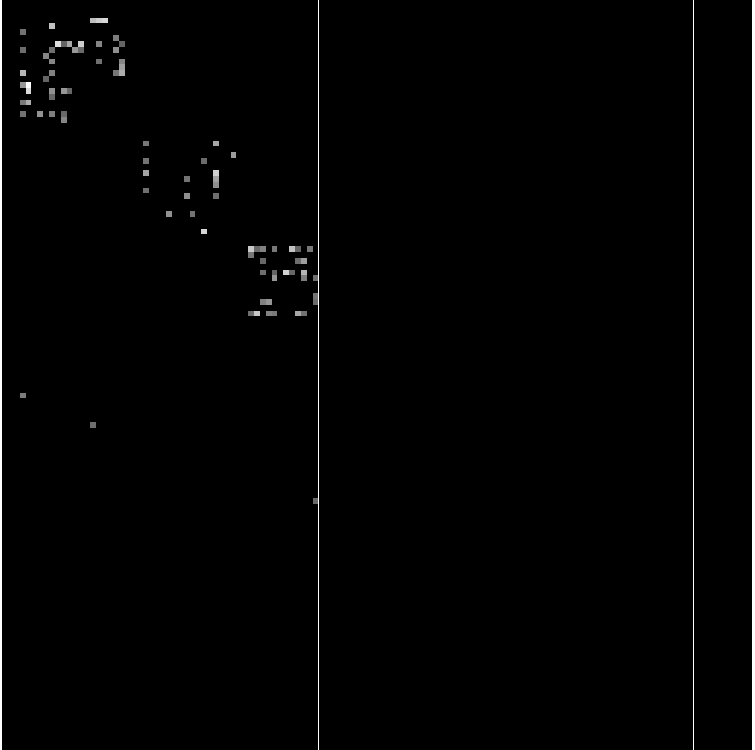}}
  }
  \mbox{    
    \subfloat{\includegraphics[width=0.14\textwidth]{figures/Motif-2-AR-FAST-Alpha-005-128_x_128}}
    \hspace{0.001in}
    \subfloat{\includegraphics[width=0.14\textwidth]{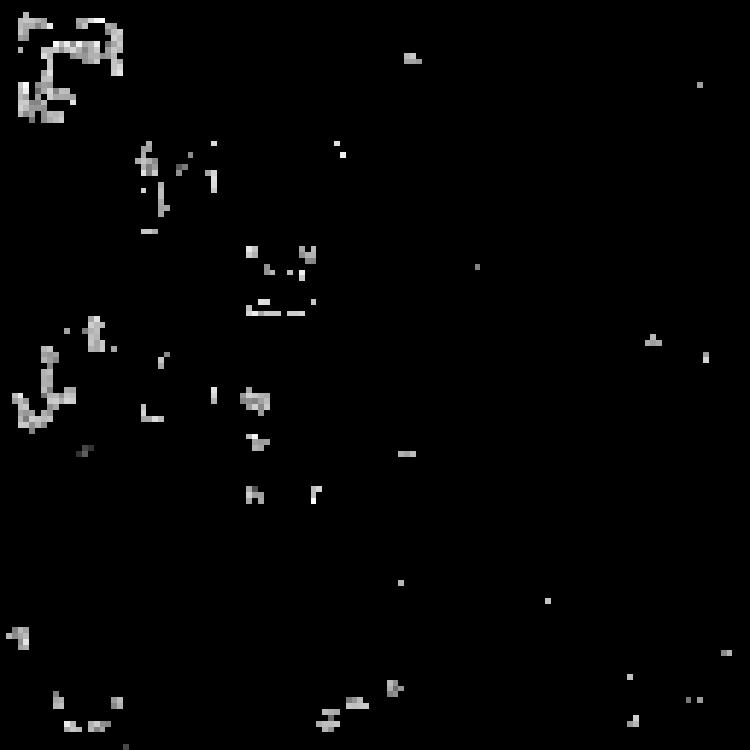}}
    \hspace{0.001in}
    \subfloat{\includegraphics[width=0.14\textwidth]{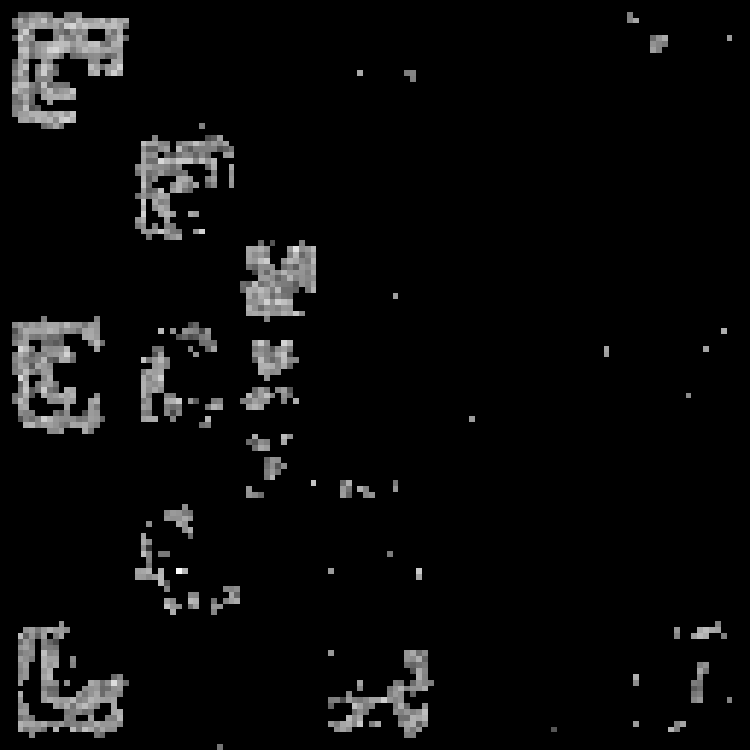}}
    \hspace{0.001in}
    \subfloat{\includegraphics[width=0.14\textwidth]{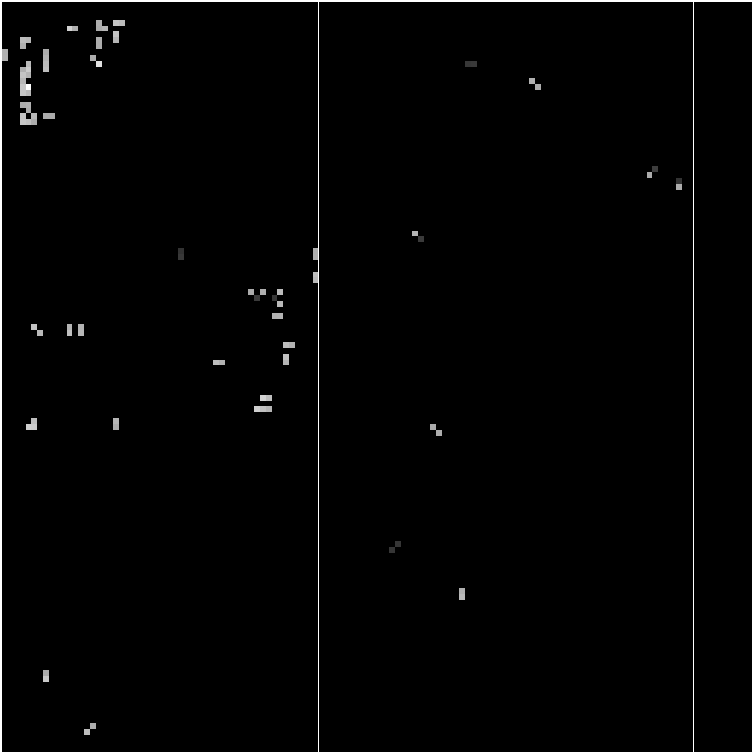}}
    \hspace{0.001in}
    \subfloat{\includegraphics[width=0.14\textwidth]{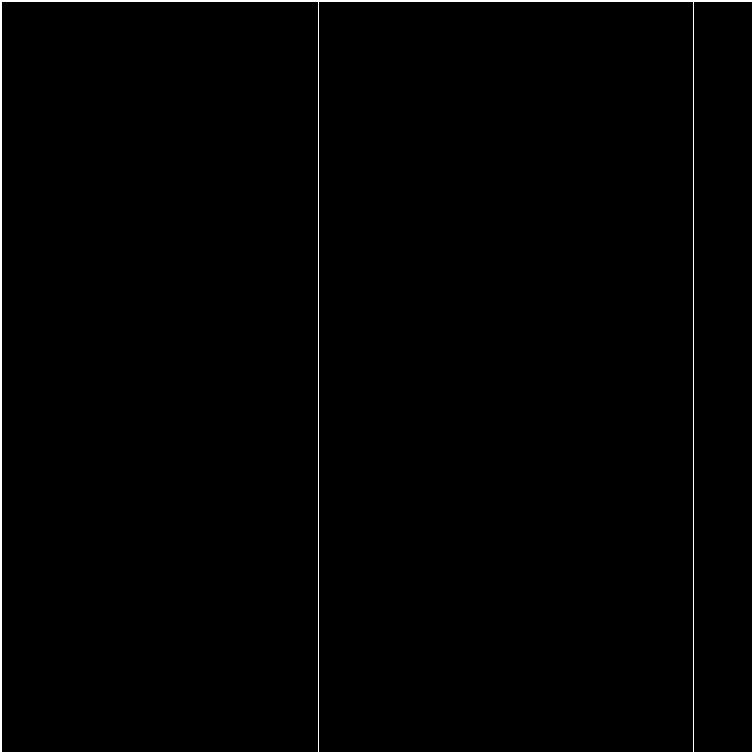}}
    \hspace{0.001in}
    \subfloat{\includegraphics[width=0.14\textwidth]{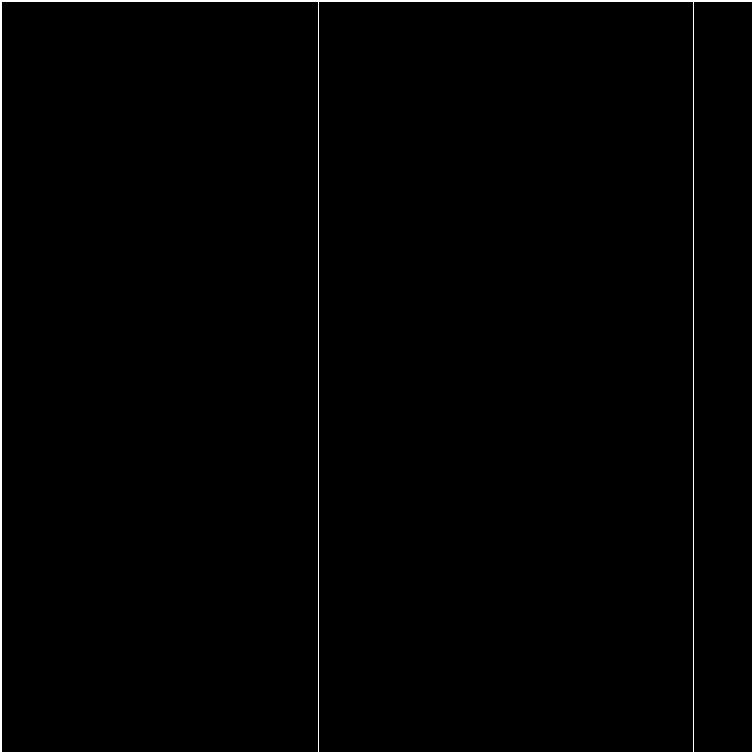}}}
  \vspace{0.02in}  \hrule
  \mbox{\subfloat{\includegraphics[width=0.14\textwidth]{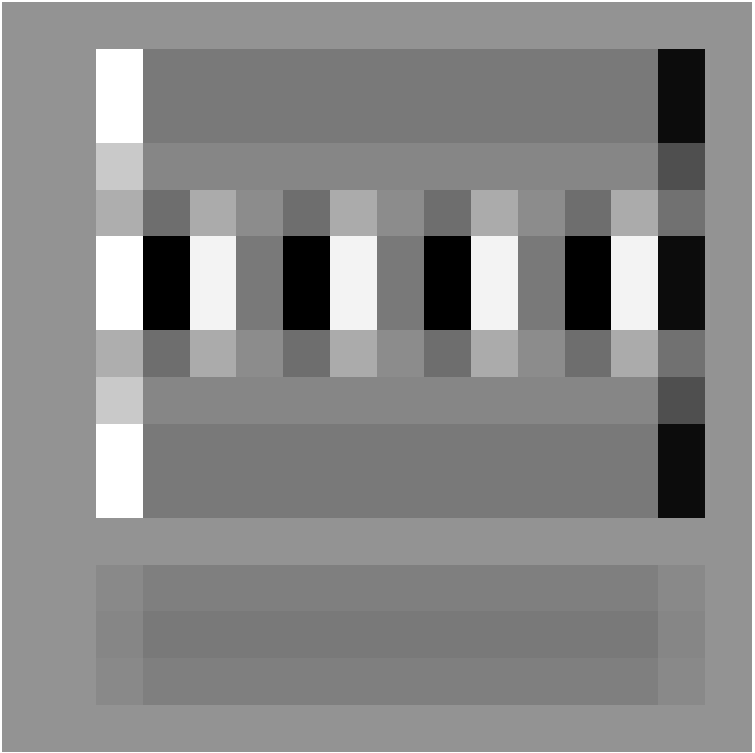}}
    \hspace{0.001in}
    \subfloat{\includegraphics[width=0.14\textwidth]{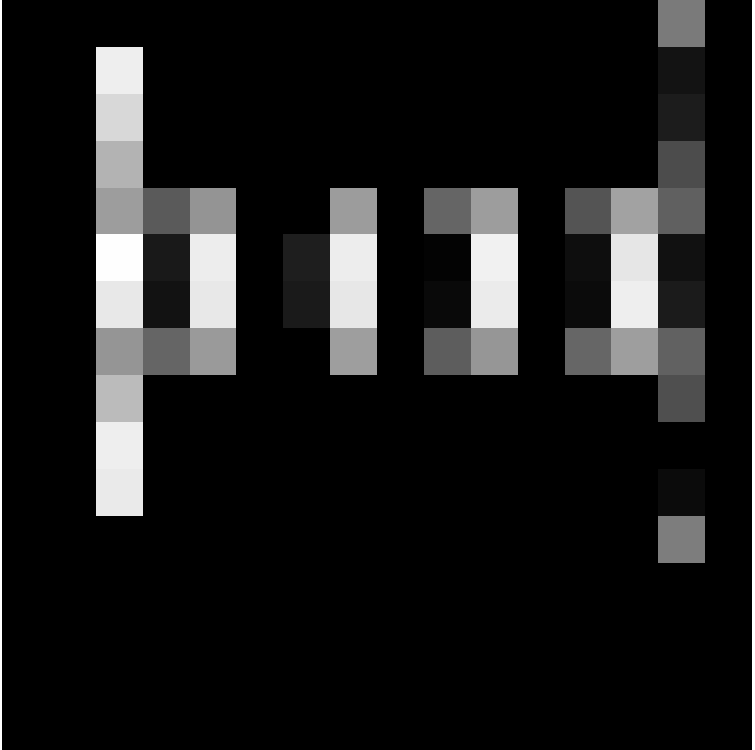}}
    \hspace{0.001in}
    \subfloat{\includegraphics[width=0.14\textwidth]{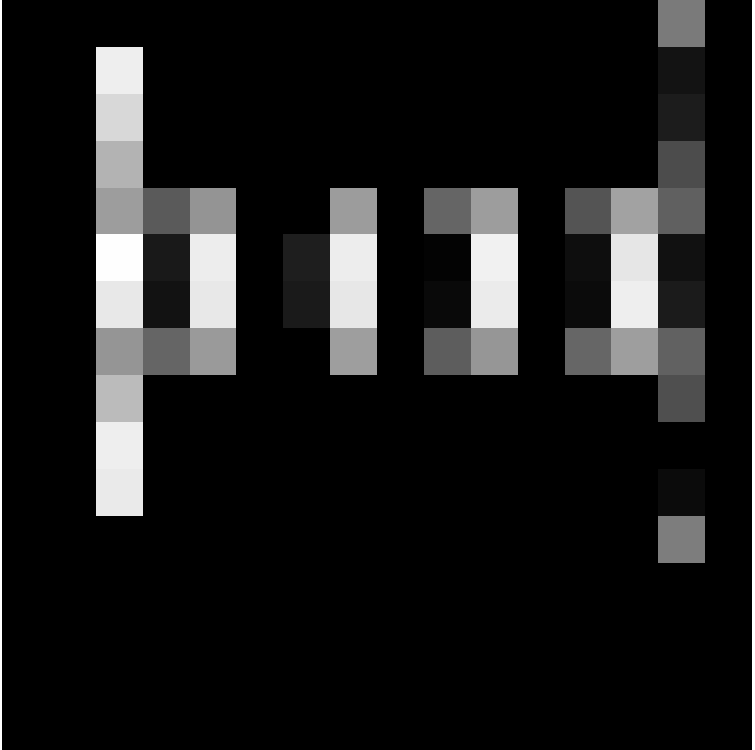}}
    \hspace{0.001in}
    \subfloat{\includegraphics[width=0.14\textwidth]{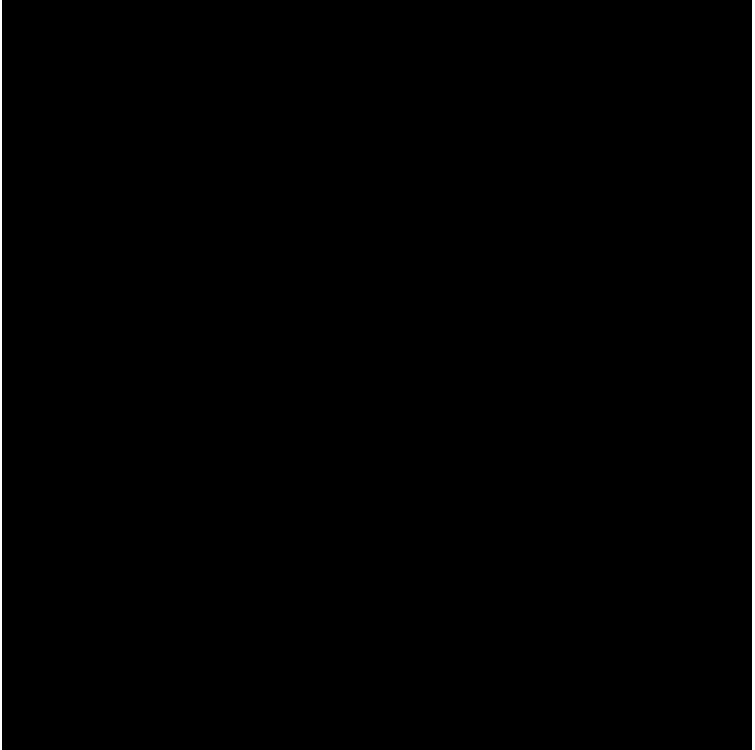}}
    \hspace{0.001in}
    \subfloat{\includegraphics[width=0.14\textwidth]{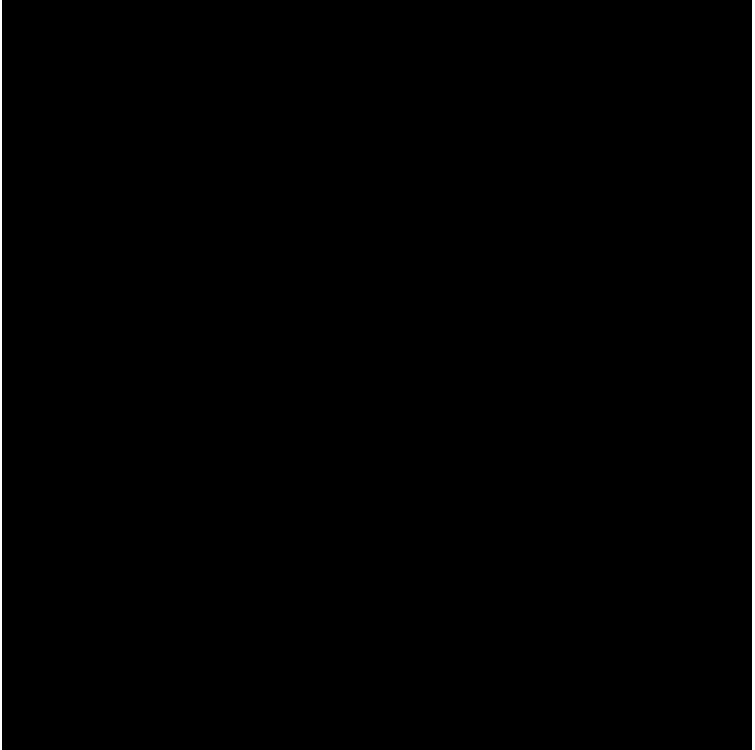}}
    \hspace{0.001in}
    \subfloat{\includegraphics[width=0.14\textwidth]{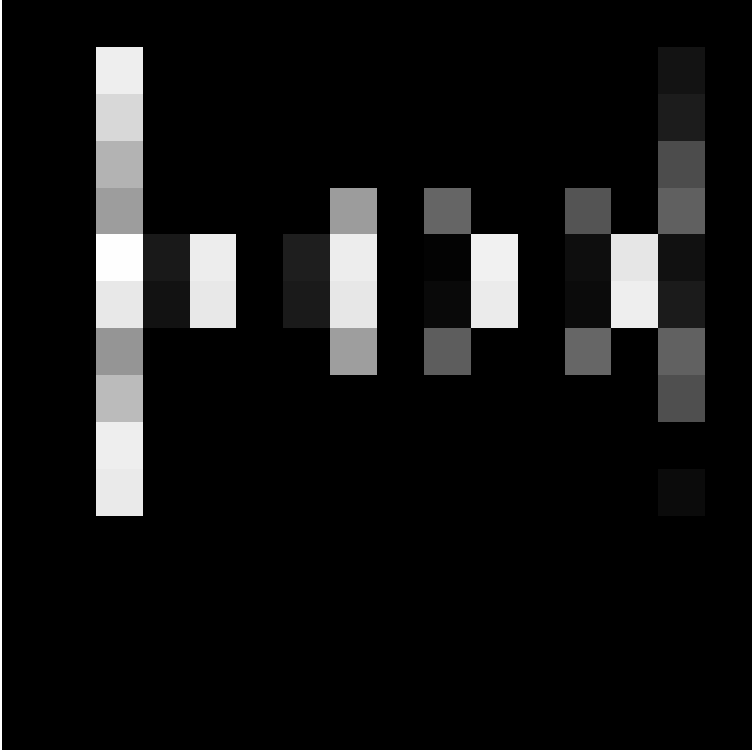}}}
  \mbox{
    \subfloat{\includegraphics[width=0.14\textwidth]{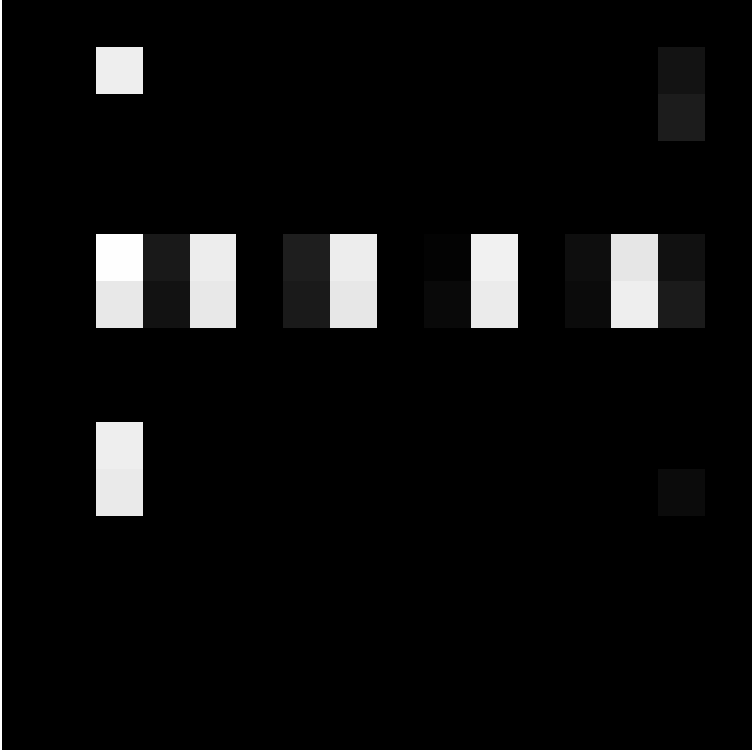}}
    \hspace{0.001in}
    \subfloat{\includegraphics[width=0.14\textwidth]{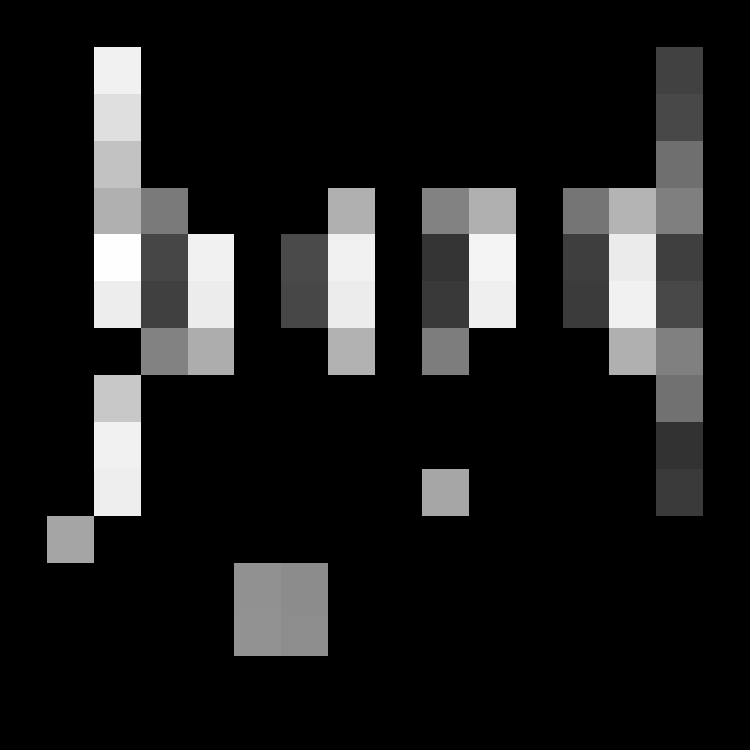}}
    \hspace{0.001in}
    \subfloat{\includegraphics[width=0.14\textwidth]{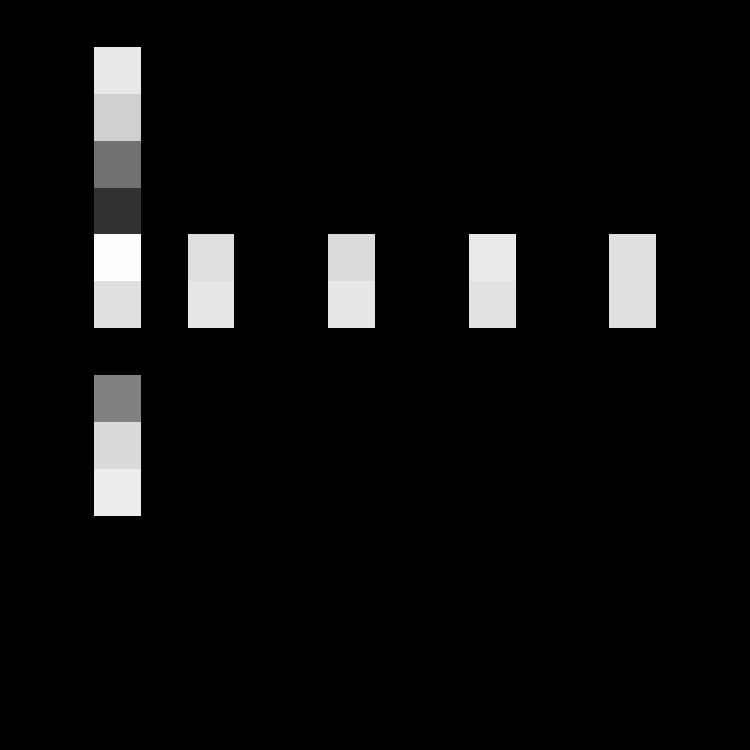}}
    \hspace{0.001in}
    \subfloat{\includegraphics[width=0.14\textwidth]{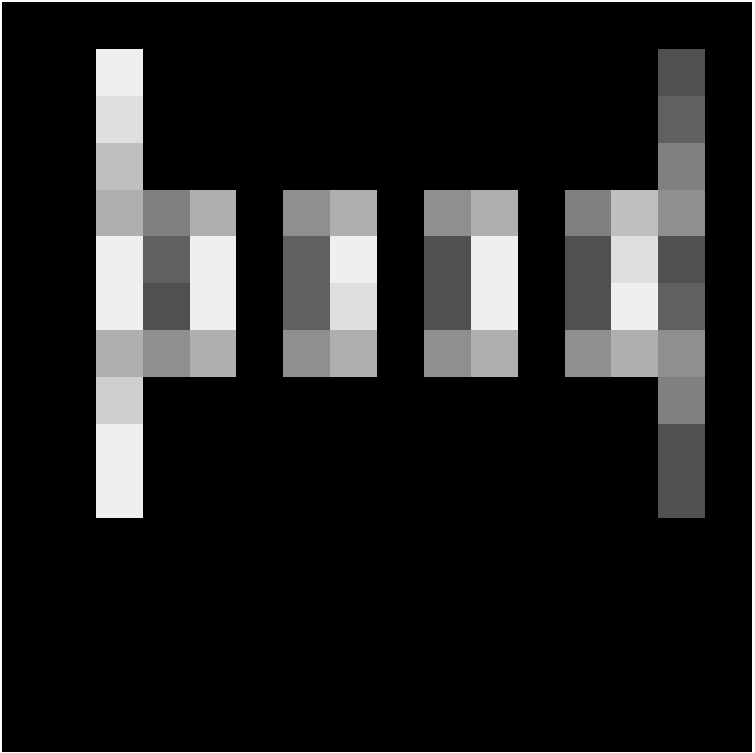}}
    \hspace{0.001in}
    \subfloat{\includegraphics[width=0.14\textwidth]{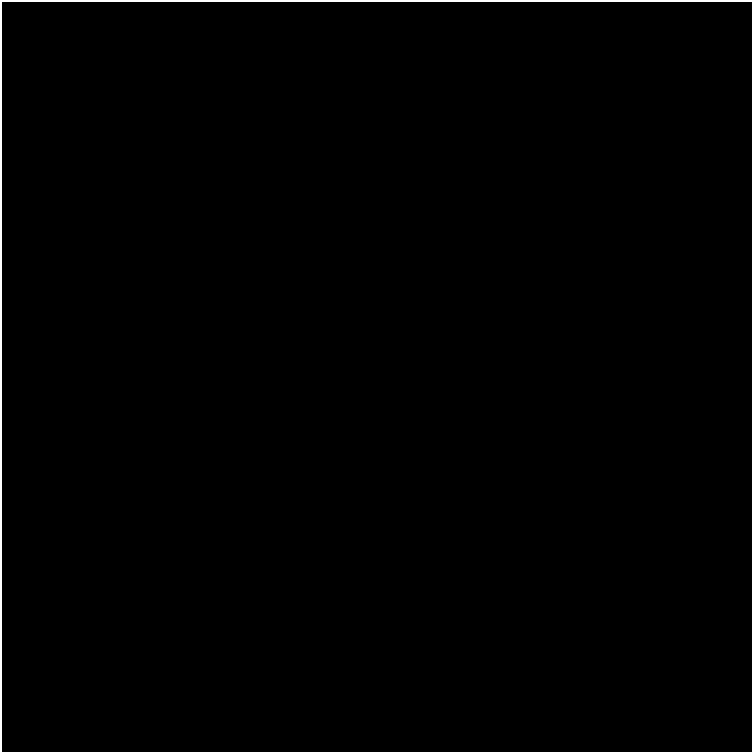}}
    \hspace{0.001in}
    \subfloat{\includegraphics[width=0.14\textwidth]{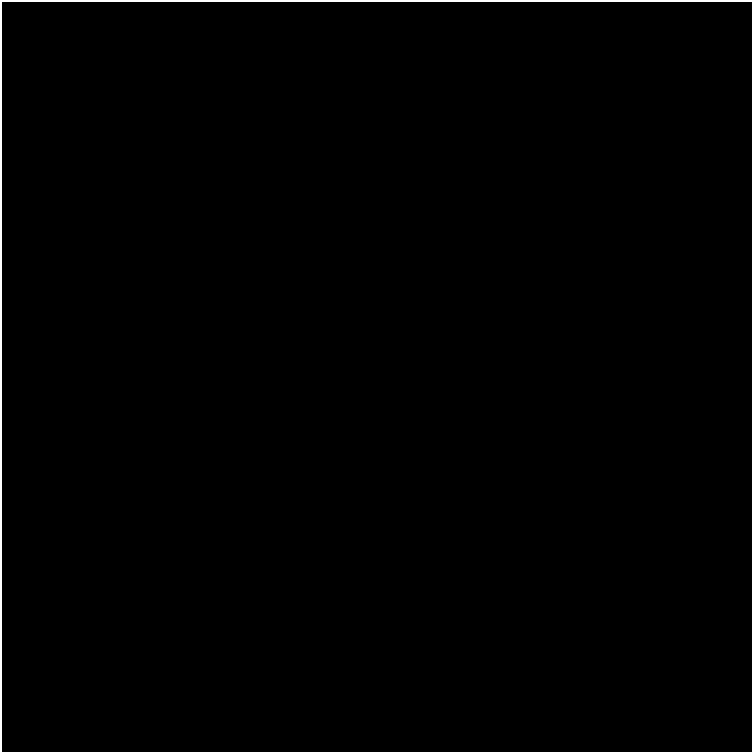}}}
  \vspace{0.02in}  \hrule
  \mbox{\subfloat{\includegraphics[width=0.14\textwidth]{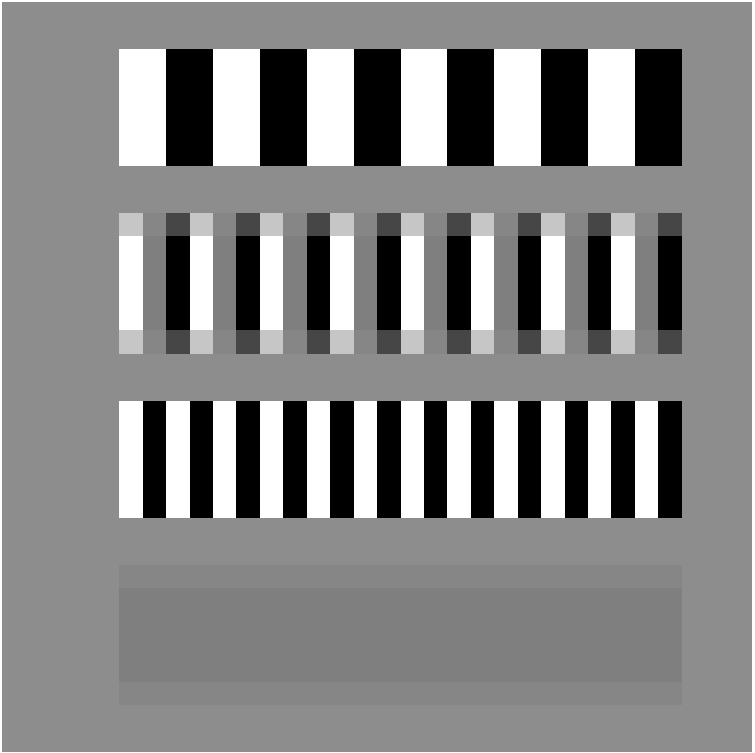}}
    \hspace{0.001in}
    \subfloat{\includegraphics[width=0.14\textwidth]{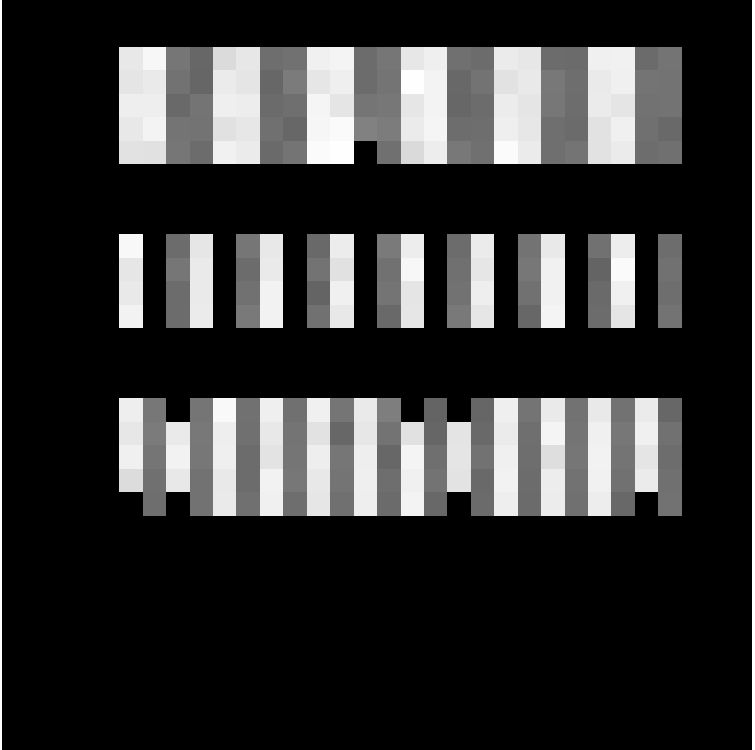}}
    \hspace{0.001in}
    \subfloat{\includegraphics[width=0.14\textwidth]{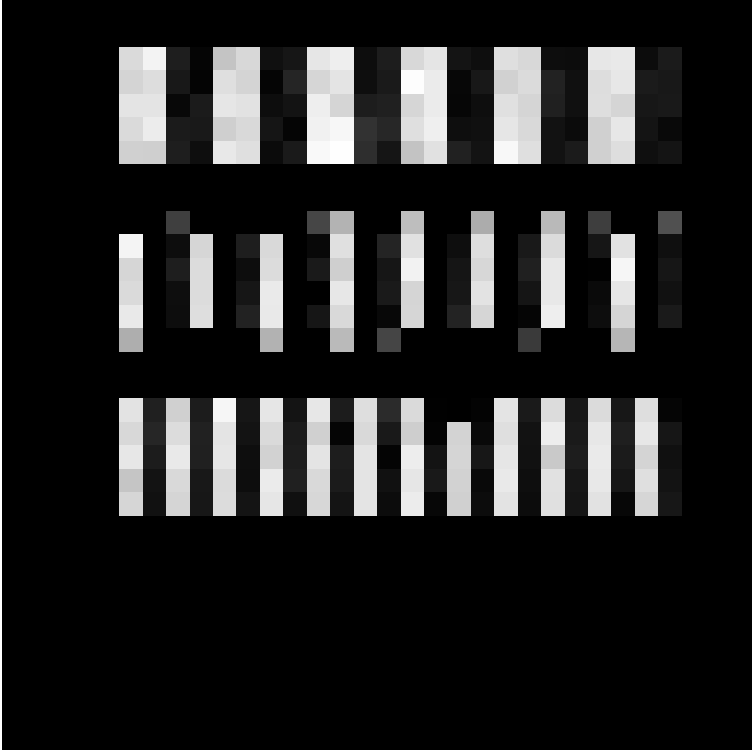}}
     \hspace{0.001in}
   \subfloat{\includegraphics[width=0.14\textwidth]{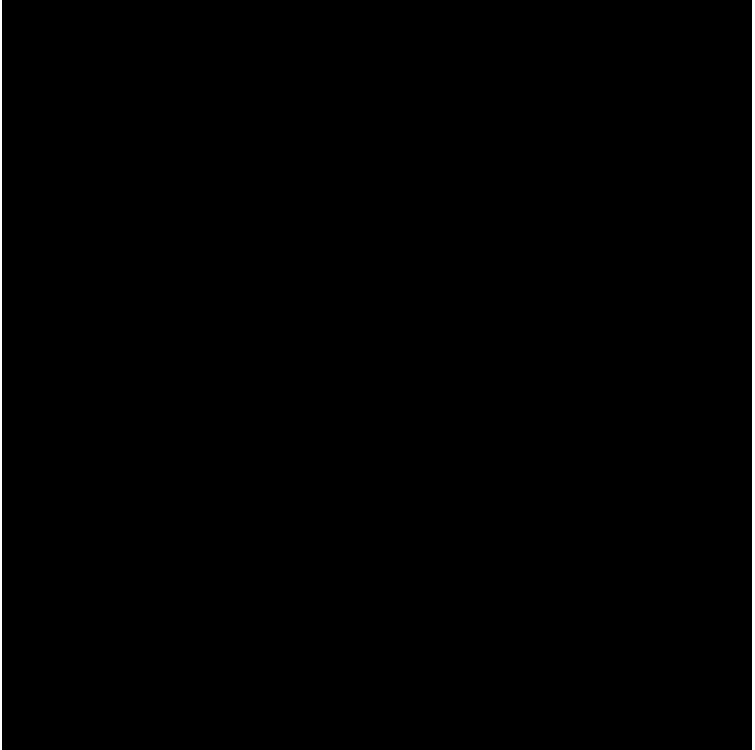}}
    \hspace{0.001in}
    \subfloat{\includegraphics[width=0.14\textwidth]{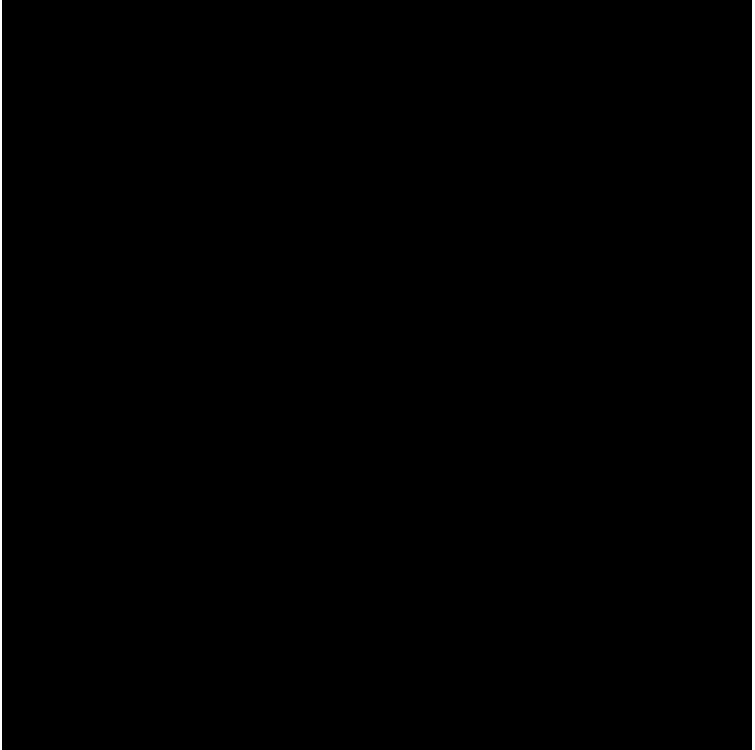}}
    \hspace{0.001in}
    \subfloat{\includegraphics[width=0.14\textwidth]{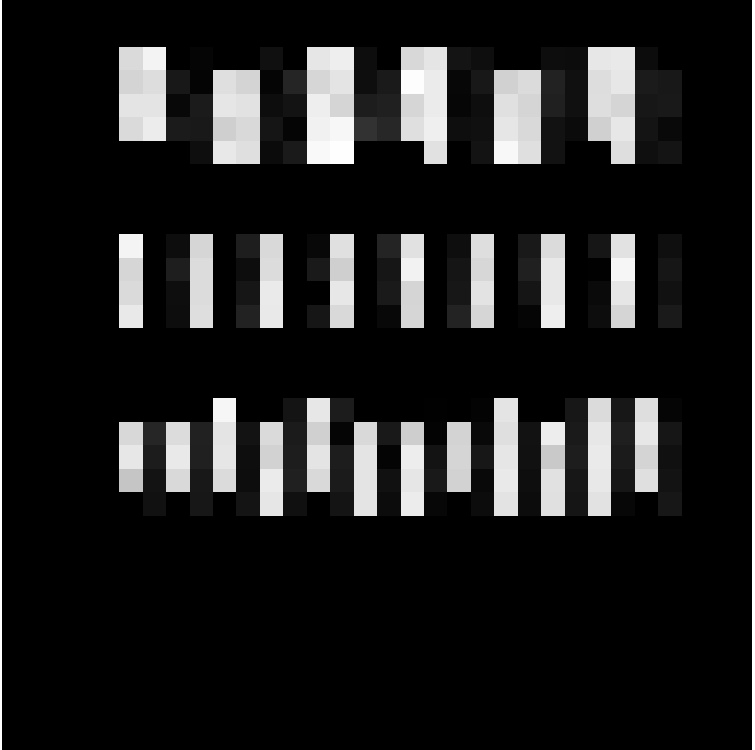}}}
  \mbox{
    \subfloat{\includegraphics[width=0.14\textwidth]{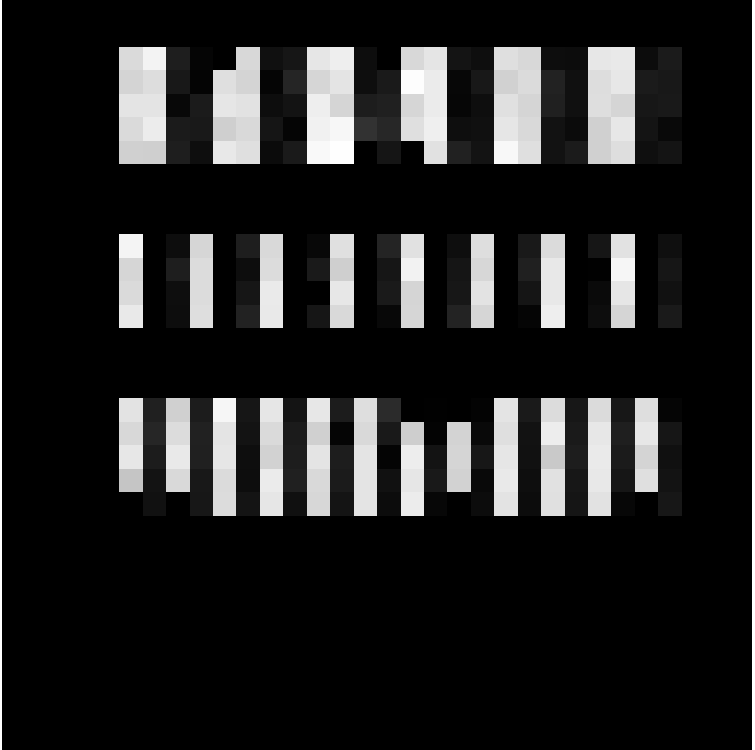}}
    \hspace{0.001in}
    \subfloat{\includegraphics[width=0.14\textwidth]{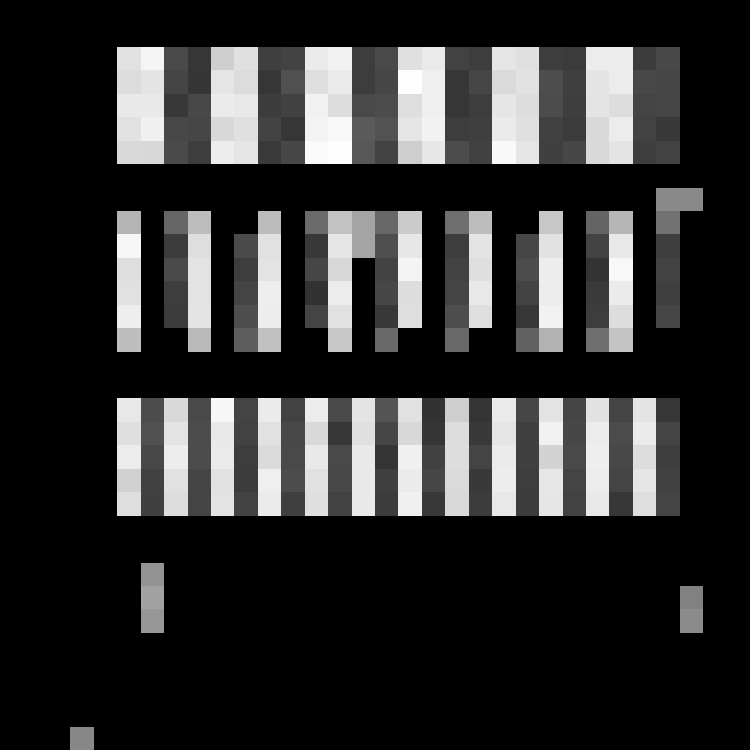}}
    \hspace{0.001in}
    \subfloat{\includegraphics[width=0.14\textwidth]{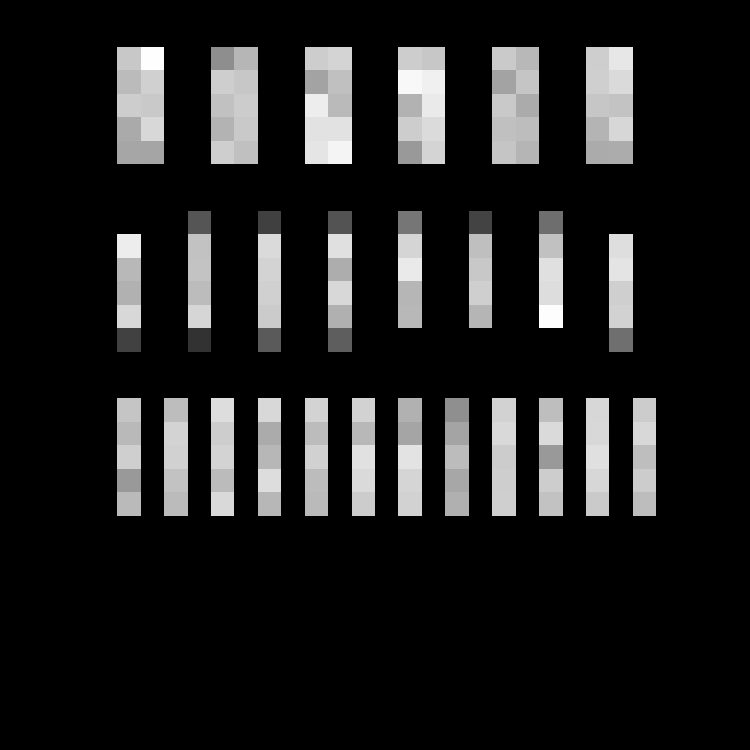}}
    \hspace{0.001in}
    \subfloat{\includegraphics[width=0.14\textwidth]{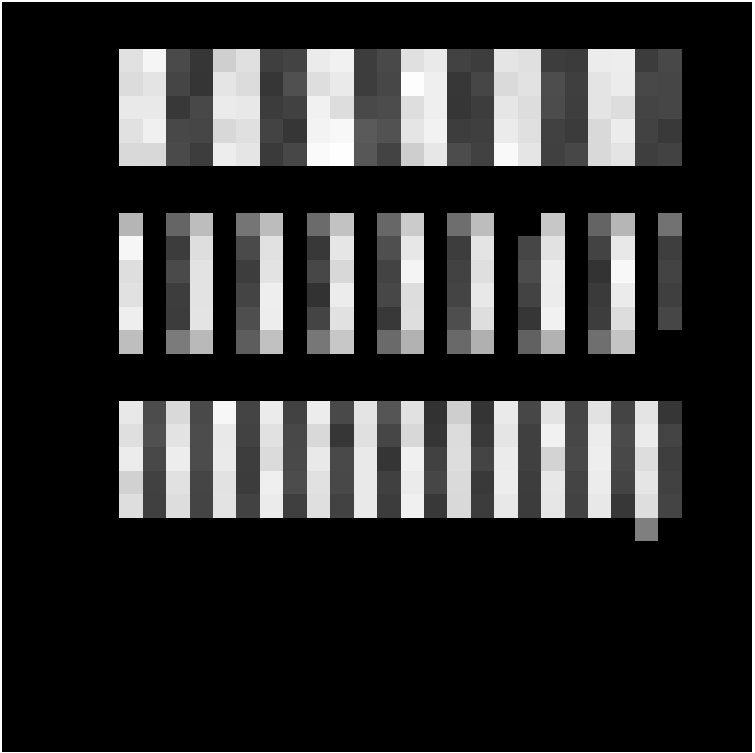}}
    \hspace{0.001in}
    \subfloat{\includegraphics[width=0.14\textwidth]{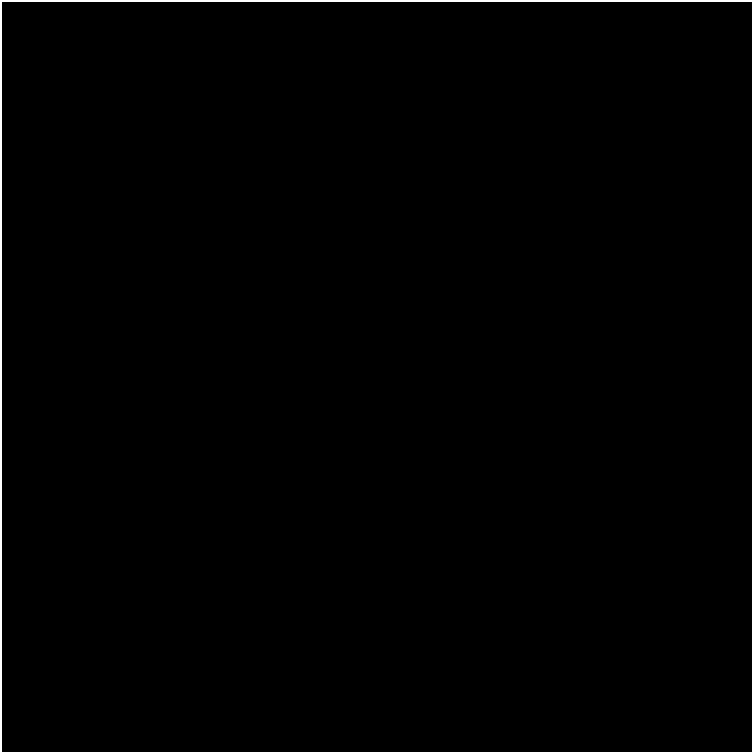}}
    \hspace{0.001in}
    \subfloat{\includegraphics[width=0.14\textwidth]{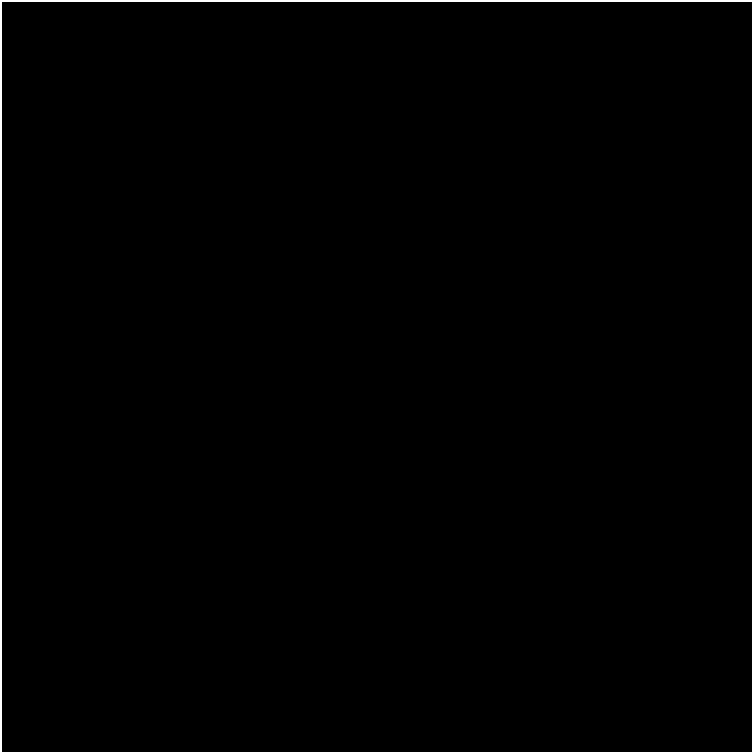}}
  }
  \vspace{0.02in}  \hrule
  \mbox{\subfloat{\includegraphics[width=0.14\textwidth]{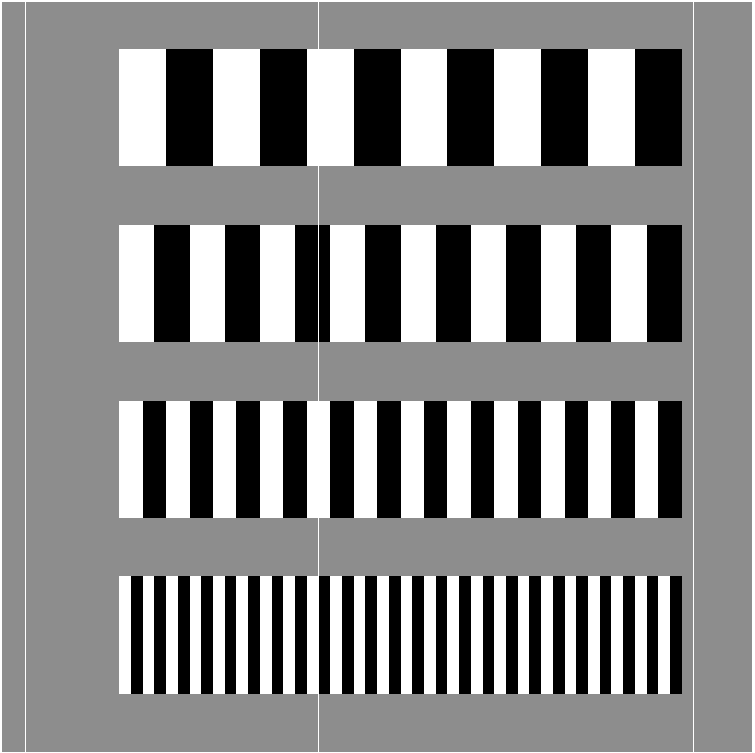}}
    \hspace{0.001in}
    \subfloat{\includegraphics[width=0.14\textwidth]{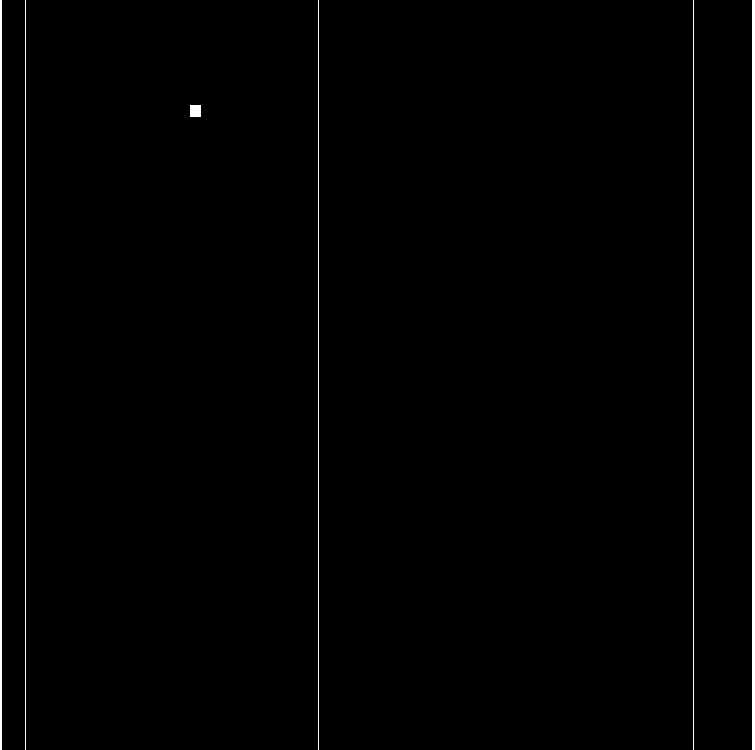}}
    \hspace{0.001in}
    \subfloat{\includegraphics[width=0.14\textwidth]{figures/AL-FAST-Alpha-001-64_x_64}}
    \hspace{0.001in}
    \subfloat{\includegraphics[width=0.14\textwidth]{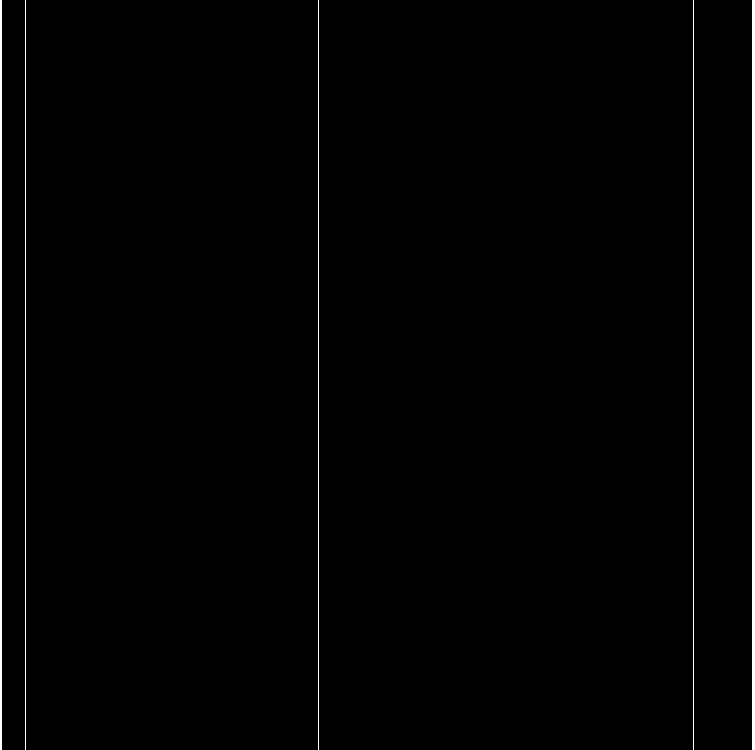}}
    \hspace{0.001in}
    \subfloat{\includegraphics[width=0.14\textwidth]{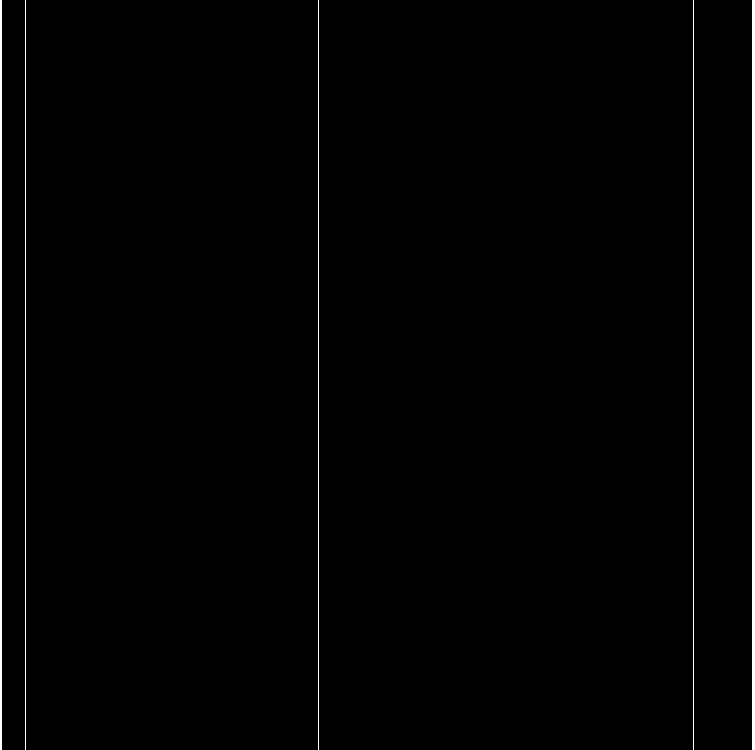}}
    \hspace{0.001in}
    \subfloat{\includegraphics[width=0.14\textwidth]{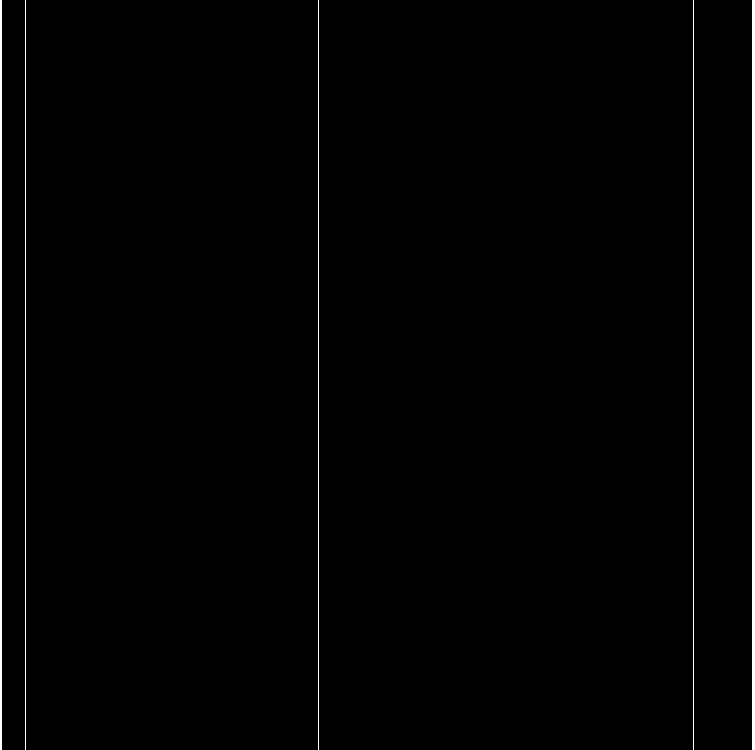}}}
  \vspace{0.02in}  \hrule
  \mbox{
    \subfloat{\includegraphics[width=0.14\textwidth]{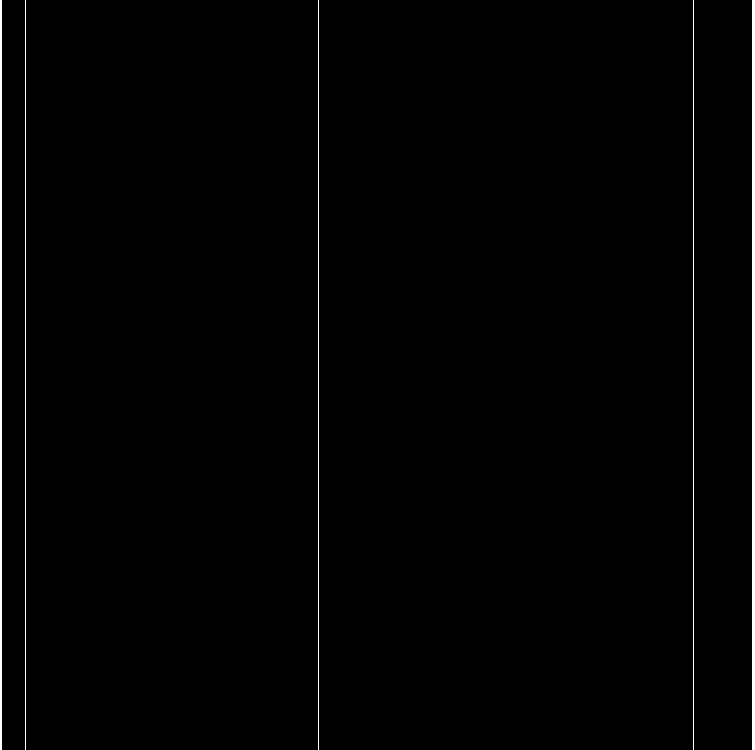}}
    \hspace{0.001in}
    \subfloat{\includegraphics[width=0.14\textwidth]{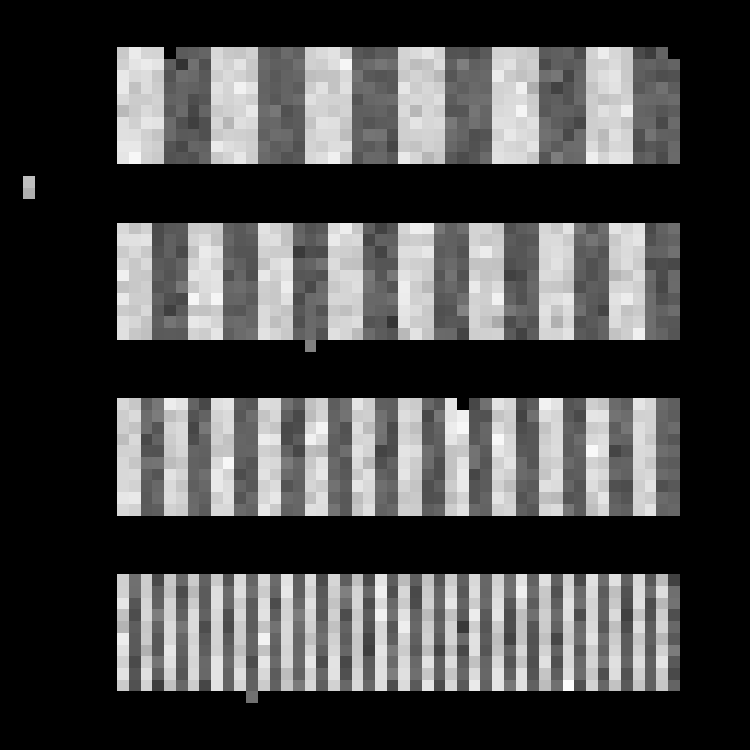}}
    \hspace{0.001in}
    \subfloat{\includegraphics[width=0.14\textwidth]{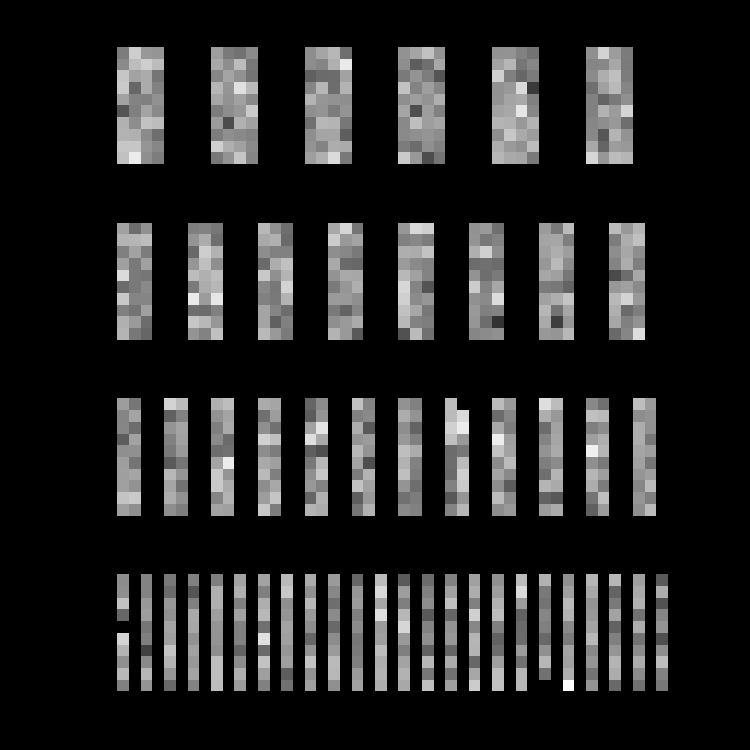}}
    \hspace{0.001in}
    \subfloat{\includegraphics[width=0.14\textwidth]{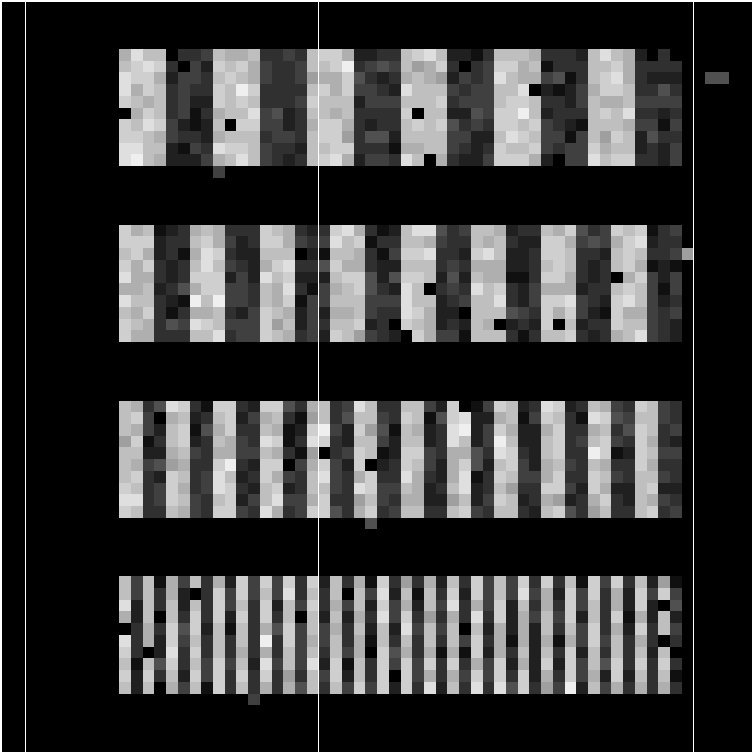}}
    \hspace{0.001in}
    \subfloat{\includegraphics[width=0.14\textwidth]{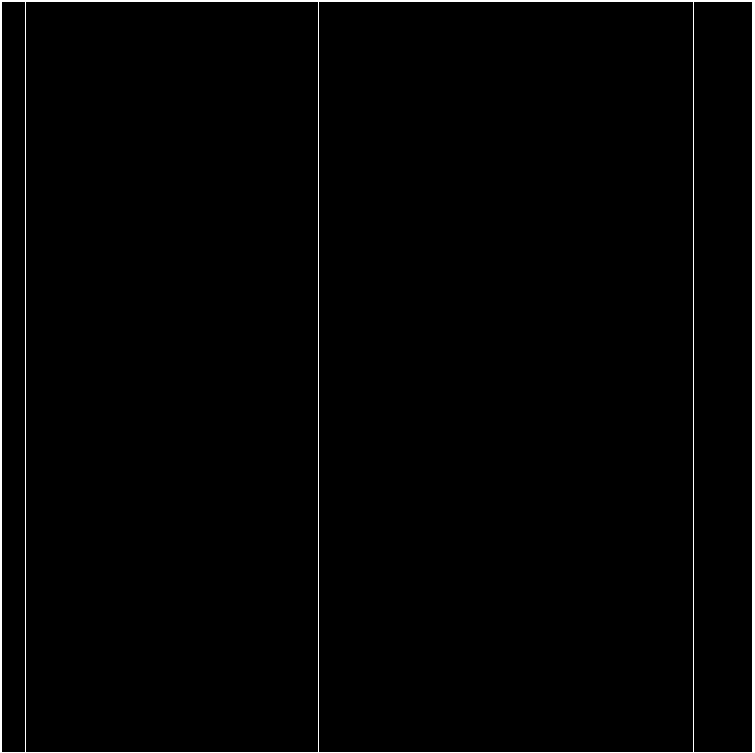}}
    \hspace{0.001in}
    \subfloat{\includegraphics[width=0.14\textwidth]{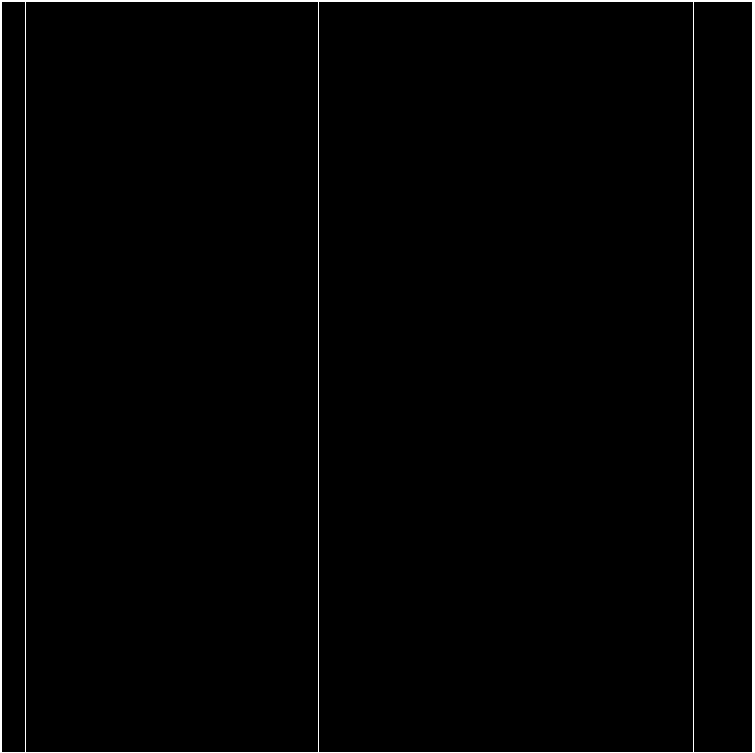}}}
  \end{center}
\caption{True and activated voxels in the SPM upon using ALL-FAST with
  $\alpha=0.01$ and 0.05, AM-FAST with
  $\alpha=0.01$ and 0.05, AR-FAST with $\alpha=0.01$ and 0.05, AS,
  AWS, CT, PBT and TFCE for the Motif (top row) and the $16\times16$, $32\times32$ and
  $64\times64$ stripes (last three rows).}
\label{fig:tabelow-all}
\end{figure}
\begin{figure}
  \mbox{
    \subfloat[AR(0)]{\includegraphics[height=0.25\textheight,width=\textwidth]{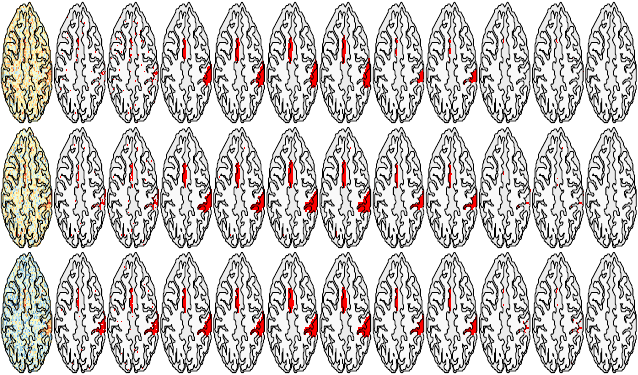}}}
  \mbox{
  \subfloat[AR(1)]{\includegraphics[height=0.25\textheight,width=\textwidth]{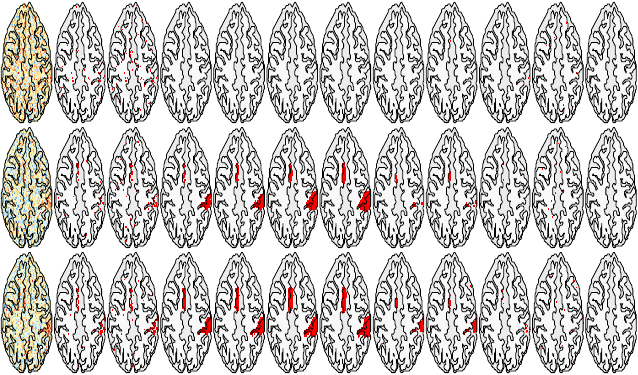}}}
    \mbox{
      \subfloat[AR(2)]{\includegraphics[height=0.25\textheight,width=\textwidth]{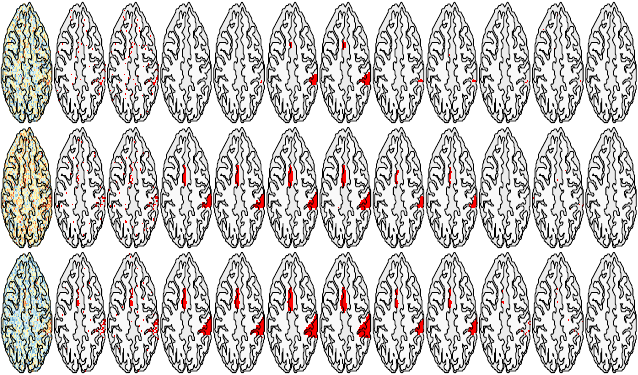}}}
      \caption{Sample SPMs (left column) and corresponding activation maps from the modified Hoffman
        phantom obtained for CNR = 0.5 (top), 0.75 (middle) and 1.0
        (bottom) using (from left to right) ALL-FAST
         with        $\alpha = 0.01$ and $\alpha = 0.05$,          AM-FAST with
        $\alpha = 0.01$ and $\alpha = 0.05$, AR-FAST with
        $\alpha = 0.01$ and $\alpha = 0.05$, AS, AWS, CT, PBT and TFCE
        for AR($p$) errors with (a) $p=0$, (b) $p=1$ and (c) $p=2$
        with AR coefficients decreasing with autocorrelation order.}\label{fig:ActMapDec1}
    \end{figure}
    \begin{figure}
    \mbox{
      \subfloat[AR(3)]{\includegraphics[height=0.25\textheight,width=\textwidth]{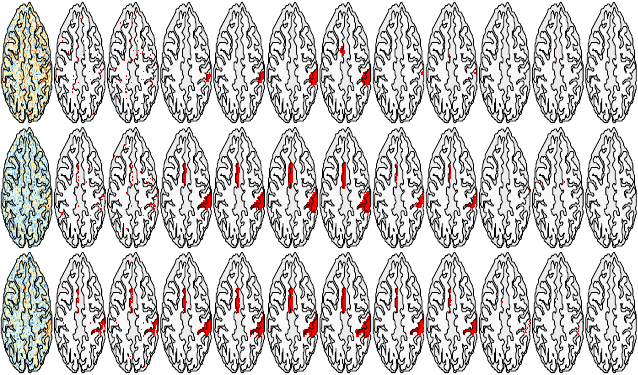}}}
    \mbox{
      \subfloat[AR(4)]{\includegraphics[height=0.25\textheight,width=\textwidth]{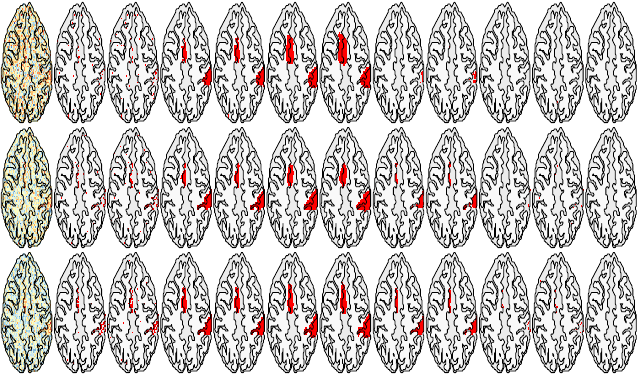}}}
    \mbox{      
      \subfloat[AR(5)]{\includegraphics[height=0.25\textheight,width=\textwidth]{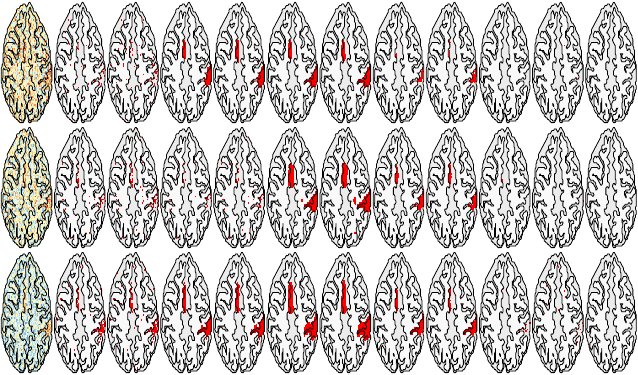}}}
    \caption{SPMs and corresponding activation maps in a similar framework as
      Figure~\ref{fig:ActMapDec1}, but for $p$: we have (a) $p=3$, (b)
      $p=4$ and (c) $p=5$.}
      \label{fig:ActMapDec2}
    \end{figure}

    \begin{figure}
      \mbox{
        \subfloat[AR(0)]{\includegraphics[height=0.25\textheight,width=\textwidth]{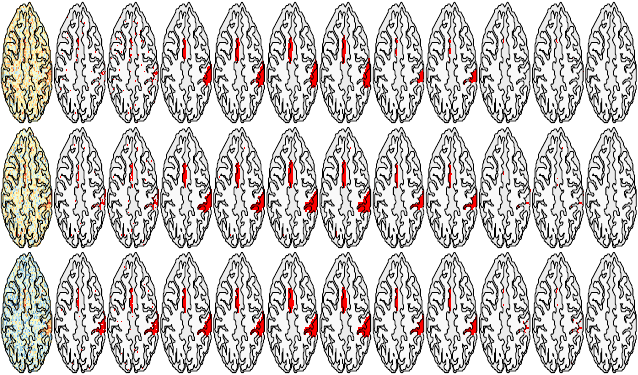}}
      }
      \mbox{
      \subfloat[AR(1)]{\includegraphics[height=0.25\textheight,width=\textwidth]{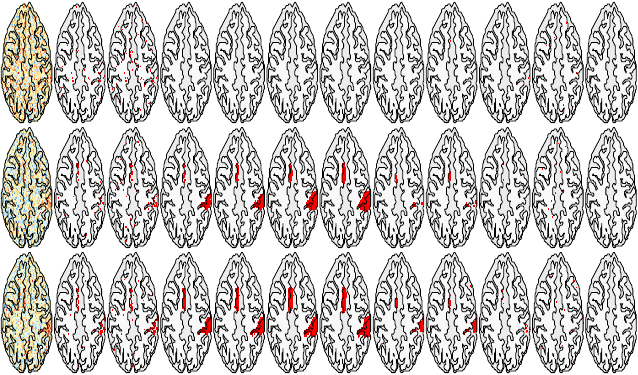}}}
    \mbox{
      \subfloat[AR(2)]{\includegraphics[height=0.25\textheight,width=\textwidth]{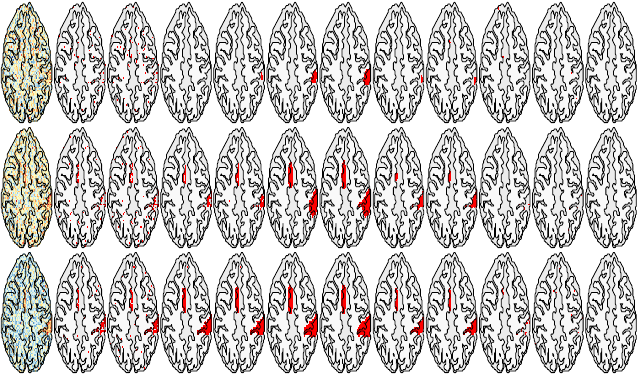}}}
    \caption{Sample SPMs and corresponding activation maps in a similar framework as
      Figure~\ref{fig:ActMapDec1}, but here the AR coefficients are in
      increasing order}.
    \label{fig:ActMapInc1}
    \end{figure}
    \begin{figure}
     \mbox{
      \subfloat[AR(3)]{\includegraphics[height=0.25\textheight,width=\textwidth]{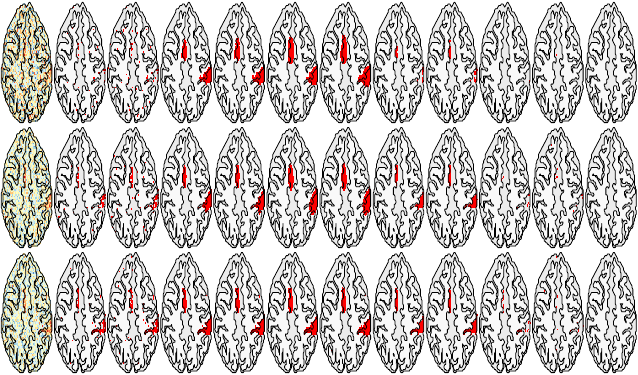}}}
    \mbox{
      \subfloat[AR(4)]{\includegraphics[height=0.25\textheight,width=\textwidth]{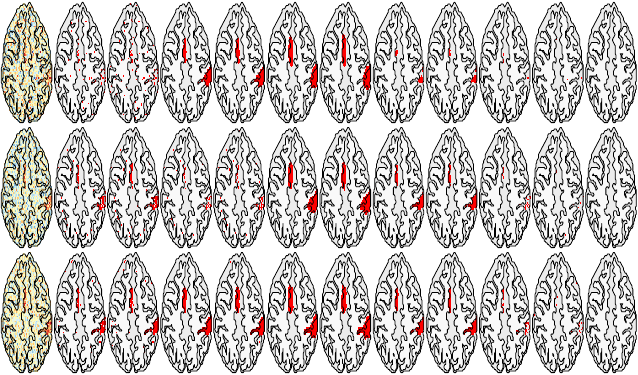}}}
    \mbox{
      \subfloat[AR(5)]{\includegraphics[height=0.25\textheight,width=\textwidth]{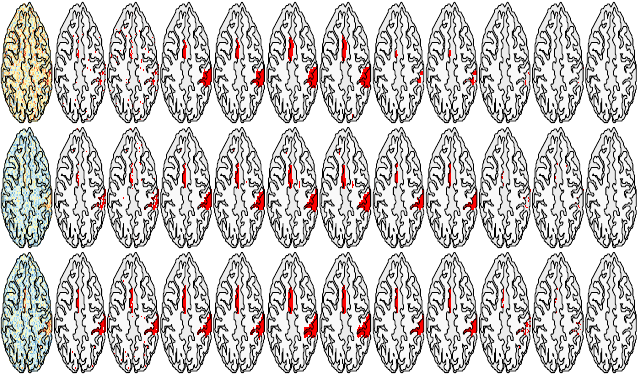}}}
    \caption{Sample SPMs and corresponding activation maps in a similar framework as
      Figure~\ref{fig:ActMapDec2}, but here the AR coefficients are in
      increasing order}.\label{fig:ActMapInc2}
    \end{figure}
\begin{figure}
\includegraphics[width=\textwidth]{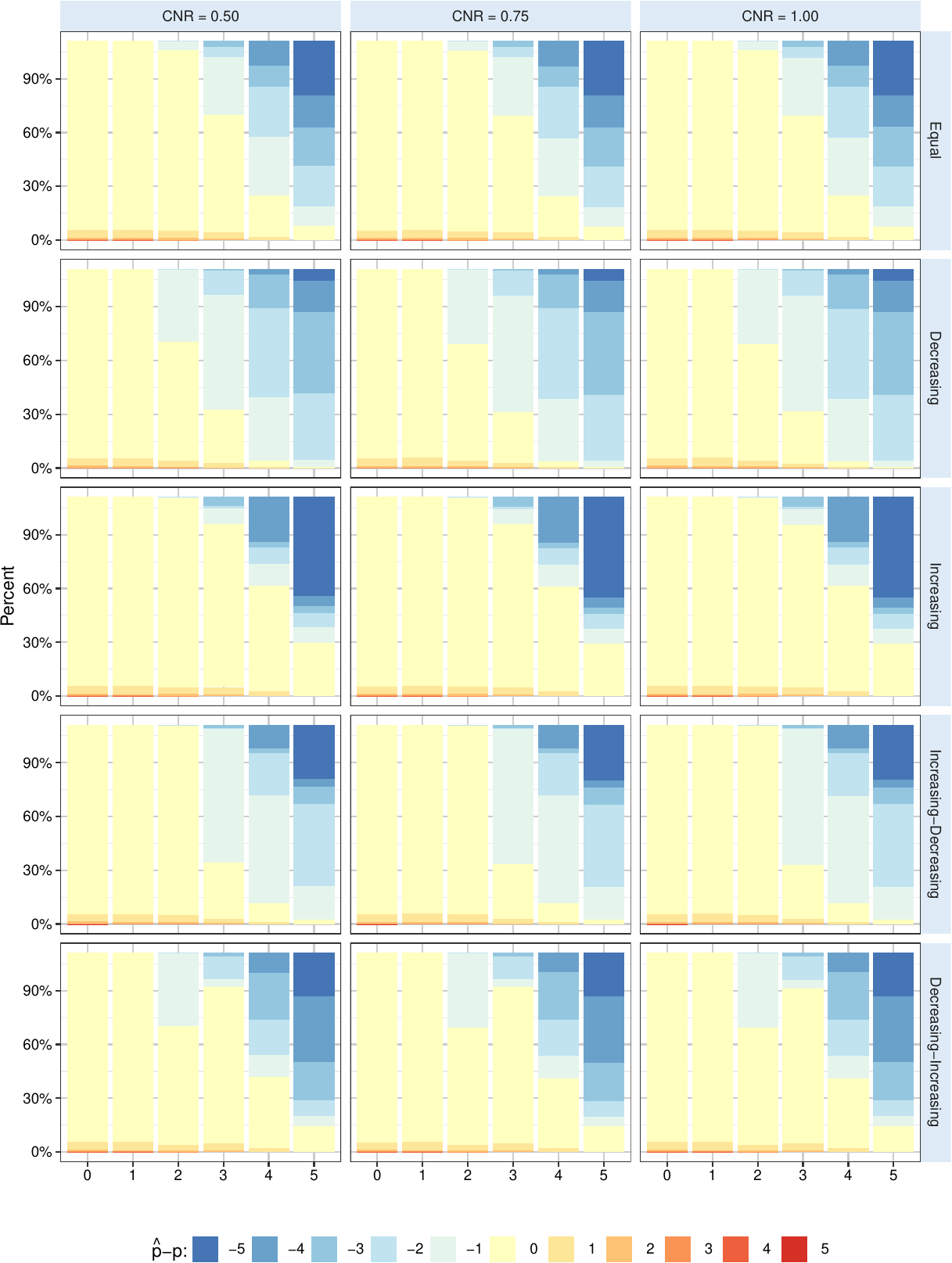}
\caption{Distribution of the difference in estimated and true orders
  for the large-scale simulation study of Section~\ref{hoff}.}
  \label{fig:ARp.estimated.order}
\end{figure}
\begin{figure}
\begin{center}
\includegraphics[width=0.95\textwidth]{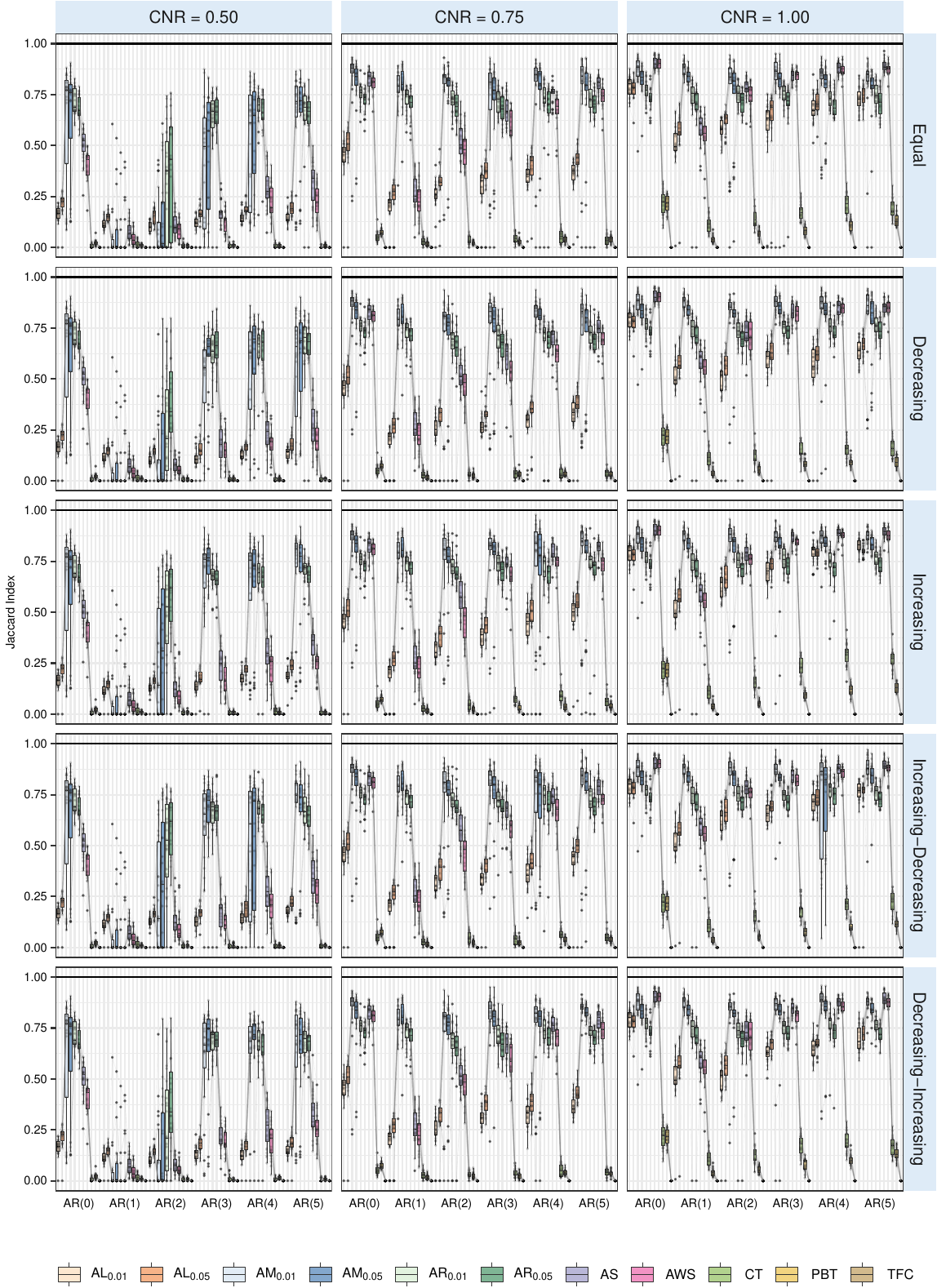}
\caption[JaccardIndex]{Performance, in terms of the Jaccard Index, for
  each of the algorithms for the different experimental setup and
  settings. As in Figure~\ref{fig:Jaccard.index.decreasing}, each
  method is displayed for each setting in the boxplot in the same
  order as in the legend.} 
  \label{fig:Jaccard.complete}
\end{center}
\end{figure}
\begin{figure}
\begin{center}
\includegraphics[width=0.95\textwidth]{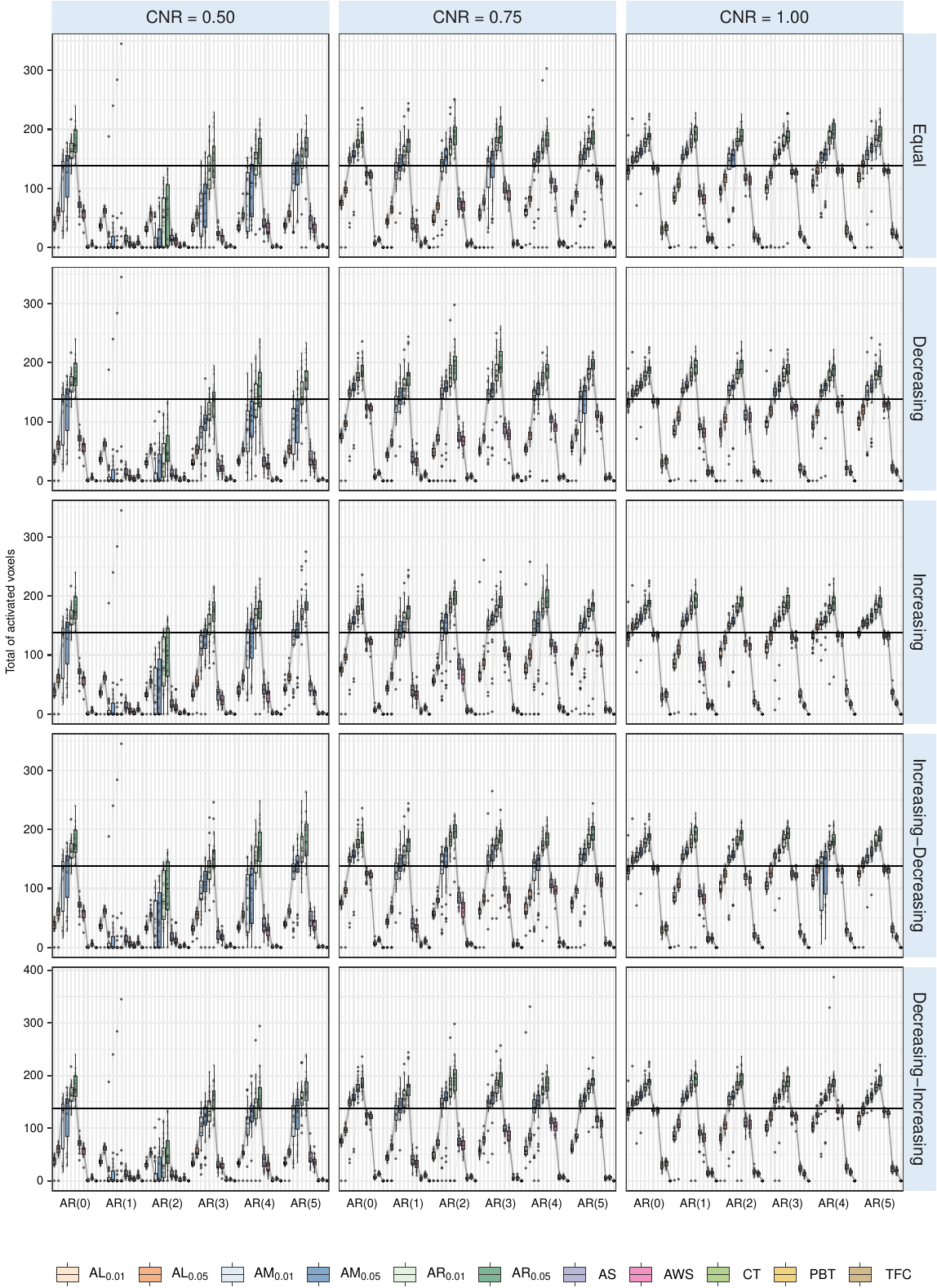}
\caption{Number of activated voxels in activation maps obtained
  with the modified Hoffman phantom and the ALL-FAST, AM-FAST, AR-FAST, AS, AWS,
  CT, PBT and TFCE algorithms 
  for the different simulation settings. The horizontal line shows
  the number of truly active voxels.}
  \label{fig:ActivatedVoxels.complete}
\end{center}
\end{figure}
\begin{figure}
\begin{center}
\includegraphics[width=0.95\textwidth]{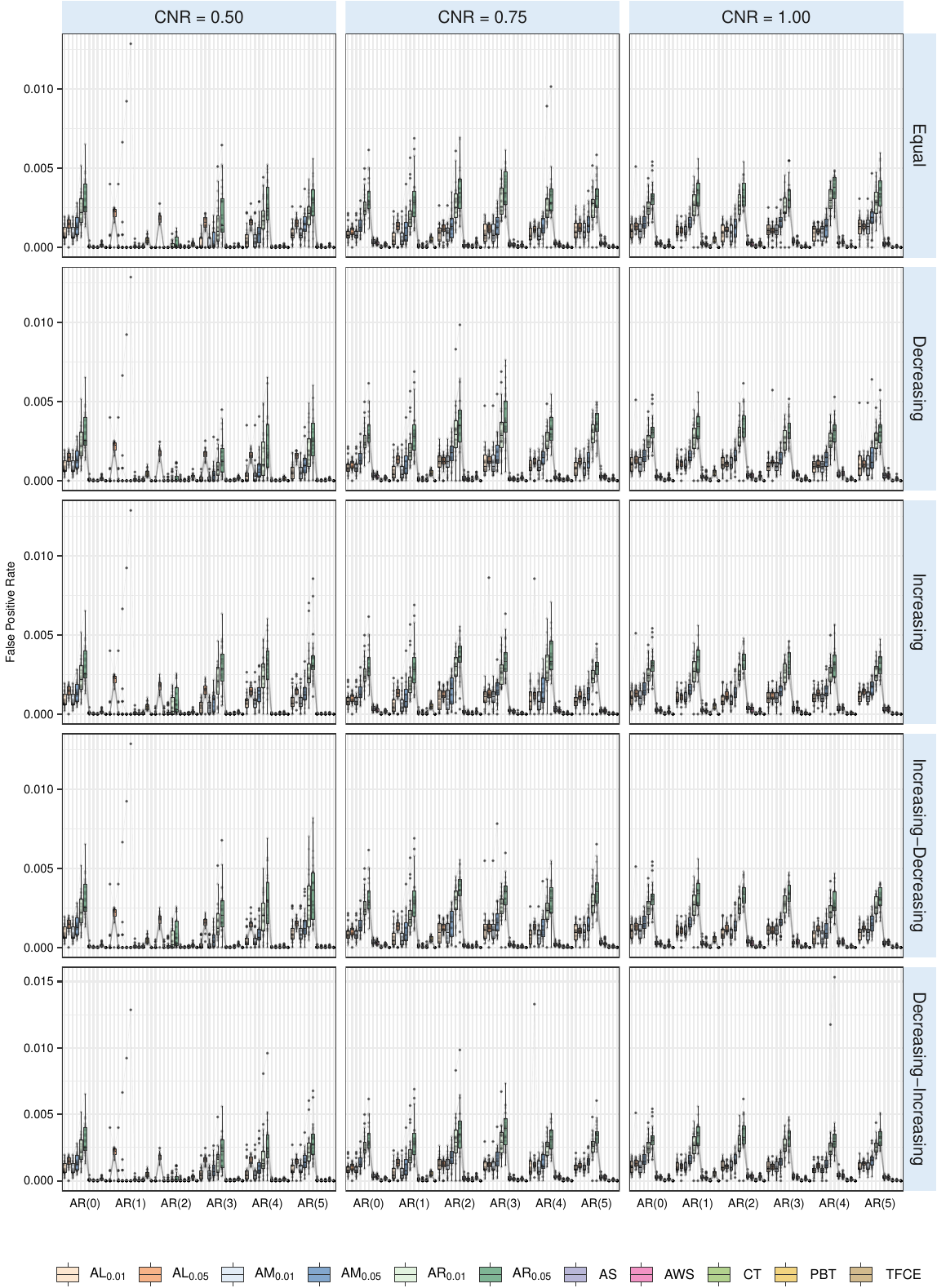}
\caption[FalsePositive]{False Positive Rate of the ALL-FAST, AM-FAST, AR-FAST,
  AS, AWS, CT, PBT and TFCE algorithms using the modified Hoffman
  phantom and for the different simulation settings.}
  \label{fig:FalsePositive.complete}
\end{center}
\end{figure}
\begin{figure*}[!h]
\begin{center}
\includegraphics[width=0.95\textwidth]{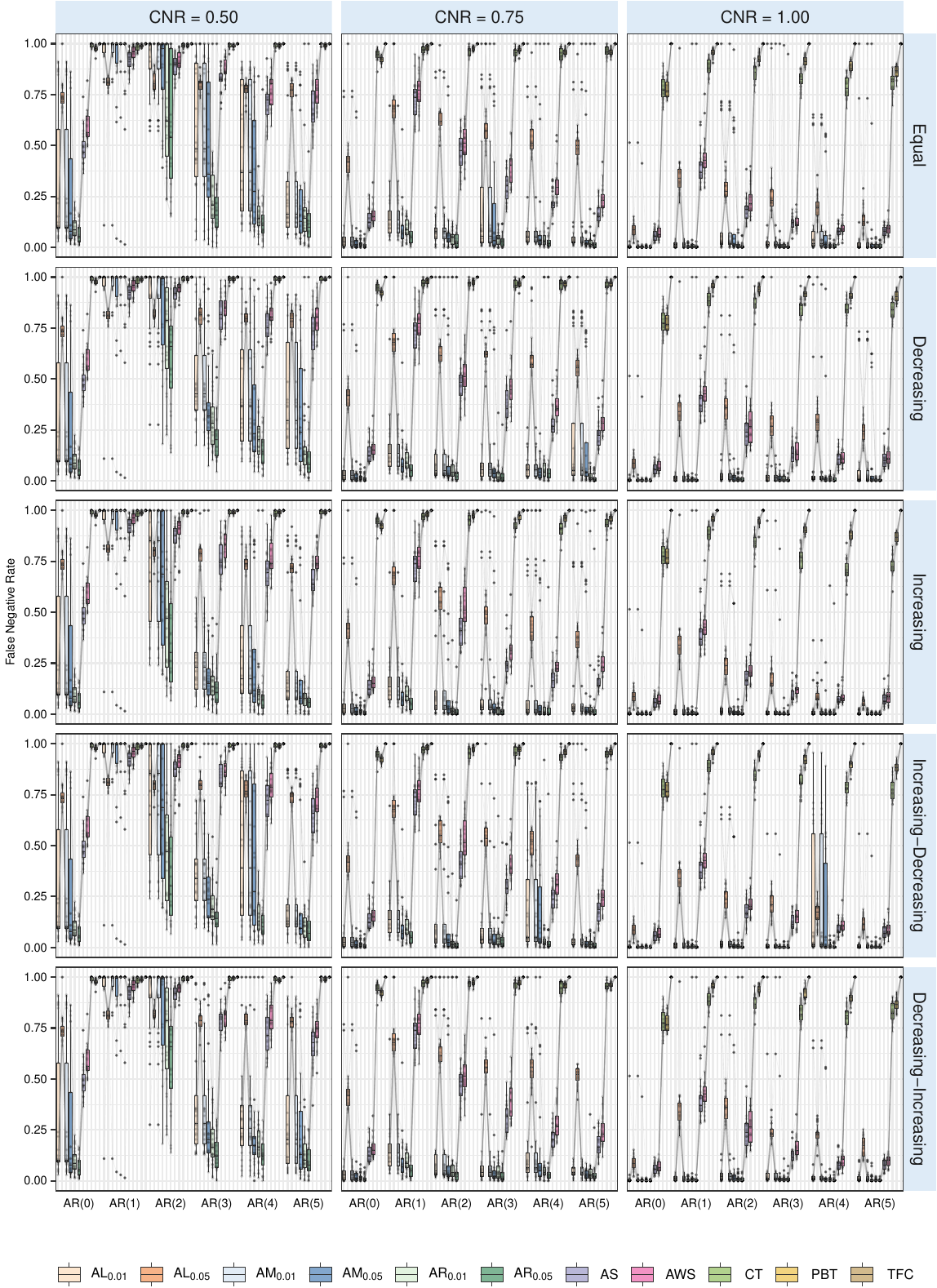}
\caption[FalseNegative]{False Negative Rate of the ALL-FAST, AM-FAST, AR-FAST, AS, AWS, CT, permutation and TFCE algorithms for the different simulation settings.}
  \label{fig:FalseNegative.complete}
\end{center}
\end{figure*}
\begin{figure}
\begin{center}
\includegraphics[width=0.95\textwidth]{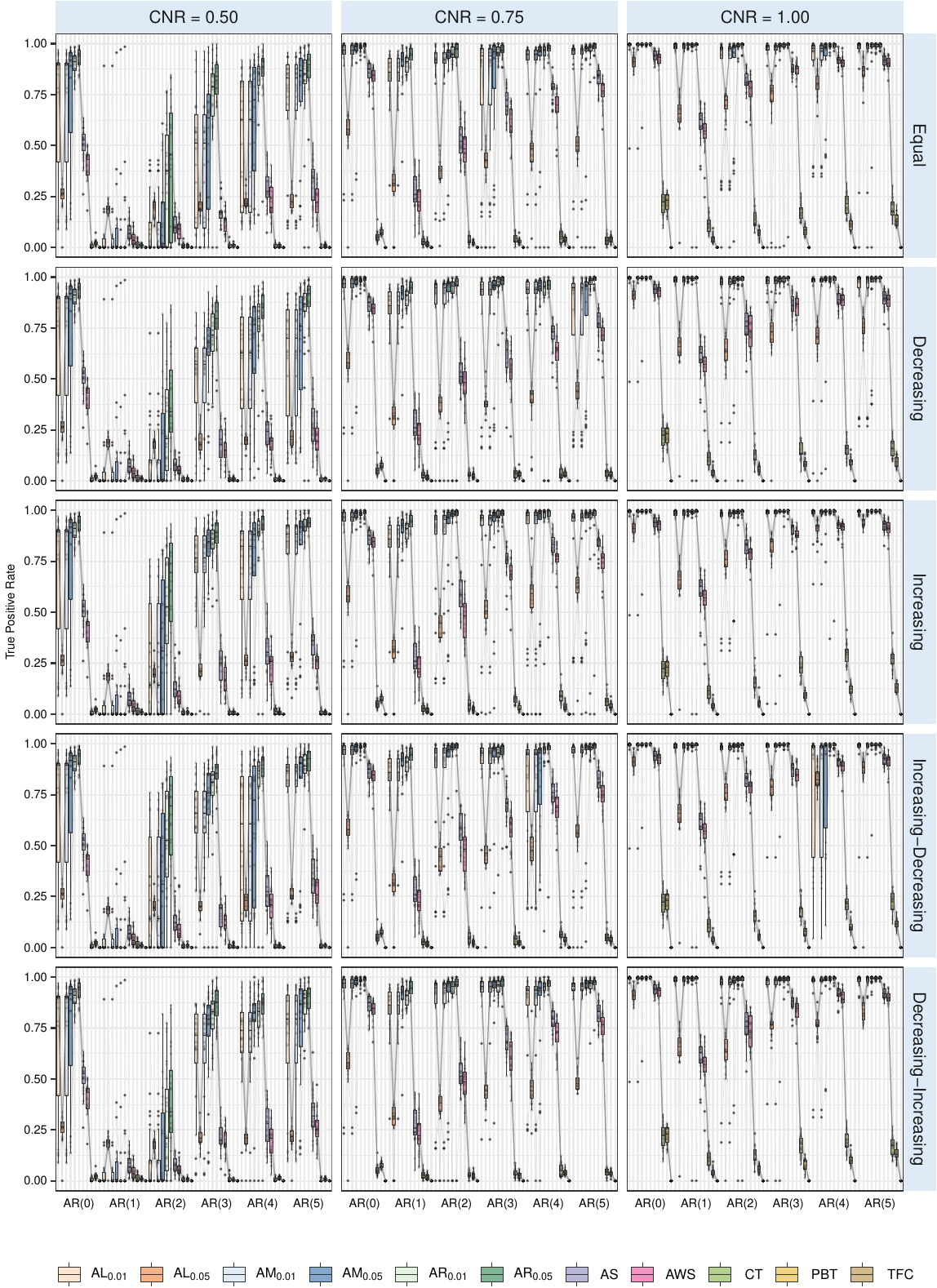}
\caption[TruePositive]{True Positive Rate of the ALL-FAST, AM-FAST, AR-FAST, AS,
  AWS, CT, PBT and TFCE algorithms for the different simulation
  settings using the modified Hoffman phantom.}
  \label{fig:TruePositive.complete}
\end{center}
\end{figure}
\begin{figure}
\begin{center}
\includegraphics[width=0.95\textwidth]{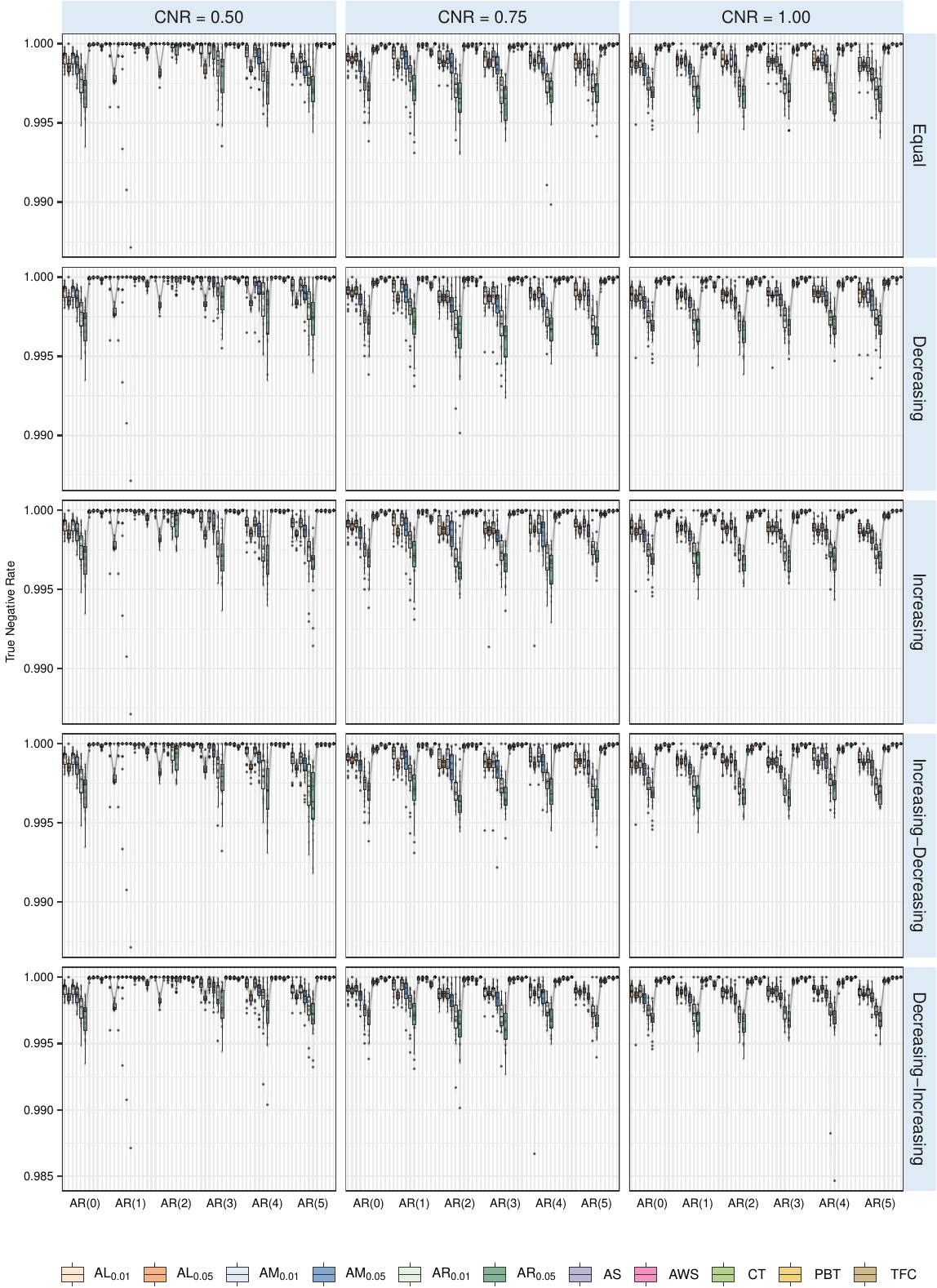}
\caption[TruePositive]{True Negative Rate of the ALL-FAST, AM-FAST, AR-FAST, AS,
  AWS, CT, PBT and TFCE algorithms for the different simulation
  settings using the modified Hoffman phantom.}
  \label{fig:TrueNegative.complete}
\end{center}
\end{figure}
\newpage
\section{Performance of FAST Algorithms on Simulated Null-Activated  SPMs}
  \label{null}
The second reviewer expressed continued concerns regarding multiple
significance. Our use of $\alpha$ is only to set a threshold:
its connections with testing in 
deciding on this threshold notwithstanding, our algorithm is not built
on, and does not use, the actual nominal significance level. Further, performance of FAST on the
resting state dataset of Section~\ref{resting} did not raise concerns
of identifying 
false positives. Nevertheless, we also performed a comprehensive
and thorough
2D simulation study to evaluate the issue of false positives. We generated SPMs
under the null hypothesis of no activation, and
used ALL-, AM- and AR-FAST for different $\alpha$ and evaluated the
incidence of even one pixel in an image (falsely) identified as
activated. Specifically, we 
generated SPMs $\bGamma\sim N(\bzero,\bR)$ where $\bR$ is a 2D
circulant matrix, with first order autocorrelation structure indexed
in each dimension by $\rho$. (For simplicity, we used a radial
correlation structure in the simulation, but not in the estimation and
activation detection.) We used $\rho\in \{0,0.01,
0.025, 0.05, 0.075, 0.1, 0.25, 0.5, 0.75, 0.99\}$ to provide a comprehensive
evaluation even though in our experience most SPMs from real datasets
have estimated $\rho$s substantially  lower than 0.1. Our SPMs were of
dimensions 
$128\times 128$. For each $\rho$, we simulated 1000 SPMs  to account for  
simulation variability in drawing our conclusions and then applied 
 both the FAST algorithms for a range of $\alpha\in\{0.001,0.01,0.025,
 0.05, 0.075,  0.1\}$. \begin{table}[h]
  \centering
\caption{Number of SPMs (out of 1000) simulated under null
  conditions where we found even one activated pixel using (a) 
  ALL-FAST, (b)  AM-FAST and (c) AR-FAST using different values of
  $\alpha$ and for  different values of $\rho$ (the assumed 
  radially symmetric first order autocorrelations between neigboring
  voxels in each simulated SPM).}
\label{null.table}
\mbox{
  \subfloat[][ALL-FAST]{\begin{tabular}{llcccccccccc}
  \hline
  &&&&&&$\rho$&&&&&\\
  & & 0 & 0.01 & 0.025 & 0.05 & 0.075 & 0.1 & 0.25 & 0.5 & 0.75 & 0.99\\  \hline
  $\alpha$ & 0.001 & 0 & 0 & 0 & 0 & 0 & 0 & 0 & 0 & 0 & 0 \\
           & 0.01 & 0 & 0 & 0 & 0 & 0 & 0 & 0 & 0 & 0 & 0\\
           & 0.025 & 0 & 1 & 0 & 0 & 0 & 0 & 0 & 0 & 0 &0\\
           & 0.05  & 0 & 2 & 2 & 1 & 0& 0 & 0 & 0 & 1 &0\\
           & 0.075 & 1 & 4 & 4 & 1 & 0& 0 & 0 & 0 & 2&2\\
           & 0.1   & 3 & 5 & 6 & 1 & 0& 0 & 0 & 0 & 2 &3\\
 \hline
                     \end{tabular}
                   }
                 }
\mbox{
  \subfloat[][AM-FAST]{\begin{tabular}{llcccccccccc}
  \hline
  &&&&&&$\rho$&&&&&\\
  & & 0 & 0.01 & 0.025 & 0.05 & 0.075 & 0.1 & 0.25 & 0.5 & 0.75 & 0.99\\  \hline
  $\alpha$ & 0.001 & 0 & 0 & 0 & 0 & 0 & 0 & 0 & 0 & 0 & 0 \\
           & 0.01 & 0 & 0 & 0 & 0 & 0 & 0 & 0 & 0 & 0 & 0\\
           & 0.025 & 0 & 0 & 0 & 0 & 0 & 0 & 0 & 0 & 0 & 0\\
           & 0.05  & 0 & 1 & 0 & 0 & 0& 0 & 0 & 0 & 0 & 0\\
           & 0.075 & 0 & 1 & 0 & 0 & 0& 0 & 0 & 0 & 0 &0\\
           & 0.1   & 0 & 1 & 0 & 0 & 0& 0 & 0 & 0 & 0 &0\\
 \hline
                     \end{tabular}
                   }
                 }
                 \mbox{
\subfloat[][AR-FAST]{\begin{tabular}{llcccccccccc}
  \hline
  &&&&&&$\rho$&&&&&\\
                       & & 0 & 0.01 & 0.025 & 0.05 & 0.075 & 0.1 &
                                                                   0.25 & 0.5 & 0.75 & 0.99\\  \hline
              $\alpha$ & 0.001 & 0 & 0 & 0 & 0 & 0 & 0 & 0 & 0 & 0 &7
                       \\
                       & 0.01 & 0 & 0 & 0 & 0 & 0 & 0 & 0 & 0& 0&14 \\
                       & 0.05 & 0 & 0 & 0 & 0 & 0& 0 & 0 & 0 & 0&15\\
                       & 0.05 & 0 & 0 & 0 & 0 & 0& 0 & 0 & 0 & 0&20\\
                       & 0.075 & 0 & 0 & 0 & 0 & 0& 0 & 0 & 0 & 0&22\\
           & 0.1   & 0 & 0 & 0 & 0 & 0& 0 & 0 & 0 & 0 &22\\
  \hline
                     \end{tabular}
                   }              }   
\end{table}
Table~\ref{null.table} represents the number  
of cases (out of 1000) where ALL-FAST, AM-FAST or AR-FAST for different
$\rho$ and $\alpha$  
found (falsely) even one activated voxel. For $\rho \leq 0.75$,
AR-FAST had no false positives. while for AM- and AR-FAST, there
were (out of several thousands) only a handful of cases where even one
pixel was falsely identified as activated and, in the case of AM-FAST,
these all happened at values of $\alpha >0.05$.
(Indeed, even with ALL-FAST, for $\alpha\leq 0.075$, these cases
with even one false positive identification mostly had 1-2 pixels in
that image falsely identified as activated.) 
 Thus, as with the case of resting state data of Section~\ref{resting}, and for all pratical
 fMRI settings, FAST performs very well
 in scenarios of no true activation. 
\begin{figure*}[t]
\mbox{
\subfloat[ALL-FAST]{\includegraphics[width=0.25\textwidth]{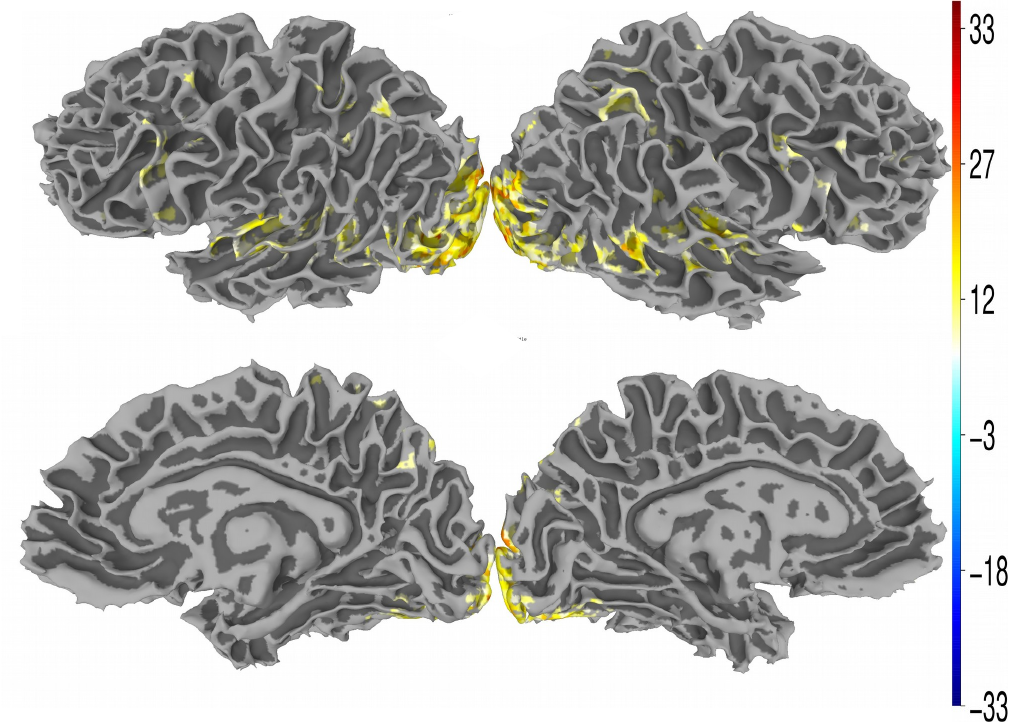}}
\subfloat[AS]{\includegraphics[width=0.25\textwidth]{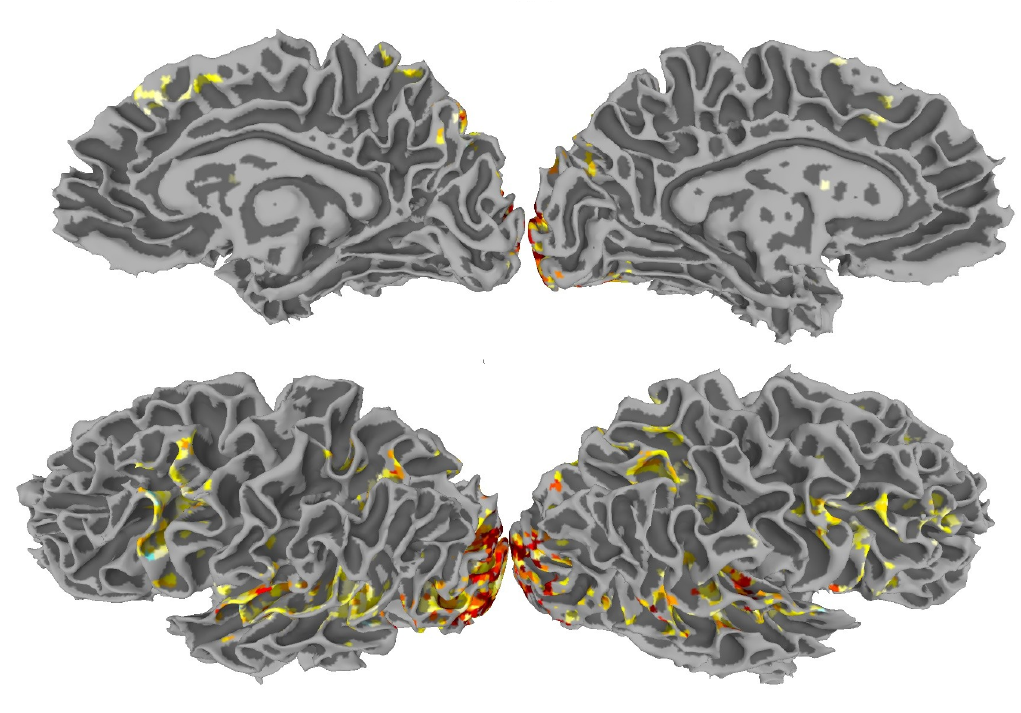}}
\subfloat[AWS]{\includegraphics[width=0.25\textwidth]{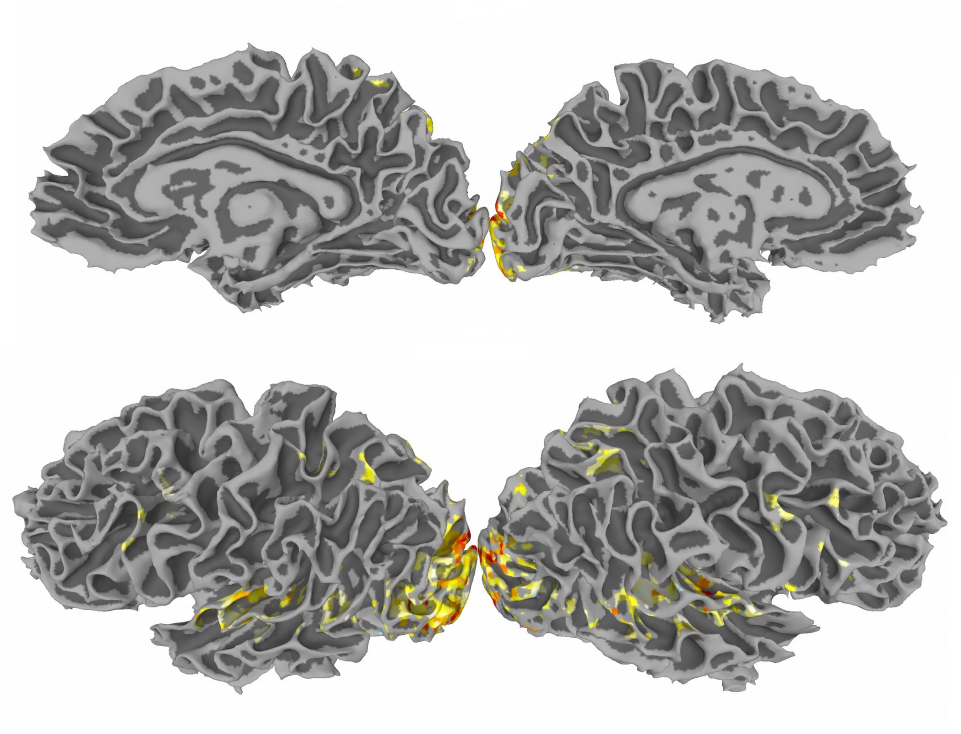}}
\subfloat[CT]{\includegraphics[width=0.27\textwidth]{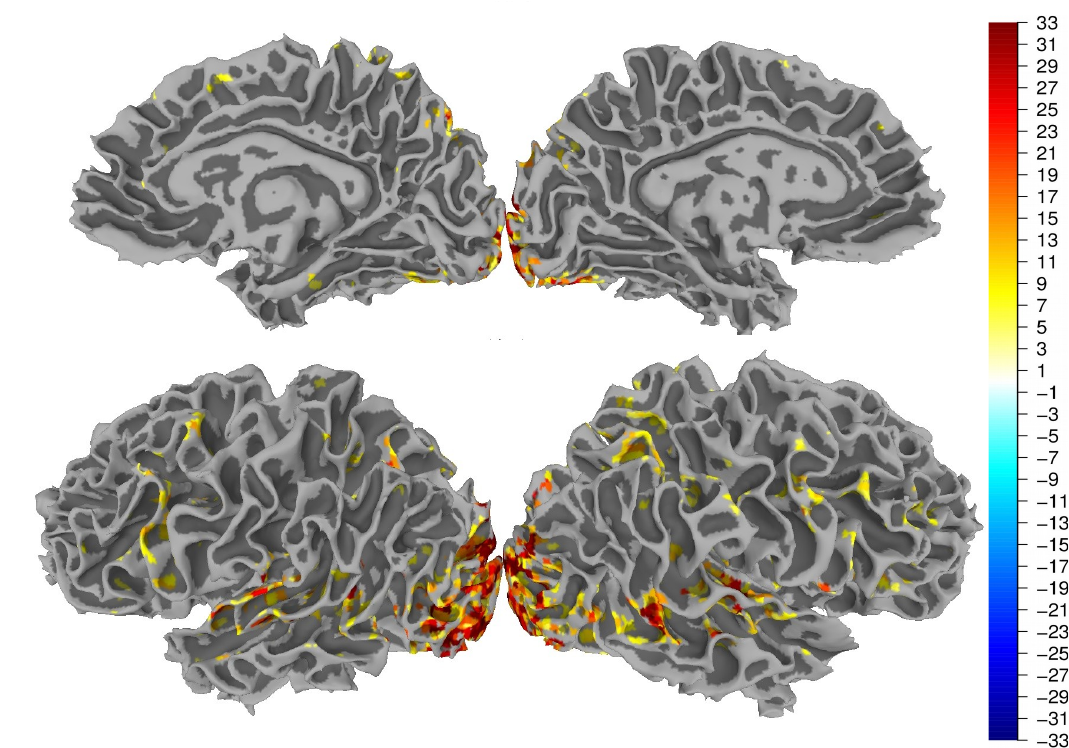}}
}
\mbox{
\subfloat[ALL-FAST]{\includegraphics[width=0.25\textwidth]{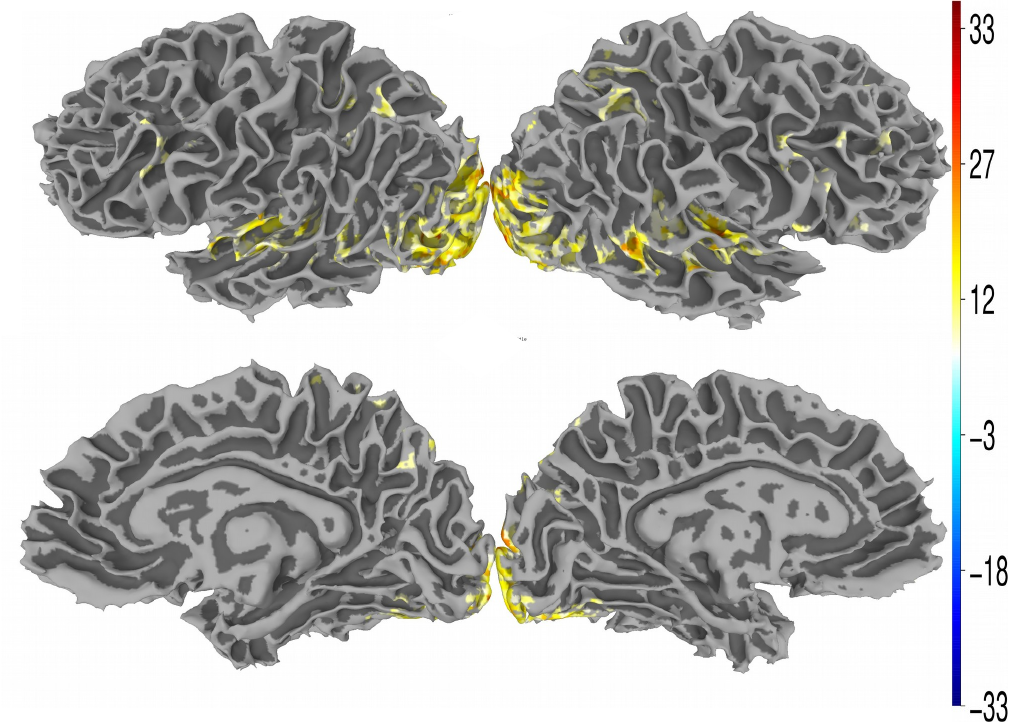}}
\subfloat[AS]{\includegraphics[width=0.25\textwidth]{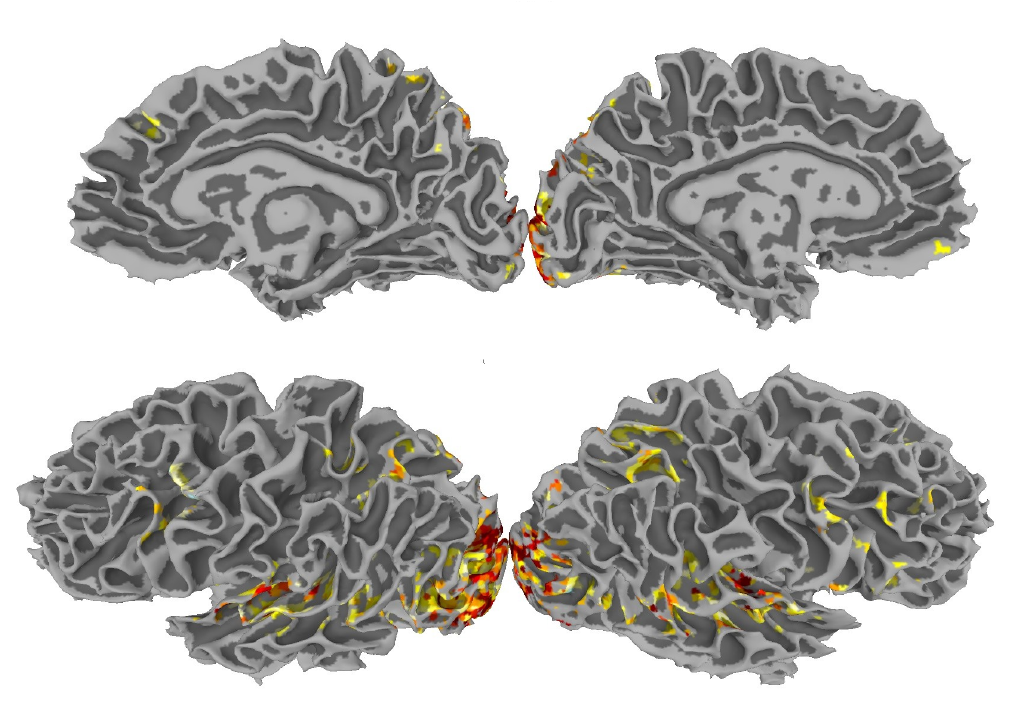}}
\subfloat[AWS]{\includegraphics[width=0.25\textwidth]{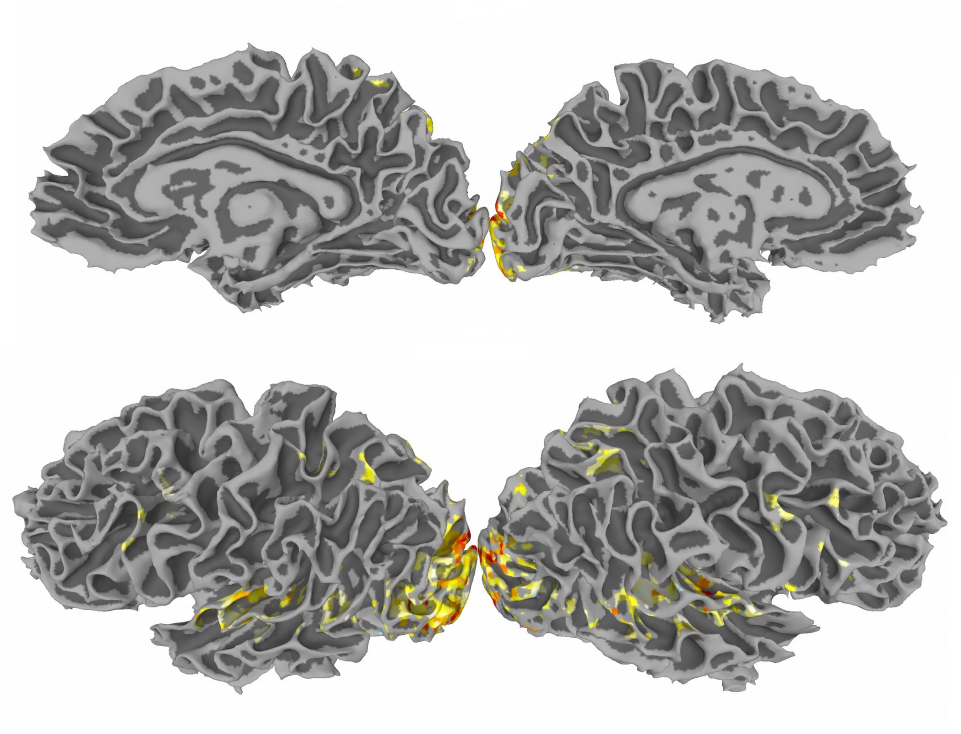}}
\subfloat[CT]{\includegraphics[width=0.27\textwidth]{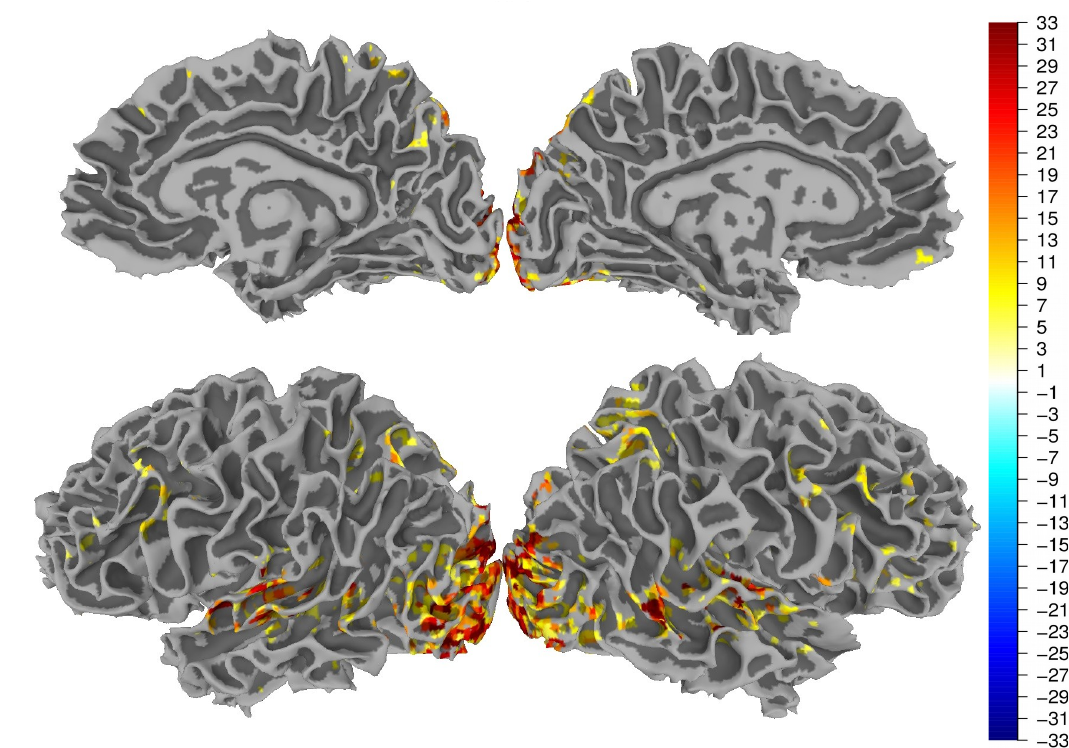}}
}
\caption{Activation regions in {\it AFNI data6} obtained using
  ALL-FAST, AS, AWS, and CT on SPMs upon fitting 
  AR($\hat{p}$) on~\ref{eq:lm}. The first row is for the
  visual-reliable stimulus, the  second row is for audio-reliable
  stimulus.}
\label{fig:Visual-Audio}
\end{figure*}

\end{document}